\documentclass[10pt,conference,letterpaper,twocolumn]{IEEEtran}

\usepackage{fullpage}
\usepackage{verbatim} 
\usepackage{float}
\usepackage{amsthm}
\usepackage{amsfonts}
\usepackage{graphicx}
\usepackage{setspace} 

\newtheorem{theorem}{Theorem}
\newtheorem{definition}[theorem]{Definition}
\newtheorem{observation}[theorem]{Observation}
\newtheorem{lemma}[theorem]{Lemma}
\newtheorem{invariant}[theorem]{Invariant}
\newtheorem{corollary}[theorem]{Corollary}

\floatstyle{ruled}
\newfloat{algorithm}{h}{lop}
\floatname{algorithm}{}
\restylefloat{algorithm}

\newcommand{\tab}{\hspace*{0.5em}}
\newcommand{\twotab}{\hspace*{1.5em}}
\newcommand{\threetab}{\hspace*{2.5em}}
\newcommand{\fourtab}{\hspace*{3.5em}}
\newcommand{\fivetab}{\hspace*{4.5em}}

\title{\LARGE  Non-blocking Patricia Tries with Replace Operations}
\author{\IEEEauthorblockN{ Niloufar Shafiei}
\IEEEauthorblockA{
 York University, Department of Computer Science and Engineering}}

\vspace{-2mm}

\IEEEpubid{0000--0000/00\$00.00~\copyrigth~2013 IEEE}
\begin{document}
\maketitle

\begin{abstract}
This paper presents a non-blocking Patricia trie implementation for an asynchronous shared-memory system using Compare\&Swap.
The trie implements a linearizable set and supports three update operations: 
insert adds an element, delete removes an element and replace replaces one element by another.  
The replace operation is interesting because it changes two different locations of tree atomically. 
If all update operations modify different parts of the trie, they run completely concurrently. 
The implementation also supports a wait-free find operation, which only reads shared memory and never changes the data structure.
Empirically, we compare our algorithms to some existing set implementations.
\end{abstract}


\section{Introduction} \label{intro-sec}
A Patricia trie \cite{PT} is a tree that stores a set of keys, which are represented as strings.
The trie is structured so that the path from the root to a key is determined by the sequence of characters in the key.
So, the length of this path is at most the length of the key (and will often be shorter).
Thus, if key strings are short, the height of the trie remains small without requiring any complicated balancing.
The simplicity of the data structure makes it a good candidate for concurrent implementations.
Patricia tries are widely used in practice.  
They have applications in routing systems, data mining, machine learning, bioinformatics, etc. \cite{bio-app, mining-app2, routing-app, routing-app2, mining-app}. 
Allowing concurrent access is essential in some applications and 
can boost efficiency in multicore systems.

We present a new concurrent implementation of Patricia tries for binary strings using single-word {\it Compare\&Swap} (CAS).
The operations on the trie are linearizable, meaning they appear to take place atomically \cite{lin}.
They are also non-blocking (lock-free): 
{\it some} process completes its operation in a finite number of steps even if other processes fail.
{\it Wait-free} algorithms satisfy the stronger guarantee that {\it every} process completes its operation in a finite number of steps.

Our implementation supports wait-free find operations and provides non-blocking insertions and deletions.
We also provide a non-blocking replace operation that makes two changes to the trie atomically:
it deletes one key and inserts another.
If all update operations are occurring at disjoint parts of the trie, they do not interfere with one another.

A Patricia trie can be used to store a set of points in $\mathbb{R}^d$.
For example, a point in $\mathbb{R}^2$ whose coordinates are $(x, y)$ can be represented as key formed by interleaving the bits of $x$ and $y$.
(This yields a data structure very similar to a quadtree.)
Then, the replace operation can be used to move a point from one location to another atomically.
This operation has applications in Geographic Information System \cite{GIS}.
The replace operation would also be useful if the Patricia trie were adapted to implement a priority queue, so that one can change the priority of an element in the queue.

Search trees are another class of data structures that are commonly used to represent sets.
When keys are not uniformly distributed, balanced search trees generally outperform unbalanced ones.
The reverse is often true when keys are uniformly distributed due to the simplicity of unbalanced search trees.
Our empirical results show that the performance of our trie is consistently good in both scenarios. 
This is because our trie implementation is as simple as an unbalanced search tree but also keeps trees short. 
For simplicity, we rely on a garbage collector (such as the one provided in Java implementations) that deallocates objects when they are no longer accessible.

For our Patricia trie algorithms, we extend the scheme used in \cite{BST} for binary search trees to coordinate processes.  Thus, we show that the scheme is more widely applicable.
In particular, we extend the scheme so that it can handle update operations that make more than one change to the tree structure.
Updates to the same part of the tree help one another to guarantee the non-blocking property.
An update first creates a descriptor object that contains enough information about the update, so that other processes can complete the update by reading the descriptor object.  
As in \cite{BST}, before an update changes the tree, it flags a small number of nodes to avoid interference with other concurrent updates.
(A node is flagged if it has a pointer to a descriptor object, otherwise it is unflagged.) 
When the update is complete, the flags are removed from nodes that are still in the tree.
Searches do not need to check for flags and can therefore traverse the trie very efficiently simply by reading child pointers.
Searches in our Patricia trie are wait-free, unlike the searches in \cite{BST} because the length of a search path in a Patricia trie is bounded by the length of the key.

There are several novel features of this work.
In our implementation, we design one fairly simple routine that is called to perform the real work of 
all update operations. 
In contrast, insert and delete operations in \cite{BST} are handled by totally separate routines.
This makes our proof of correctness more modular than the proof of \cite{BST}.
Our techniques and correctness proof can be generalized to other tree-based data structures.

In \cite{BST}, modifications were only made at the bottom of the search tree.
Our new Patricia trie implementation also copes with modifications that can occur anywhere in the trie.
This requires proving that changes in the middle of the trie do not cause concurrent search operations passing through the modified nodes to go down the wrong branch.
Howley and Jones \cite{IntST} introduced changes in the middle of a search tree but only to keys stored in internal nodes, not the structure of the tree itself.

In \cite{BST}, atomic changes had to be done by changing a single pointer.
Our replace operation makes two changes to the trie atomically.
Both changes become visible at the first CAS operation on a child pointer.
This new scheme can be generalized to make several changes to the trie atomically by making all changes visible at a single linearization point.
Cederman and Tsigas \cite{DCAS-Move} proposed a non-blocking replace operation for a tree-based data structure, but they require double-CAS (that modifies two non-adjacent locations conditionally).

To summarize:
\begin{itemize}
\item We present a non-blocking linearizable Patricia trie.
\item We employ one routine to implement the real work of any update operation.
\item We present a non-blocking update operation that requires changes to two child pointers using single-word CAS. 
\item We provide a modular correctness proof that can be adapted for other data structures.
We give a sketch of the correctness proof here.
A more detailed proof is provided in Appendix.
\item We compare our implementation empirically to other existing concurrent data structures.
\end{itemize}

The remainder of the paper is organized as follows.
In Section \ref{relatedWorkSec}, we present related work.
We describe the structure of the algorithms and how the operations are implemented in detail in Section \ref{AlgSec}.
Section \ref{proofSec} provides a sketch of the correctness proof. 
In Section \ref{evaluationSec}, we compare our trie to other algorithms empirically.
Some concluding remarks are given in Section \ref{conclusionSec}.

\section{Related Work} \label{relatedWorkSec}

Most concurrent data structures are lock-based.
However, lock-based implementations have  drawbacks such as priority inversion, deadlock and convoying.
Two  state of the art examples of lock-based implementations of set data structures are 
the AVL tree by Bronson et al.~\cite{AVL}, which maintains an approximately balanced tree,
and the self-adjusting binary search tree by Afek et~al.\ \cite{CBTree}, which moves frequently accessed nodes closer to the root.
Aref and Ilyas \cite{SP-GiST} described how lock-based implementations could be designed for a class of space-partitioning trees that includes Patricia tries.
Lock-coupling can also be applied to implement a concurrent Patricia trie \cite{lock-PT}. 

In this paper, we focus on non-blocking algorithms, which do not use locks.
There are two general techniques for obtaining non-blocking data structures:
universal constructions (see the related work section of \cite{WFConstruction} for a recent survey of work on this) and transactional memory \cite{STM} (see \cite{STM-survey} for a survey). 
Such general techniques are usually not as efficient as algorithms that are designed for specific data structures.

Tsay and Li \cite{WF-tree} gave a general wait-free construction for tree-based data structures.
To access a node, a process makes a local copy of the path from the root to the node, performs computations on the local copy, and then atomically replaces the entire path by its local copy.
Since this approach copies many nodes and causes high contention at the root, their approach is not very efficient.
Barnes \cite{lock-free} presented another general technique to obtain non-blocking implementations of data structures in which processes cooperate to complete operations.

Ellen et al.~\cite{BST} presented the first non-blocking binary search tree data structure from CAS operations. 
Their approach has some similarity to the cooperative technique of \cite{lock-free}.
As discussed in Section \ref{intro-sec}, our Patricia trie implementation extends the approach used in \cite{BST}.
Brown and Helga \cite{k-ary} generalized the binary search trees of \cite{BST} to non-blocking $k$-ary search trees and 
compared the non-blocking search trees with the lock-based search tree of Bronson et al.~\cite{AVL} empirically on a multicore system. 

Howley and Jones \cite{IntST} presented a non-blocking search tree from CAS operations using a cooperative technique similar to \cite{BST}.
Their tree store keys in both leaf and internal nodes.
However, search operations sometimes help update operations by performing CASs. 

Braginsky and Petrank proposed a non-blocking balanced B+tree from CAS operations \cite{Btree}.
The implementation uses the marking technique of \cite{Harris-LL} and the flagging technique of \cite{BST}.

Earlier this year, Prokopec et al.~\cite{Ctrie} described a non-blocking hash trie that uses CAS operations. 
Their approach is very different from our implementation.
Unlike Patricia tries, in their trie implementation, an internal node might have single child.
In their implementation, nodes have up to $2^k$ children (where $k$ is a parameter) and 
extra intermediate nodes are inserted between the actual nodes of the trie. 
With $k=5$, the height of their trie is very small, making their implementation very fast 
when contention is low.  However, our experiments suggest that it is not
very scalable under high contention.
Unlike our implementation, their search operation may perform CAS steps. 

Non-blocking implementations of set data structures have also been proposed based on skip lists using CAS operations \cite{Fomitchev-SL, Fraser-SL, Sundell-SL}.
A non-blocking skip list (ConcurrentSkipListMap) was then implemented in the Java class library by Doug Lea. 

\section{Algorithm Description} \label{AlgSec}
We assume an asynchronous shared-memory system with single-word CAS operations.
We first give the sequential specification of the operations. 
The trie stores a set $D$ of keys from a finite universe $U$.
If $v \notin D$, insert($v$) changes $D$ to  $D \cup \{v\}$ and returns true; otherwise,  it returns false.
If $v \in D$, delete($v$) changes $D$ to $D - \{v\}$ and returns true; otherwise, it returns false.
If $v \in D$ and $v' \notin D$, replace($v, v'$) changes $D$ to $D - \{v\} \cup \{v'\}$ and returns true; otherwise, it returns false.
If $v \in D$, find($v$) returns true; otherwise, it returns false.
In either case, find($v$) does not change $D$.
We assume elements of $D$ can be encoded as $\ell$-bit binary strings. 
(In Section \ref{conclusionSec}, we describe how to handle unbounded length keys.)

\subsection{Data Structures} 
First, we describe the structure of a binary Patricia trie.
(See Figure \ref{trie-fig}.)
Each internal node has exactly two children.
The elements of $D$ are stored in the leaves of the trie.
Each internal node  stores a binary string that is the longest common prefix of its children.
If a node's label has length $k-1$, then
the $k$th bit of the node's left and right child is 0 and 1, respectively.
The root stores the empty string.
The height of the trie  is at most $\ell$. 

\begin{figure}[t]
\centering
\leavevmode 
\includegraphics[width=3.2cm]{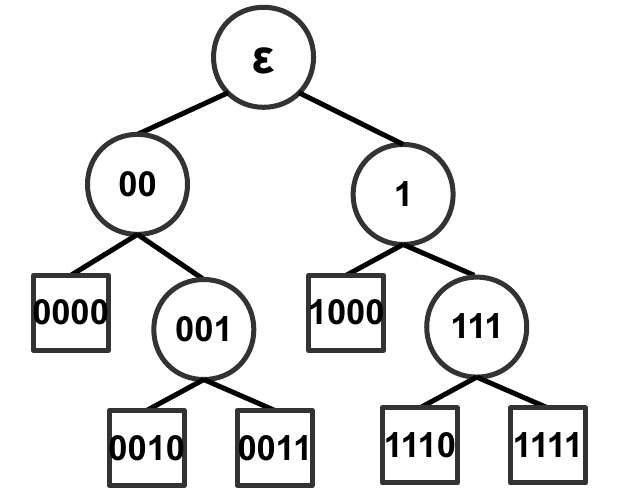}
\vspace*{-4mm}
\caption{An example of a Patricia trie. (Leaves are represented by squares and internal nodes are represented by circles.)}
\label{trie-fig}
\vspace*{-7mm}
\end{figure}

Next, we describe the objects that are used in the implementation (Figure \ref{data-fig}).
The Patricia trie is represented using Leaf and Internal objects which are subtypes of Node objects.
A Node object has a $label$ field representing its binary string, which is never changed after initialization.
An Internal object has an array of Node objects of size two, denoted $child$, that stores pointers to the children of the node. 

Each Node object also has an {\it info} field that stores a pointer to an Info object that represents an update operation that is in progress at the node.  The Info object contains enough information to allow other processes to help the update to complete.
The Info object has two subtypes: Flag and Unflag.
An Unflag object is used to indicate that no update is in progress at a node.
Unflag objects are used instead of null pointers to avoid the ABA problem in the $info$ field of a node.
Initially, the {\it info} field of each Node object is an Unflag object.
We say that a node is {\it flagged} or {\it unflagged}, depending on whether its {\it info} field stores a Flag or Unflag object.
The {\it info} and $child$ field of an internal node are changed using CAS steps.
However, a leaf node gets flagged by writing a Flag object into its {\it info} field.

To perform an update operation, first some internal nodes get flagged, then some $child$ fields are changed and then nodes that are still in the trie get unflagged.
The nodes that must be flagged to perform an update operation are the internal nodes whose $child$ field will be changed by the update or that will be removed from the trie by the update.
Flagging nodes is similar to locking nodes: it avoids having other operations change the part of the trie that would be changed by the update. 

A Flag object has a number of fields.
The $flag$ field stores nodes to be flagged and the $unflag$ field stores nodes to be unflagged.
Before creating a Flag object, an update reads the $info$ field of each node that will be affected by the update before reading that node's $child$ field.  
This value of the $info$ field is stored in the Flag's $oldInfo$ field, and is used for the CAS that flags the node.
This ensures that if the node is successfully flagged, it has not changed since its children were read.
Moreover, once it is flagged, its children will not be changed by any other update operation.
The boolean $flagDone$ field indicates whether the flagging for the update has been completed.
In the case of a replace operation, the  $rmvLeaf$ field points to the leaf to be removed by the update after flagging is complete.
The actual changes to the trie to be made are described in three more array fields of the Flag object:
$pNode$, $oldChild$ and $newChild$.
For each $i$, the update should CAS the appropriate $child$ pointer of $pNode[i]$ from $oldChild[i]$ to $newChild[i]$.
If all nodes are successfully flagged, then the CAS on each $child$ pointer will be guaranteed to succeed because that pointer cannot have changed since the old value was read from it.
Thus, like locks, the $info$ field of a node is used to give an operation exclusive permission to change the $child$ field of that node.

For simplicity, the root node of the trie is initially set to an Internal object whose children are two leaf nodes whose $label$s are the strings $0^\ell$ and $1^\ell$. 
We assume the keys $0^\ell$ and $1^\ell$ cannot be elements of $D$. 
This ensures that the trie always has at least two leaf nodes and the root node never needs to be replaced. 
(This avoids some special cases that would occur when the root is a leaf.)

\begin{figure}[t]
\begin{enumerate}
\newcounter{c}
{\footnotesize
\item \tab Leaf: (subtype of Node)
\item  \twotab $label \in U$
\item \twotab $info:$ Info

\item \tab Internal: (subtype of Node)
\item \twotab $label \in U$
\item \twotab $child:$ Node[2]\hfill $\rhd$ left and right child
\item \twotab $info:$ Info 

\item \tab Flag: (subtype of Info)
\item \twotab $flag:$ Internal[4]
\item \twotab $oldInfo:$ Info[4] 
\item \twotab $unflag:$ Internal[2]
\item \twotab $pNode:$ Internal[2] 
\item \twotab $oldChild:$ Node[2]
\item \twotab $newChild:$ Node[2]
\item \twotab $rmvLeaf:$ Leaf
\item \twotab $flagDone:$ Boolean 

\item \tab Unflag: (subtype of Info) \hfill $\rhd$ has no field \\ \label{obj-end}

\item \tab{\bf Initialization:}
\item  \twotab $root \leftarrow$ new Internal($\varepsilon$, [new Leaf(000....0, new unflag), \\ \threetab new Leaf(111...1, new unflag)], new Unflag) \label{initial-root}

\vspace*{-1mm}

\vspace*{-1mm}
\setcounter{c}{\value{enumi}}
}
\end{enumerate}
\vspace*{-2mm}

\caption{\label{data-fig}
Data types used in the implementation}
\vspace*{-7mm}

\end{figure}

\subsection{Update Operations} 
The implementation has three update operations: insert, delete and replace.
All three have the same overall structure.
The pseudo-code for our implementation is given on page 5.
An update $op$ uses the search routine to find the location(s) in the trie to be changed.
It then creates a new Flag object $I$ containing all the information required to complete the update by calling newFlag.
If newFlag sees that some node that must be flagged is already flagged with a different Flag $I'$, it calls help($I'$) at line \ref{newFlag-call-help} to try completing the update described by $I'$, 
and then $op$ retries its update from scratch.
Otherwise, $op$ calls help($I$) to try to complete its own update.


As mentioned earlier, flagging nodes ensures exclusive access for changing $child$ pointers.
Thus, an update flags the nodes whose $child$ pointers it wishes to change and permanently flags any node that is removed from the trie to avoid applying updates to  a deleted portion of the trie

Unlike locks, the Info objects store enough information, so that if an operation dies while nodes are flagged for it, other processes can complete the operation and remove the flags.
This ensures that a failed operation cannot prevent others from progressing.
To avoid deadlock, if an update must flag more than one internal node, we order the internal nodes by their $label$s.

The help($I$) routine carries out the real work of an update using the information stored in the Flag object $I$.  It first uses {\it flag CAS} steps to flag some nodes (line \ref{help-flag-node}) by setting their $info$ fields to $I$.
If all nodes are flagged successfully, help($I$) uses {\it child CAS} steps to change the $child$ fields of some internal nodes to perform the update (line \ref{help-change-child}).
Then, it uses {\it unflag CAS} steps to unflag nodes that were flagged earlier, except the ones that have been removed from the trie (line \ref{help-unflag}) by setting their $info$ fields to a new Unflag object.
In this case, any nodes deleted by the update remain flagged forever. 
If any node is not flagged successfully, the attempt to perform the update has failed and {\it backtrack CAS} steps are used to unflag any nodes that were flagged earlier (line \ref{help-unflag-bt}).

\begin{figure}
\label{op-fig}
\begin{enumerate}
\setcounter{enumi}{\value{c}}
{\footnotesize

\item {\bf insert}($v \in U$)
\item \tab while(true)
\item \twotab $I \leftarrow$ null
\item \twotab $\langle -, p, node, -, pInfo, rmvd\rangle \leftarrow$ {\bf search}($v$) \label{ins-call-search}
\item \twotab if {\bf keyInTrie}($node$, $v$, $rmvd$) then return false \label{ins-return-false}
\item \twotab $nodeInfo = node.info$ \label{ins-read-node-info}
\item \twotab $nodeCopy \leftarrow$ new copy of $node$ \label{create-copy}
\item \twotab $newNode \leftarrow $ {\bf createNode}($copy$, new Leaf containing $v$, \\ \threetab $nodeInfo$)\label{ins-call-create-node}
\item \twotab if $newNode \ne$ null then
\item \threetab if $node$ is Internal then
\item \fourtab $I \leftarrow$ {\bf newFlag}$([p, node]$, $[pInfo, nodeInfo]$, $[p]$, \\ \fivetab $[p]$, $[node]$,  $[newNode]$, null)\label{ins-set-info1}
\item \threetab else $I \leftarrow$ {\bf newFlag}($[p], [pInfo]$, $[p]$, $[p], \\ \fourtab [node], [newNode]$,  null)\label{ins-set-info2}
\item \threetab if $I \ne$ null and {\bf help}($I$) then return true \label{ins-call-help} \\
 \vspace*{-1mm}

\item {\bf delete}($v \in U$)
\item \tab while(true)
\item \twotab $I \leftarrow$ null
\item \twotab $\langle gp,p,node,gpInfo,pInfo, rmvd\rangle \leftarrow$ {\bf search}($v$) \label{del-call-search}
\item \twotab if $\neg$ {\bf keyInTrie}($node$, $v$, $rmvd$) then return false \label{del-return-false}
\item \twotab $nodeSibling \leftarrow p.child[1 - (|p.label|+1)$th bit of $v]$\label{del-read-sibling}
\item \twotab if $gp \ne$ null then
\item \threetab $I \leftarrow$ {\bf newFlag}($[gp, p]$, $[gpInfo, pInfo]$, $[gp]$, $[gp]$, $[p]$, \\ \fourtab $[nodeSibling]$, null)\label{del-set-info}
\item \threetab if $I \ne$ null and {\bf help}($I$) then return true \\ \label{del-call-help}
\vspace*{-1mm}

\item {\bf replace}($v_d \in U$, $v_i \in U$)
\item \tab while(true)
\item \twotab $I \leftarrow$ null
\item \twotab \makebox{\shortstack[l]{$\langle gp_d,p_d,node_d,gpInfo_d,pInfo_d, rmvd_d \rangle \leftarrow$ {\bf search}($v_d$)}}\label{mov-call-search1}
\item \twotab if $\neg$ {\bf keyInTrie}($node_d$, $v_d$, $rmvd_d$) then return false \label{mov-return-false1}
\item \twotab $\langle -, p_i, node_i, -, pInfo_i, -, rmvd_i \rangle \leftarrow$ {\bf search}($v_i$) \label{mov-call-search2}
\item \twotab if {\bf keyInTrie}($node_i$, $v_i$, $rmvd_i$) then return false \label{mov-return-false2}
\item \twotab $nodeInfo_i = node_i.info$ \label{mov-read-node-info}
\item \twotab \makebox{\shortstack[l]{$nodeSibling_d \leftarrow p_d.child[1 - (|p_d.label|+1)$th bit of $v_d]$}}\label{mov-read-sibling}

\item \twotab if $gp_d \ne$ null and $node_i \notin$ \{ $node_d, p_d, gp_d$ \} \\ \threetab and $p_i \ne p_d$ then\label{mov-check-gc-condition}
\item \threetab $copy_i \rightarrow$ new copy of $node_i$ \label{create-copy2} \label{mov-create-copy}
\item \threetab $newNode_i \leftarrow$ {\bf createNode}($copy_i$, new Leaf containing \\ \fourtab $v_i$,  $nodeInfo_i$)\label{mov-call-create-node}

\item \threetab if $newNode_i \ne$ null and $node_i$ is Internal then
\item \fourtab $I \leftarrow$ {\bf newFlag}($[gp_d, p_d, p_i, node_i]$, $[gpInfo_d,$ \\ \fivetab $pInfo_d, pInfo_i$,  $nodeInfo_i]$, $[gp_d, p_i]$, $[p_i$, \\ \fivetab $gp_d]$, $[node_i, p_d]$, 
 $[newNode_i$,  $nodeSibling_d]$, \\ \fivetab$node_d$)\label{mov-set-info1} 

\item \threetab else if $newNode_i \ne$ null and $node_i$ is Leaf then
\item \fourtab $I \leftarrow$ {\bf newFlag}($[gp_d, p_d, p_i]$, $[gpInfo_d, pInfo_d,$ \\ \fivetab $pInfo_i]$, $[gp_d$, $p_i]$, $[p_i$, $gp_d]$, $[node_i, p_d]$, 
\\ \fivetab $[newNode_i,$  $nodeSibling_d]$, $node_d$)\label{mov-set-info2} 

\item \twotab else if $node_i = node_d$ then \label{mov-sc3}
\item \threetab $I \leftarrow$ {\bf newFlag}($[p_d], [pInfo_d]$, $[p_d], [p_d]$, $[node_i]$, \\ \fourtab [new Leaf  containing $v_i$], null)\label{mov-set-info-sc1}  

\item \twotab else if ($node_i = p_d$ and $p_i = gp_d$) or \label{mov-condition-sc2} \label{mov-sc4}
\item \threetab ($gp_d \ne$ null and $p_i = p_d$) then  \label{mov-condition-sc3}
\item \threetab $newNode_i \leftarrow$ {\bf createNode}($nodeSibling_d$, new Leaf \\ \fourtab containing $v_i$,  $nodeSibling_d.info$) \label{mov-create-node1}
\item \threetab if $newNode_i \ne$ null then
\item \fourtab $I \leftarrow$ {\bf newFlag}($[gp_d, p_d]$, $[gpInfo_d, pInfo_d]$, $[gp_d]$,  \\ \fivetab $[gp_d]$,  $[p_d]$,  $[newNode_i]$, null) \label{mov-set-info-sc23}

\item \twotab else if $node_i = gp_d$ then \label{mov-sc5}
\item \threetab \makebox{\shortstack[l]{$pSibling_d \leftarrow gp_d.child[1 - (|gp_d.label|+1)$th bit of $v_d]$}}\label{mov-read-pSibling}
\item \threetab \makebox{\shortstack[l]{$newChild_i \leftarrow$ {\bf createNode}($nodeSibling_d$, $pSibling_d$, -)}}\label{mov-create-node2}
\item \threetab if $newChild_i \ne$ null then 
\item \fourtab $newNode_i \leftarrow$ {\bf createNode}($newChild_i$, new Leaf \\ \fivetab containing $v_i$, -)\label{mov-create-node3}
\item \fourtab if $newNode_i \ne$ null then $I \leftarrow$ {\bf newFlag}($[p_i, gp_d,$ \\ \fivetab $p_d]$, $[pInfo_i, gpInfo_d, pInfo_d]$, $[p_i]$, $[p_i]$, \\ \fivetab $[node_i]$, $[newNode_i]$, null) \label{mov-set-info-sc4}
\item \twotab if $I \ne$ null and {\bf help}($I$) then return true\label{mov-call-help}

\setcounter{c}{\value{enumi}}
}
\end{enumerate}
\vspace*{-2mm} 

\caption{Update operations}

\vspace*{-5mm} 

\end{figure}

\begin{figure}
\label{routine-fig}

\begin{enumerate}
\setcounter{enumi}{\value{c}}
{\footnotesize
\item {\bf find}($v \in U$)
\item  \tab $\langle$ -, -, $node$, -, -, $rmvd \rangle \leftarrow $ {\bf search}($v$) \label{find-call-search}
\item \tab \makebox{\shortstack[l]{if {\bf keyInTrie}($node$, $v$, $rmvd$)  then return true}}
\item \tab else return false\\
\vspace*{-1mm}

\item {\bf search}($v \in U$)
\item \tab $\langle p, pInfo \rangle$ $\leftarrow \langle$ null, null$\rangle$ \label{search-initial-p}
\item \tab $node \leftarrow root$ \label{search-set-initial-node}
\item \tab while ($node$ is Internal and $node.label$ is prefix of $v$)\label{search-check-condition}
\item \twotab $\langle gp, gpInfo \rangle$ $\leftarrow$ $\langle p, pInfo \rangle$ \label{search-set-gp}
\item \twotab $\langle p, pInfo \rangle$ $\leftarrow$ $\langle node, node.info \rangle$ \label{search-set-p}
\item \twotab $node \leftarrow p.child[(|p.label|+1)$th bit of $v]$\label{search-set-node}
\item \tab if $node$ is Leaf then \hfill $\rhd$ if Leaf is replaced \label{search-check-moved}
\item \twotab $rmvd \leftarrow$ {\bf logicallyRemoved}($node.info$) \label{search-set-moved}  
\item \tab return $\langle gp, p, node, gpInfo, pInfo, rmvd\rangle$ \\
 \vspace*{-1mm}

\item {\bf  help}($I$: Flag)
\item \tab $i \leftarrow 0$
\item \tab $doChildCAS \leftarrow$ true
\item \tab while ($i < |I.flag|$ and $doChildCAS$)
\item  \twotab CAS($I.flag[i].info$, $I.oldInfo[i]$, $I$) \hfill $\rhd$ flag CAS \label{help-flag-node}
\item  \twotab $doChildCAS \leftarrow (I.flag[i].info = I)$ \label{help-set-doChildCAS}
\item  \twotab $i \leftarrow i + 1$ 
\item \tab if $doChildCAS$ then \label{help-check-doChildCAS}
\item \twotab $I.flagDone \leftarrow$ true \label{help-set-done}
\item \twotab if $I.rmvLeaf \ne$ null then $I.rmvLeaf.info \leftarrow I$ \label{help-flag-leaf} 
\item \twotab for $i = 0$ to $(|I.pNode|-1)$ \label{help-start-change-child}
\item \threetab $k \leftarrow (|I.pNode[i].label|+1)$th bit of \\ \fourtab $I.newChild[i].label$ \label{help-read-index}
\item \threetab CAS($I.pNode[i].child[k]$, $I.oldChild[i]$, \\ \fourtab $I.newChild[i]$)\hfill$\rhd$ child CAS\label{help-change-child}
\item \tab if $I.flagDone$ then \label{help-check-done}
\item \twotab for $i = (|I.unflag| - 1)$ down to $0$
\item \threetab CAS($I.unflag[i].info$, $I$, new Unflag) \hfill $\rhd$ unflag CAS\label{help-unflag}
\item \twotab return true \label{help-return-true}
\item \tab else
\item \twotab for $i = (|I.flag| - 1)$ down to 0 \label{bt-start}
\item \threetab CAS($I.flag[i].info$, $I$, new Unflag)\hfill $\rhd$ backtrack CAS\label{help-unflag-bt}
\item \twotab return false \label{bt-end} \\ \label{help-return-false}
 \vspace*{-1mm}
 
\item {\bf newFlag}($flag$, $oldInfo$, $unflag$, $pNode$, $oldChild$, \\ \tab $newChild$,  $rmvLeaf$)
\item \tab for $i = 0$ to $(|oldInfo| - 1)$,
\item \twotab if $oldInfo[i]$ is Flag then \label{newFlag-old-info}
\item \threetab {\bf help}($oldInfo[i]$) \label{newFlag-call-help}
\item \threetab return null \label{newFlag-return-null}\label{newFlag-return-null1}
\item \tab \makebox{\shortstack[l]{if $flag$ has duplicates with different values in $oldInfo$ then}}\label{newFlag-dup-val}
\item \twotab return null\label{newFlag-return-null2}
\item \tab else remove duplicates in $flag$ and $unflag$ (and \\ \twotab corresponding  entries of $oldInfo$)\label{newFlag-keep-one-copy}
\item \tab sort elements of $flag$ and permute elements of $oldInfo$ \label{newFlag-sort-nodes} 
\item \tab return new Info($flag$, $oldInfo$, $unflag$, $pNode$, $oldChild$, \\ \twotab $newChild$,  $rmvLeaf$, false) \\ \label{newFlag-return}
 \vspace*{-1mm}

\item {\bf createNode}($node_1$: Node, $node_2$: Node, $info$: Info) 
\item \tab if $node_1.label$ is prefix of $node_2.label$ or \\ \twotab $node_2.label$ is prefix of $node_1.label$ then
\item \twotab if $info$ is Flag then {\bf help}($info$) \label{createNode-call-help1}
\item \twotab return null
\item \tab else return new Internal whose children are $node_1$ and $node_2$ \\ \label{create-internal}
 \vspace*{-1mm}

\item {\bf logicallyRemoved}($I$: Info) 
\item \tab if $I$ is Unflag then return false \label{logRemoved-return-false}
\item \tab return ($I.oldChild[0]$ not in $I.pNode[0].child$)\\ \label{logRemoved-check-old} \label{logRemoved-check-oldChild}
 \vspace*{-1mm}

\item {\bf keyInTrie}($node$: Node, $v \in U$, $rmvd$: Boolean)
\item \tab \makebox{\shortstack[l]{return ($node$ is Leaf and $node.label = v$ and $rmvd =$ false)}}\\ \label{keyInTrie-return}

}
\end{enumerate}
\vspace*{-4mm} 

\caption{The find operation and additional subroutines}
\vspace*{-5mm} 

\end{figure}




If any child CAS step is executed inside help($I$), the update is successful and it is linearized at the first such child CAS.
If a replace operation performs two different child CAS steps, it first executes a child CAS to insert the new key, and then a child CAS to delete the old key.
In this case, the replace also flags the leaf node of the old key before the first child CAS step.
We say the leaf is {\it logically removed} from the trie at the first child CAS step and any operation that reaches the leaf node after this determines that the key is already removed.
We say a node is {\it reachable} at time $T$ if there is path from the root to the node at $T$.
We say a leaf node is {\it logically in the trie} at time $T$ if the node is reachable and not logically removed at $T$.
We shall prove that the following invariant holds:
The leaf nodes that are logically in the trie at time $T$ contain exactly those keys in the set $D$, according to the sequence of updates that are linearized before $T$. 

Whenever a child pointer is changed, the old child is permanently flagged and it is removed from the trie to avoid the ABA problem.
(In some cases, this requires the update to add a new copy of the old child to the trie.)
When a call to help($I$) performs a child CAS on $I.pNode[i]$ (for some $i$), it uses $I.oldChild[i]$ as the old value.
Since there is no ABA problem, only the first such CAS on $I.pNode[i]$ can succeed.
Moreover, we prove that the flagging mechanism ensures that this first CAS does succeed.
Since processes might call help($I$) to help each other to complete their operations, there might be a group of child CASs on each node.
However, the $child$ pointer is changed exactly once for the operation.

\subsection{Detailed Description of Algorithms} \label{algSec}
A search($v$) is used by updates and find to locate key $v$ within the trie.
The search($v$) starts from the root node and traverses down the trie. 
At each step of the traversal, search($v$) chooses the child according to the appropriate bit of $v$ (line \ref{search-set-node}). 
The search($v$) stops if it reaches an internal node whose $label$ is not a prefix of $v$.
We show that any node visited by the search was reachable at some time during the search.
If the search($v$) does not return a leaf containing $v$, there was a time during the search when no leaf containing $v$ was reachable.
Moreover, the node that is returned is the location where an insert would have to put $v$.
If search($v$) reaches a leaf node and the leaf node is logically removed by a replace operation, search($v$) sets $rmvd$ to true (line \ref{search-set-moved}).

As we shall see, update operations must change the $child$ pointers of the parent or grandparent of the node returned by search.
The search operation returns $gp$, $p$ and $node$, the last three nodes reached (where $p$ stands for parent and $gp$ stands for grandparent).
A search also returns the values $gpInfo$ and $pInfo$ that it read from the $info$ fields of $gp$ and $p$ before reading their $child$ pointers.
More formally, if search($v$) returns $\langle gp$, $p$, $node$, $gpInfo$, $pInfo$, $rmvd \rangle$, it satisfies the following post-conditions.
(1) At some time during search($v$), $gp.info$ was $gpInfo$ (if $gp$ is not null).
(2) Then, at some later time during search($v$), $p$ was a child of $gp$ (if $gp$ is not null).
(3) Then, at some later time during search($v$), $p.info$ was $pInfo$.
(4) Then, at some later time during search($v$), $p.child[i]$ was $node$ for some $i$.
(5) $(p.label) \cdot i$ is a prefix of $v$.
(6) If $node$ is an internal node, $node.label$ is not a prefix of $v$.
(7) If $rmvd$ is true, $node$ is logically removed at some time during search($v$). 
(8) If $rmvd$ is false, $node$ is logically in the trie at some time during search($v$). 

After calling search, an update uses the newFlag routine to create a Flag object.
For each node that the update must flag, a value read from the {\it info} field during search of the node is passed to newFlag as the old value to be used in the flag CAS step. 
The old value for a flag CAS was read before the old value for the corresponding child CAS, so if the flag CAS succeeds, then the node's $child$ field has not been changed since the last time its old value was read.
The newFlag routine checks if all old values for {\it info} fields are Unflag objects (line \ref{newFlag-old-info}).
If some {\it info} field is not an Unflag object, then there is some other incomplete update operating on that node.
The newFlag routine tries to complete the incomplete update (line \ref{newFlag-call-help}), and then returns null, which causes the update to restart.
In some cases of the replace operation that change the trie in two steps, $gp$ returned by the first search might be equal to $p$ returned by the second search.
So, the operation might send duplicate elements to get flagged to the newFlag routine. 
If the duplicate elements do not have the same old values, their $child$ fields might have changed since the operation read them, so newFlag returns null and the operation starts over (line \ref{newFlag-dup-val}-\ref{newFlag-return-null2}).
Otherwise, only one copy of each duplicate element is kept (line \ref{newFlag-keep-one-copy}).
The newFlag routine sorts the nodes to be flagged (to avoid deadlocks) and returns the new Flag object (line \ref{newFlag-sort-nodes}-\ref{newFlag-return}).

After an update $u$ creates a Flag object $I$, it calls help($I$).
This routine attempts to complete the update.
First, it uses CAS steps to put the Flag object $I$ in the {\it info} field of the nodes to be flagged (line \ref{help-flag-node}).
If all nodes are flagged successfully, the $flagDone$ field of the Flag object is set to true (line \ref{help-set-done}).
The value of the $flagDone$ field is used to coordinate processes that help the update. 
Suppose a process $p$ is executing help($I$). 
After $p$ performs a flag CAS on a node $x$, if it sees a value different from $I$ in the $x$'s {\it info} field, there are two possible cases.
The first case is when all nodes were already successfully flagged for $I$ by other processes running help($I$), and then $x$ was unflagged before $p$ tries to flag $x$.
(Prior to this unflagging, some process performed the child CAS steps of $I$ successfully.)
The second case is when no process flags $x$ successfully for $I$. 
Since the $flagDone$ field of $I$ is only set to true after all nodes are flagged successfully, $p$ checks the value of the $flagDone$ field to determine which case happened.
If $flagDone$ is true, the modifications to the trie for update $u$ have been made. 
If $flagDone$ is false, the update operation cannot be successfully completed, so all internal nodes that got flagged earlier are unflagged by the back-tracking CAS steps at line \ref{bt-start}-\ref{bt-end} and 
the update $u$ will have to start over.

After flagging all nodes successfully and setting $I.flagDone$, if $I.rmvLeaf$ is a leaf, its {\it info} field is set to $I$ (line \ref{help-flag-leaf}).
Only the two-step replace operations flag a leaf.
Then, help($I$) changes the $child$ fields of nodes in $I.pNode$ using child CASs (line \ref{help-start-change-child}-\ref{help-change-child}).
Finally, help($I$) uses unflag CASs to unflag the nodes in $I.unflag$ and returns true (line \ref{help-check-done}-\ref{help-return-true}).

\begin{figure}[t]
\centering
\leavevmode
\includegraphics[width=8.0cm]{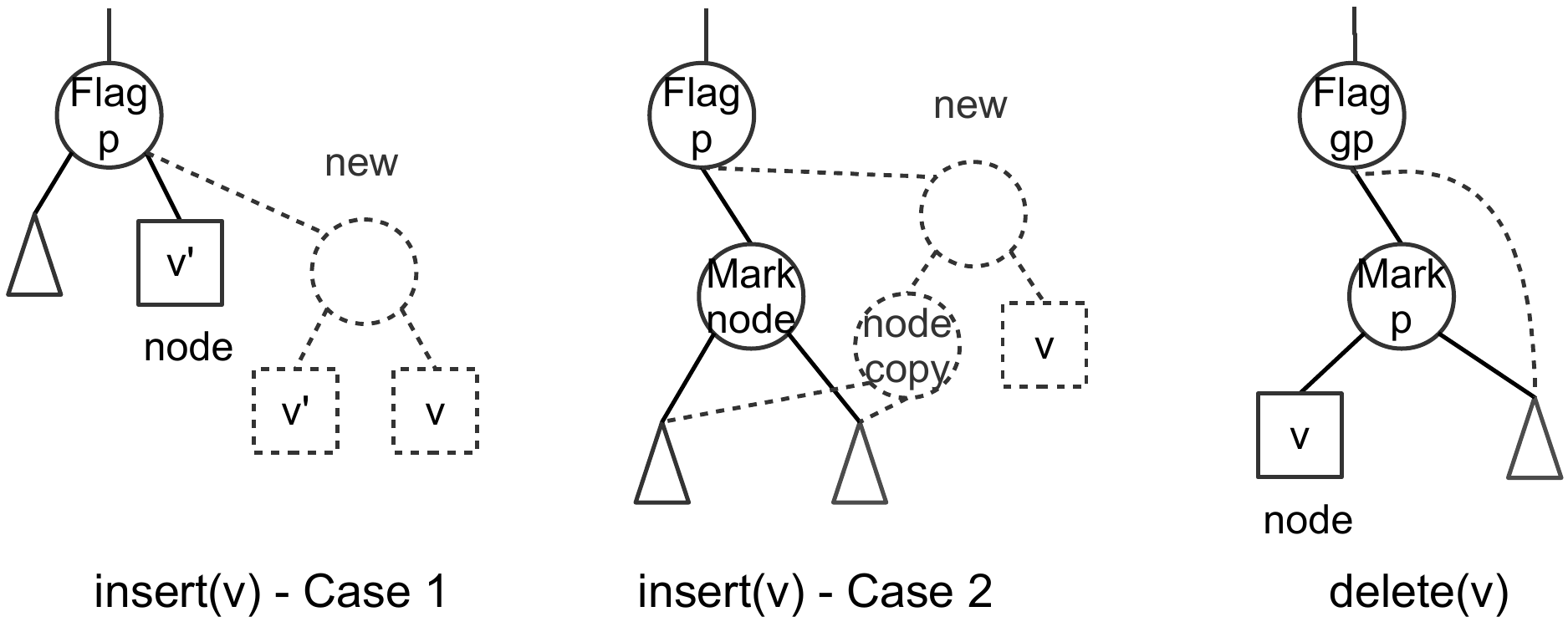}
 \vspace*{-8mm}
\caption{Different cases of insert($v$) and delete($v$).  
(Triangles are either a leaf node or a subtree.
The dotted lines are the new child pointers that replace the old child pointers (solid lines) and the dotted circles are newly created nodes.)}
\label{ins-del-fig}
\vspace*{-7mm}
\end{figure}

The insert($v$) routine first calls search($v$).
Let $\langle$ -, $p$, $node$, -, -, $rmvd\rangle$ be the result returned by search($v$). 
If $node$ is a leaf containing $v$ and $rmvd$ is false, insert($v$) returns false since the trie already contains $v$ (line \ref{ins-return-false}).
Otherwise, the insertion attempts to replace $node$ with a node created line \ref{create-internal}, whose children are a new leaf node containing $v$ and a new copy of $node$.
(See Figure \ref{ins-del-fig}.)
Thus, the parent $p$ of $node$ must be flagged.
A new copy of $node$ is used to avoid the ABA problem.
If $node$ is an internal node, since $node$ is replaced by a new copy, insert($v$) must flag $node$ permanently (line \ref{ins-set-info1}).

The delete ($v$) routine first calls search($v$).
Let $\langle gp$, $p$, $node$, -, -, $rmvd \rangle$ be the result returned by the search($v$). 
If $node$ is not a leaf node containing $v$ or $rmvd$ is true, delete($v$) returns false since the trie does not contain $v$ (line \ref{del-return-false}).
Then, delete($v$) replaces $p$ by the sibling of $node$.
(See Figure \ref{ins-del-fig}.)
So, delete($v$) must flag the grandparent $gp$ of $node$ and mark $p$ (line \ref{del-set-info}). 

A replace works as follows.
First, replace($v_d$, $v_i$) calls search($v_d$) and search($v_i$), which return $\langle gp_d$, $p_d$, $node_d$, -, -, $rmvd_d\rangle$ and $\langle$-, $p_i$, $node_i$,  -, -, $rmvd_i\rangle$.
The replace checks that $v_d$ is in the trie and $v_i$ is not, as in the insert and delete operations (line \ref{mov-call-search1}-\ref{mov-return-false2}).
If either test fails, the replace returns false.

If insert($v_i$) and delete($v_d$), as described in Figure \ref{ins-del-fig}, would not overlap, replace($v_d$, $v_i$) is done by two child CAS steps and is linearized at the first of these two changes.
This is called the general case of replace. 
Situations when the insertion and deletion would occur in overlapping portions of the trie are handled as special cases as shown in Figure \ref{move-fig}.
In the special cases, the replace changes the trie with one child CAS.

In the general case of the replace operation  (line \ref{mov-check-gc-condition}-\ref{mov-set-info2}), we create a Flag object which instructs the help routine to perform the following actions.
The replace flags the same nodes that an insert($v_i$) and a delete($v_d$) would flag.
After flagging these nodes, the leaf $node_d$ also gets flagged.
Then, $v_i$ is added to the trie, as in insert($v_i$).
When the new leaf node is added, the leaf $node_d$, which contains $v_d$, becomes logically removed, but not physically removed yet.
Then, $node_d$ is physically deleted as in delete($v_d$).
After $node_d$ is flagged, any search that reaches $node_d$ checks if $p_i$ is a parent of the old child of $p_i$ using $node_d.info$.
If it is not, it means the new leaf containing $v_i$ is already inserted and the operation behaves as if $v_d$ is already removed. 

There are four special cases of replace($v_d$, $v_i$) where the changes required by the insertion and deletion are on the overlapping portions of the trie and the replace operation is done using one child CAS step.
Although the code for these cases looks somewhat complicated, it simply implement the actions described in Figure \ref{move-fig} by creating a Flag object and calling help.
The insertion of $v_i$ replaces $node_i$ by a new node.
The cases when the deletion must remove $node_i$ or change $node_i.child$ are handled as special cases.
So, the case that $node_d = node_i$ is one special case (line \ref{mov-sc3}-\ref{mov-set-info-sc1}).
In the deletion, $p_d$ is removed, so the case that $p_d = node_i$ or $p_d = p_i$ are also handled as a special case (line \ref{mov-sc4}-\ref{mov-set-info-sc23}).
In the deletion, $gp.child$ is changed.
So, the last special case is when $gp_d = node_i$ (line \ref{mov-sc5}-\ref{mov-set-info-sc4}).
In all special cases, $node_i$ is replaced by a new node.
Here, we explain one special case in detail.
The others are handled in a similar way.
In case 2, $p_d = node_i$ and $gp_d = p_i$ (line \ref{mov-sc4}).
Since $node_i$ must be replaced, a new copy of $node_i$ is made.
So, replace($v_d$, $v_i$) flags $gp_d$ and marks $node_i$, and replaces $node_i$ with a new internal node whose non-empty children are a new leaf node containing $v_i$ and the sibling of $node_d$
(line \ref{mov-sc4}-\ref{mov-set-info-sc23}).

\begin{figure}[t]
\centering
\leavevmode
\includegraphics[width=8.30cm]{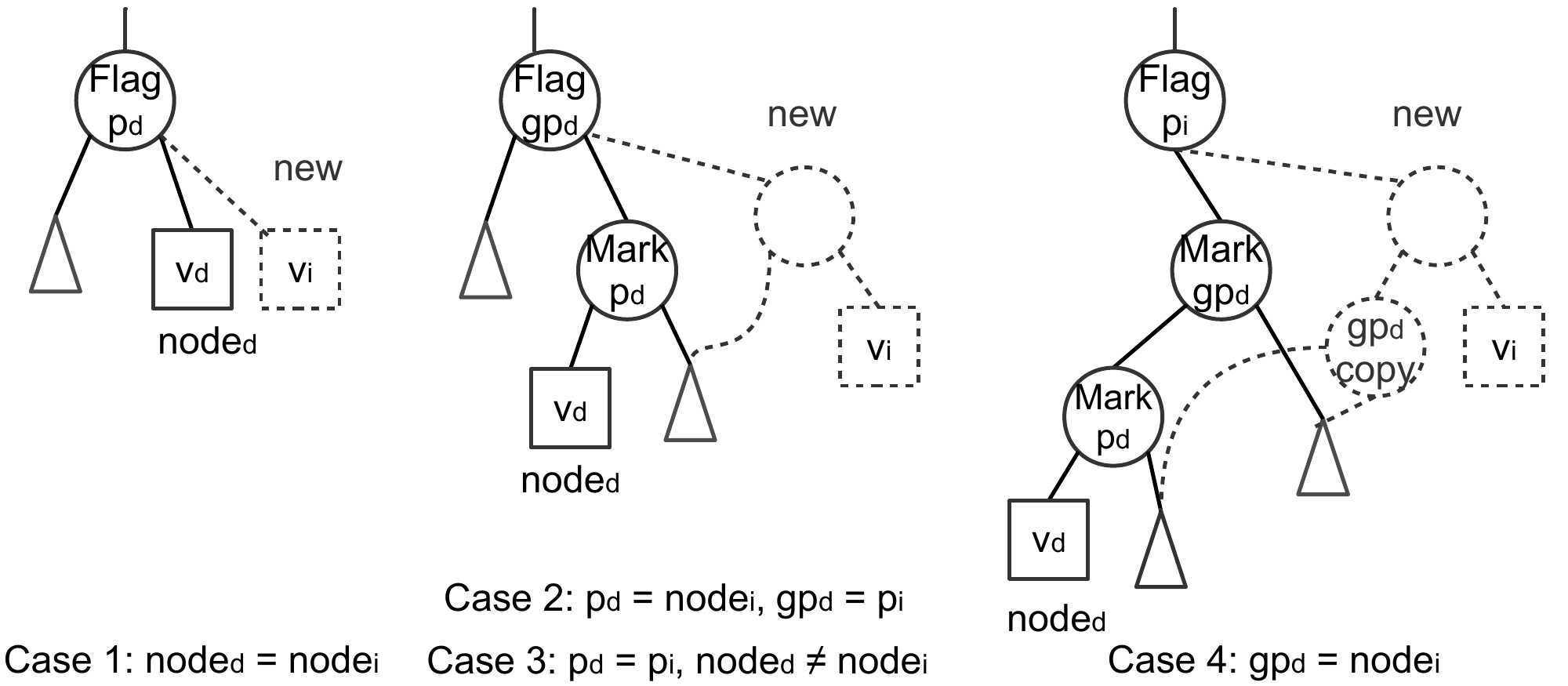}
\vspace*{-8mm}
\caption{Special cases of replace($v_d, v_i$).} 
\label{move-fig}
\vspace*{-7mm}
\end{figure}



\section{Algorithm Correctness} \label{proofSec}
A detailed proof of correctness is provided in Appendix.
It is quite lengthy, so we can only provide a brief sketch here.
First, we define the linearization point of each operation.
Let $\langle$ -, -, $node$, -, -, $rmvd \rangle$ be the result returned by a search. 
If $node$ is a leaf containing $v$ and $rmvd$ is false, we prove there is a time during the search when $node$ is logically in the trie and the search is linearized at that time.
Otherwise, we show there is a time during the search when no leaf containing $v$ is logically in the trie and the search is linearized at that time.
If an update returns false, it is linearized at the linearization point of the search that caused the update to fail.
Let $I$ be a Flag object created by an update.
If a child CAS performed by any call to help($I$) is executed, the update is linearized at the first such child CAS.
Next, we sketch the correctness proof in four parts.

Part 1 is the heart of the proof.
The goal of Part 1 is to prove that, for any Flag object $I$, the successful CAS steps performed by all calls to help($I$) proceed in the expected order.
(See Figure \ref{step-fig}.)
First, the flag CAS steps are performed on nodes in order, according to the nodes' $label$s.
We prove that only the first flag CAS (by any of the helpers) on each node can succeed.
If one if these fails, then the nodes that have been flagged are unflagged by backtrack CAS steps and all calls to help($I$) return false, indicating that the attempt at performing the update has failed.
Otherwise, the child CAS steps are performed, and then the unflag CAS steps remove flags from nodes that are still reachable.
If several helpers perform one of these CAS steps, we prove that the first helper succeeds and no others do.
In this case, all calls to help($I$) return true.

In Part 1, we first prove that the post-conditions of the search described in Section \ref{algSec} are satisfied.
Then, we prove that each type of update preserves the main invariant of the Patricia trie data structure: 
if $x.child[i] = y$, then $(x.label) \cdot i$ is a prefix of $y.label$. 
Thus, the structure is a correct trie.
(This also implies that $label$s of reachable nodes are distinct, so flagging can be done in order by $label$s of nodes to avoid deadlock.)
We show that the ABA problem on the {\it info} fields is avoided because whenever an {\it info} field is changed, it is set to a newly created Flag or Unflag object.
Then, we show that the CAS steps succeed in the correct order even if they are performed by helpers.

We say that a node is {\it marked} if its {\it info} field is a Flag object $I$ and the node does not appear in $I.unflag$ and some call to help($I$) has performed a child CAS.
We show that if an internal node is removed from the trie, it is marked at all times after that.
After a node is removed, it is never inserted into the trie again.
Next, we show that the ABA problem on the $child$ fields is avoided because whenever a child pointer is changed, the old child is permanently removed from the trie.
The proofs of the lemmas in Part 1 are mostly focused on the structure of the help routine.
So, any new update that preserves the main invariants of the trie can be added with minor changes to the correctness~proof.

\begin{figure*}[t]
\centering
\leavevmode
\includegraphics[width=13.70cm]{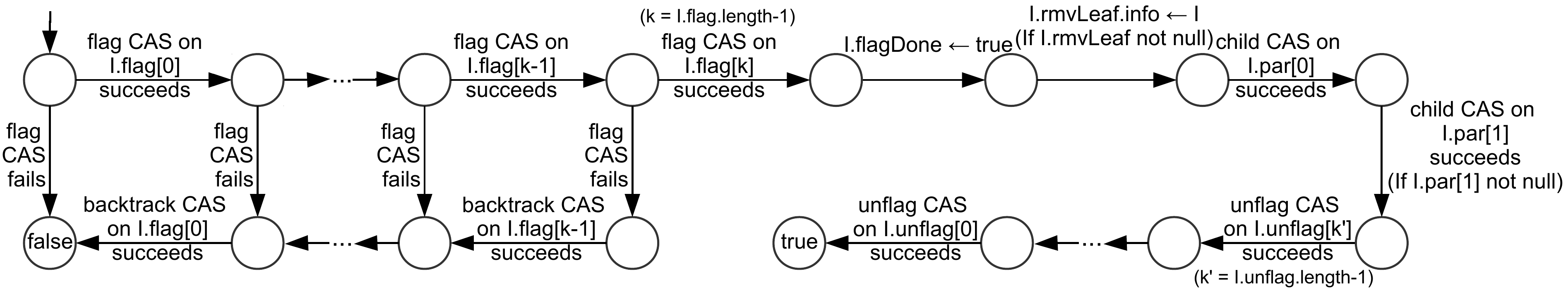}
\vspace*{-4mm}
\caption{The correct order of steps inside help($I$) for each Flag object $I$. (Steps can be performed by different calls to help($I$).) }
\label{step-fig}
\vspace*{-5mm}
\end{figure*}

Part 2 proves that search operations are linearized correctly.
First, we show that each node a search visits was reachable at some time during the operation.
Let $\langle$ -, -, $node$,  -, -, $rmvd\rangle$ be the result returned by a search($v$). 
If $node$ is a leaf containing $v$ and $rmvd$ is false, we show that $node$ was reachable and not logically removed at some time during the search.
The search is linearized at that time.
If $node$ is not a leaf containing $v$ or $node$ is a leaf containing $v$ but $rmvd$ is true, 
we show that there is a time during the search that no leaf containing $v$ is logically in the trie.
The search is linearized at that time. 

Part 3 proves that update operations are linearized correctly.
Let $T$ be the linearization point of a successful update operation.
Since all nodes are flagged successfully, no other concurrent update can change a flagged node's child 
between the time when the {\it info} field of node is read  for the last time during the search and the time the node is unflagged.
Thus, only the child CAS steps of the update would change that part of the trie during that period of the time.
Flagging ensures that the first child CAS of the update has the effect of implementing precisely the change shown in Figure \ref{ins-del-fig} or \ref{move-fig} atomically.

Part 4 proves that the implementation is non-blocking.
To derive a contradiction, assume after time $T$, no operation terminates or fails.
Let $I$ be a Flag object created by an update that is running after $T$.
If a call to help($I$) returns true, the update terminates, so after $T$, all calls to help($I$) return false.
Thus, all calls to help($I$) set $doChildCAS$ to false because they failed to flag an internal node successfully after $T$.
Consider the group of all calls to help($I$).
We say the group blames an internal node which is the first node that no call to help($I$) could flag successfully.
Let $g_0$, ..., $g_m$ be all these groups ordered by the $label$s of the nodes that they blame.
Since $g_m$ blames an internal node $x$, $x$ is flagged by some other group $g_i$ where $0 \le i < m$.
Thus, $g_i$ blames some other node $y$ whose $label$ is less than $x$.
So, $g_i$ flags $x$ before attempting to flag $y$, contradicting the fact that $g_i$ flags internal nodes in order.

\section{Empirical Evaluation} \label{evaluationSec}
We experimentally compared the performance of our implementation (PAT) with non-blocking binary search trees (BST) \cite{BST}, non-blocking k-ary search trees (4-ST) \cite{k-ary}, 
ConcurrentSkipListMap (SL) of the Java library, lock-based AVL trees (AVL) \cite{AVL} and non-blocking hash tries (Ctrie) \cite{Ctrie}.
For the k-ary search trees, we use the value $k=4$, which was found to be optimal in \cite{k-ary}.
Nodes in Ctrie have up to 32 children.

The experiments were executed on a Sun SPARC Enterprise T5240 with 32GB RAM.
The machine had two UltraSPARC T2+ processors, each having eight 1.2GHz cores, for a total of 128 hardware threads. 
The experiments were run  in Java. 
The sun JVM version 1.7.0\_3 was run in server mode.
The heap size was set to 2G.
This ensures the garbage collector would not be invoked too often, so that the measurements reflect the running time of the algorithms themselves. 
Using a smaller heap size affects the performance of BST, 4-ST and PAT more than AVL and SL since they create more objects.

We evaluated the algorithms in different scenarios.
We ran most experiments using uniformly distributed random keys.
We ran the algorithms using uniformly distributed keys in two different ranges:
$(0, 10^2)$ to measure performance under high contention
and $(0, 10^6)$ for low contention.
(We also ran the experiments for the key range of $(0, 10^3)$ for medium contention, but since the results were very similar to the low contention case, we do not present them here.)
We ran experiments with two different operation ratios: 
 5\% inserts, 5\% deletes and 90\% finds (i5-d5-f90),
and 50\% inserts, 50\% deletes and 0\% finds (i50-d50-f0).
(We also ran the experiments with ratio of 15\% inserts, 15\% deletes and 70\% finds.
Since the results were similar to the experiments with the ratio of (i5-d5-f90), we do not present them here.)

Since the replace operation is not used in these sets of experiments, we made some minor optimization to the pseudo-code.
For example, we eliminated the $rmvd$ variable in search operations.

Since the Java compiler optimizes its running code, before each experiment, we perform (i50-d50-f0) for ten seconds for each implementation.
We start each experiment with a tree initialized to be half-full, created by running updates in the ration i50-d50-f0.   
Each data point in our graphs is the average of eight 4-second trials.
(The error bars in the charts shows the standard deviation.)

For uniformly distributed keys, algorithms scale well under low contention (key range of $(0, 10^6)$).
(See Figure \ref{chart-random1}.)
Under very high contention (key range of $(0, 10^2)$), most scale reasonably well when the fraction of updates is low, but experience problems when all operation are updates.
(See Figure \ref{chart-random2}.)
When the range is $(0, 10^6)$, Ctrie outperforms all others since the height of the Ctrie is small compared to the others because node can have up to 32 children.
However when the range is $(0, 10^2)$ and the contention is very high, Ctrie does not scale.
Excluding Ctrie, when the range is  $(0, 10^6)$, PAT, 4-ST and BST outperform AVL and SL.
Since updates are more expensive than finds, the throughput is greater for i5-d5-f90 than for i50-d50-f0.


\begin{figure}[t]
\centering
\leavevmode
\includegraphics[width=6.7cm]{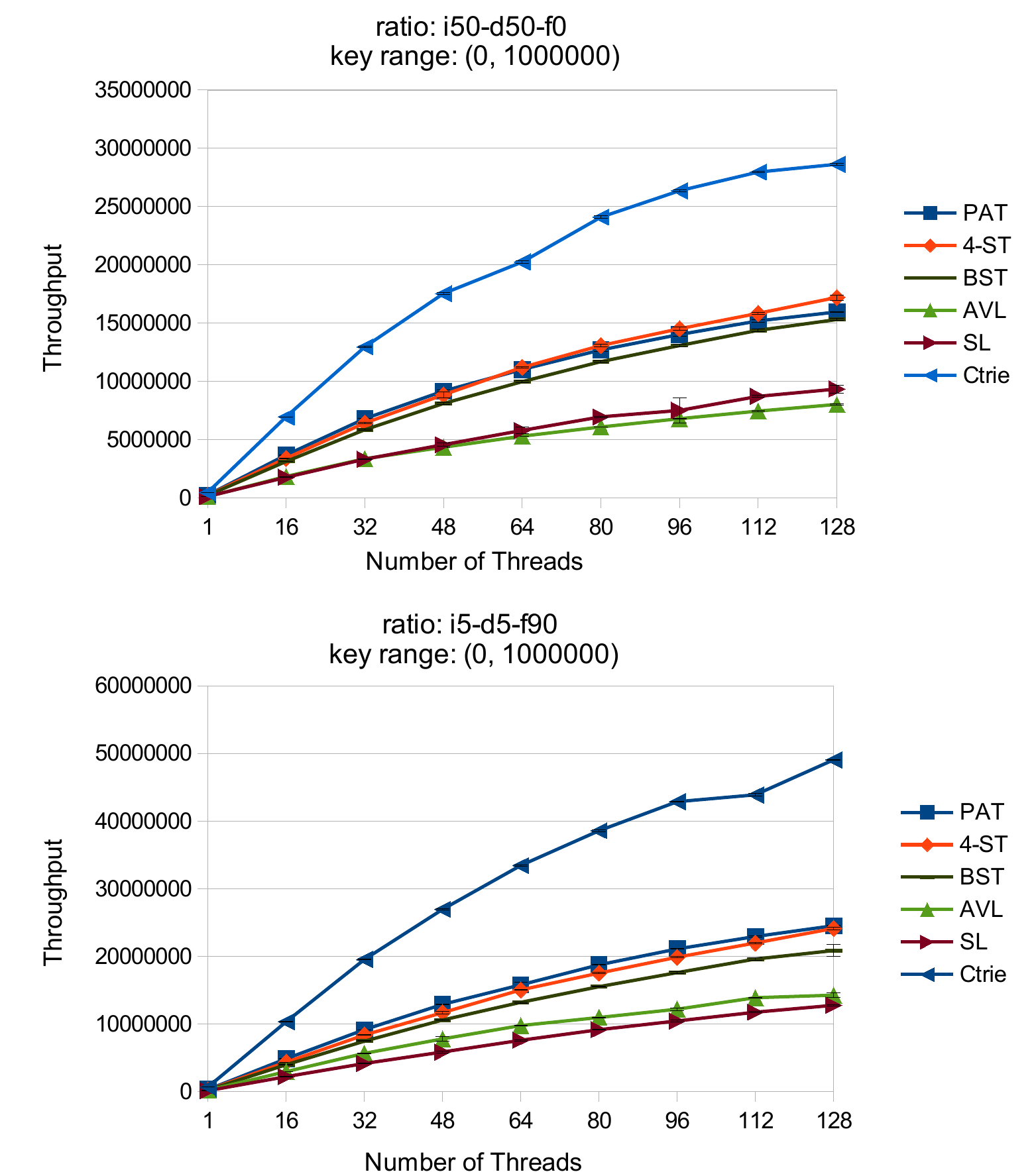}
\vspace*{-3mm}
\caption{Uniformly distributed keys}
\label{chart-random1}
\vspace*{-2mm}
\end{figure}

\begin{figure}[t]
\centering
\leavevmode
\includegraphics[width=6.5cm]{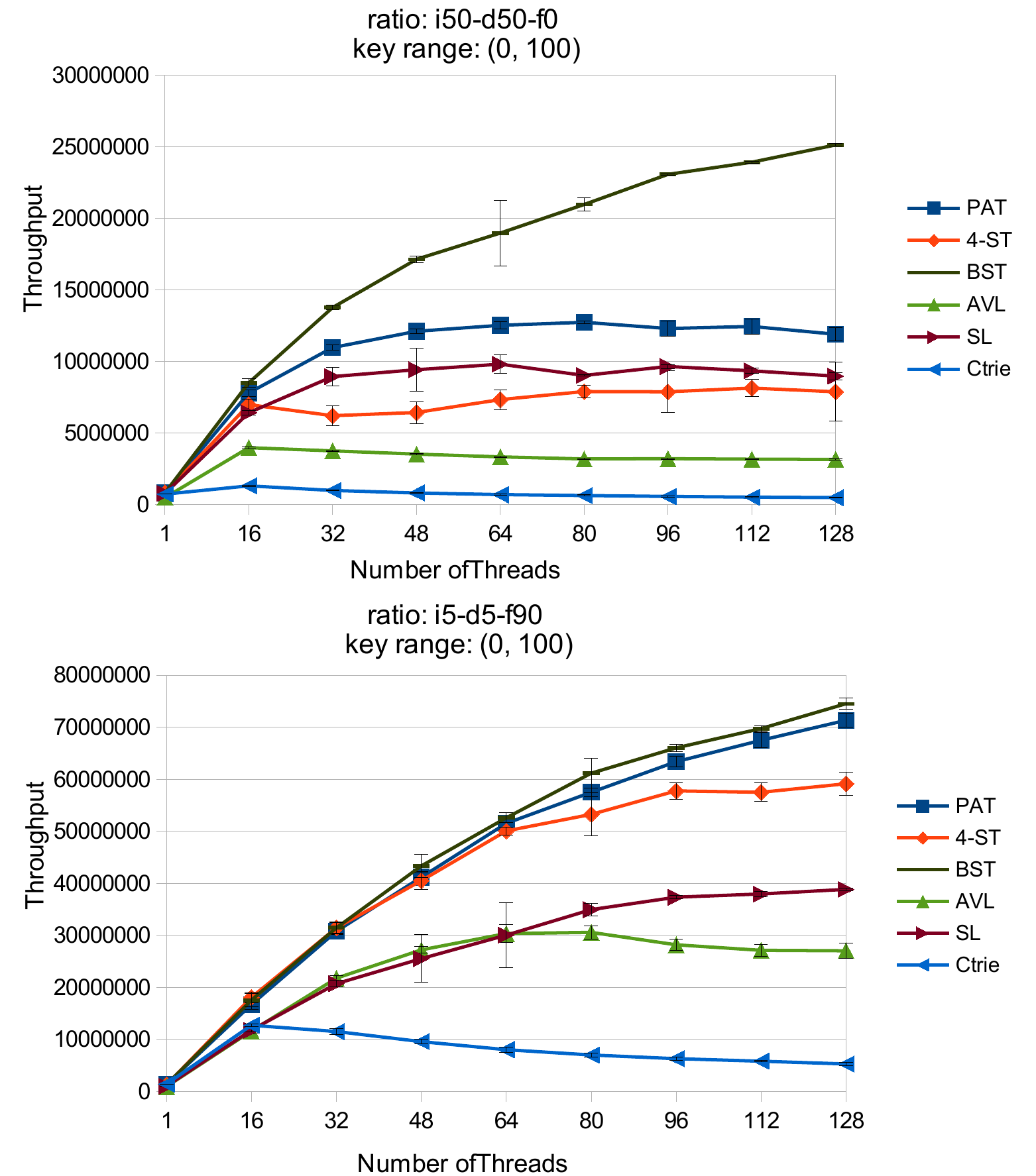}
\vspace*{-3mm}
\caption{Uniformly distributed keys}
\label{chart-random2}
\vspace*{-2mm}
\end{figure}

To evaluate the replace operations, we ran an experiment with 10\% inserts, 10\% deletes and 80\% replace operations (i10-d10-r80) and a key range of $(0, 10^6)$ on uniformly random keys.
(See Figure \ref{chart-move}.)
We could not compare these results with other data structure since none provide atomic replace operations.
As the chart shows, the replace operation scales well as the number of threads increases.
 
 \begin{figure}[t]
\centering
\leavevmode
\includegraphics[width=6.00cm]{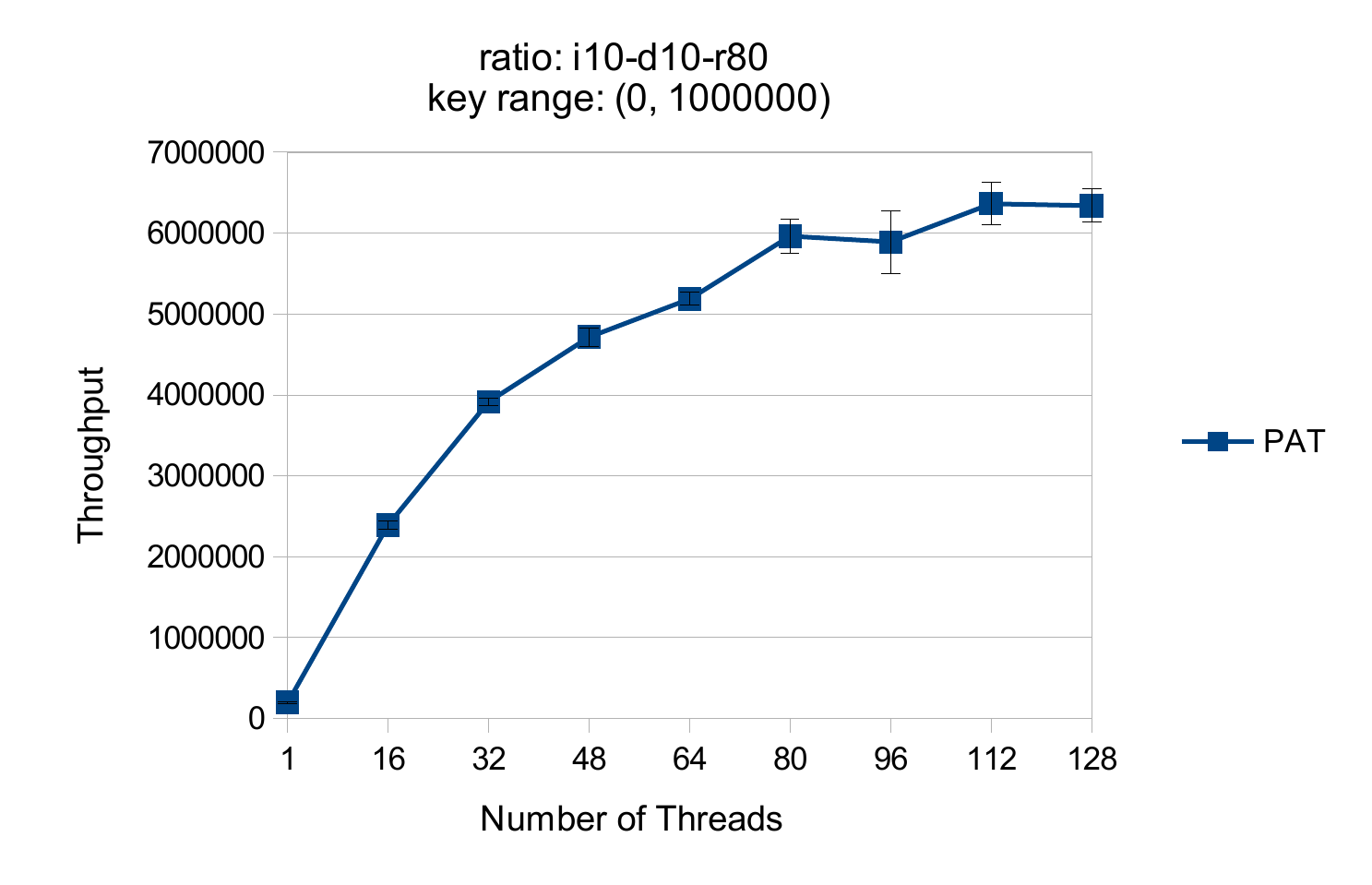}
\vspace*{-4mm}
\caption{Replace operations of PAT}
\label{chart-move}
\vspace*{-3mm}
\end{figure}

We also performed some experiments on non-uniformly distributed random keys.
To generate non-uniform keys, processes performed operations on sequence of 50 consecutive keys, starting from a randomly chosen key.
When keys are not uniformly distributed and key range is $(0, 10^6)$, Ctrie and then PAT outperform others greatly since they maintain a fixed height without doing expensive balancing operations.
(See Figure \ref{chart-chain}.)
Since the results of these experiments for other operations ratios were similar, only the chart for the ratio i15-d15-f70 is presented here.
Longer sequences of keys degrade the performance of BST and 4-ST even further.

\begin{figure}[t]
\centering
\leavevmode
\includegraphics[width=6.00cm]{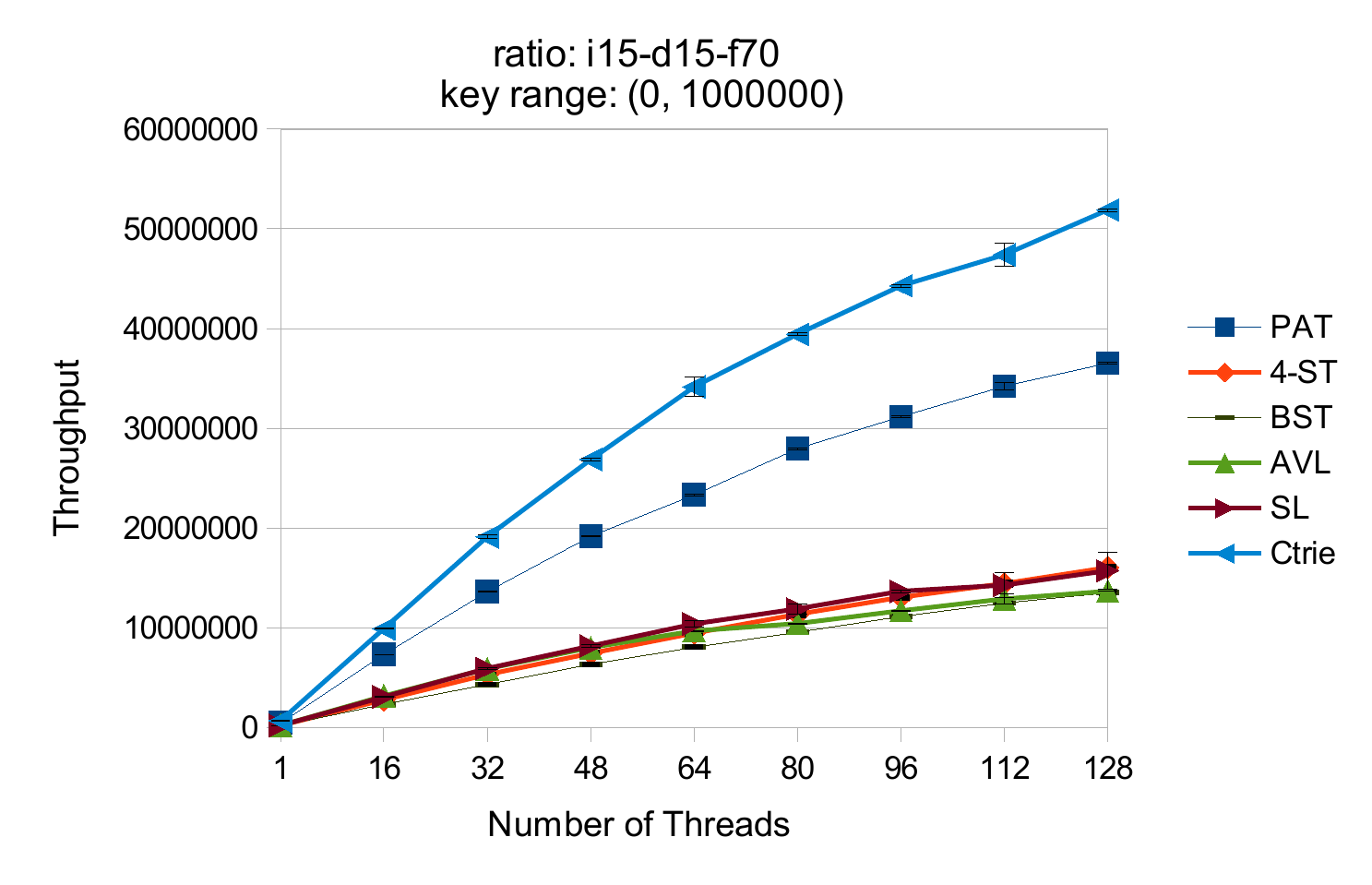}
\vspace*{-4mm}
\caption{Non-uniformly distributed keys (The lines for BST, 4-ST, AVL and SL overlap.)}
\label{chart-chain}
\vspace*{-7mm}
\end{figure}


\section{Conclusion} \label{conclusionSec}
Our algorithms can also be used to store unbounded length strings.
One approach would be to append \$ to the end of each string. 
To encode a binary string, $0$, $1$ and \$ can be represented by $01$, $10$ and $11$.
Then, every encoded key is greater than $00$ and smaller than $111$, so $00$ and $111$ can be used as keys of the two dummy nodes.
With this modifications, searches would be non-blocking but not wait-free.
Moreover, since $label$s of nodes never change, they need not fit in a dingle word.
The approach used in the replace operation can be used for operations on other data structures that must change several pointers atomically.
Future work includes providing the general framework for doing this on any tree-based structure.
Such a framework would have to guarantee that all changes become visible to query operations at the same time.
Brown et al.~\cite{gen-NB-tree} proposed a general technique for non-blocking trees that support one change to the tree atomically. 

Since our algorithms create many Flag objects to avoid using locks, finding more efficient memory management techniques is an important area for future work.

\noindent {\bf Acknowledgments.}
I thank my supervisor, Eric Ruppert for his guidance and support, Trevor Brown for providing lots of help and code for the experiments and Michael L. Scott for giving us access to his multicore machines.

\bibliographystyle{abbrv}
\bibliography{ref-2}

\onecolumn

\appendix

\section{Proof of Correctness}

\subsection{Preconditions and Basic Invariants}
In this section, we show that the algorithms satisfy some basic invariants.
First, we have the following observations from the pseudo-code.

\begin{observation} \label{basic-obs}
The $label$ field of a Node is never changed.
No field of an Info object is changed except the $state$ field.
\end{observation}

\begin{observation}\label{root-obs}
The $root$ pointer is never changed and $root.label$ is $\varepsilon$.
\end{observation}

Now, we prove that the precondition of the help routine is satisfied.
\begin{lemma} \label{help-pre-con-lem}
Each call to the help routine satisfies its precondition.
\end{lemma}
\begin{proof}
We show that, for each call to help($I$), $I$ is a Flag object.
Just before calling help($I$) at line \ref{newFlag-call-help} or \ref{createNode-call-help1}, the operation checks that $I$ is a Flag object.
Update operations call help($I$) at line \ref{ins-call-help}, \ref{del-call-help} or \ref{mov-call-help} just after creating a new Flag object $I$. 
\end{proof}

A CAS step that tries to change the $child$ field of an internal node at line \ref{help-change-child} inside help($I$) is called a {\it child CAS of $I$}.

Next, we show that each internal node has two non-null children, non-null values are passed to createNode and $node$ is set to non-null value during the search operation. 

\begin{lemma} \label{child-num-lem} \label{p-gp-internal-lem}
\begin{enumerate}
\item Every internal node has two non-null children.
\item Every call to createNode satisfies its pre-condition.
\item Any process executing the loop of search has a non-null value in its $node$ variable.
\item If a search returns $\langle gp$, $p$, $node$, -, -, -$\rangle$, then $p$ is an internal node and if $gp$ is not null, $gp$ is an internal node.
\end{enumerate}
\end{lemma}

\begin{proof}
Assume the lemma is true before the step at time $T$.
We shall prove the lemma is true after that step.

(1) We show that Claim 1 is true if a new internal node is created at $T$ or if a $child$ field is changed at $T$.  
First, suppose a new internal node is created at time $T$. 
We show that the internal node that is created at $T$ has two non-null children.
A new internal node is created at line \ref{initial-root}, \ref{create-internal}, \ref{create-copy} or \ref{mov-create-copy}. 
If a new internal node is created at line \ref{initial-root} at time $T$, the children of the new node are set to two new nodes that are created at that line.
If a new copy of an internal node is made at line \ref{create-copy} or \ref{mov-create-copy} at time $T$, 
since the lemma is true at all times before $T$, a new copy of some non-null node is made at $T$ and the children of the new copy of the node are set to two non-null nodes at $T$. 
If a new node is created at line \ref{create-internal} inside createNode($node_1$, $node_2$),
since the invariant is true before $T$, $node_1$ and $node_2$ are non-null nodes and the children of the new node are set to $node_1$ and $node_2$ at $T$.

Now, suppose the $child$ field of an internal node is changed at $T$.
The $child$ field of an internal node is changed only at line \ref{help-change-child}.
Let $I$ be a Flag object such that a child CAS of $I$ is executed at $T$.
We consider each way that $I$ could have been created.
For each case, we show that, for all $i$, $I.newChild[i]$ is a non-null node. 

If $I$ is created at line \ref{ins-set-info1}, \ref{ins-set-info2}, \ref{mov-set-info1}, \ref{mov-set-info2}, \ref{mov-set-info-sc23} or \ref{mov-set-info-sc4}, 
$I.newNode[0]$ is set to a non-null node that is created at line \ref{create-internal} inside 
createNode, which is called at line \ref{ins-call-create-node}, \ref{mov-call-create-node}, \ref{mov-create-node1} or \ref{mov-create-node3}.

If $I$ is created at line \ref{mov-set-info-sc1}, $I.newNode[0]$ is set to a non-null node that is created at that line.

If $i=0$ and $I$ is created at line \ref{del-set-info} or $i=1$ and $I$ is created at line \ref{mov-set-info1} or \ref{mov-set-info2}, 
let $\langle$ -, $p$, -, -, -, -$\rangle$ be the result returned by the preceding call to search($val$) on line \ref{del-call-search} or \ref{mov-call-search2}.
Then, $I.newNode[0]$ is set to a child of $p$ that is read at line \ref{del-read-sibling} or \ref{mov-read-sibling} before $T$.
Since the lemma is true before $T$, $I.newNode[0]$ is set to a non-null node.

(2) Suppose createNode($node_1$, $node_2$) is called at $T$ from line  \ref{ins-call-create-node}, \ref{mov-call-create-node}, \ref{mov-create-node1}, \ref{mov-create-node2} or \ref{mov-create-node3}.
We show that $node_1$ and $node_2$ are non-null.

If createNode($node_1$, $node_2$) is called at line \ref{ins-call-create-node} or \ref{mov-call-create-node}, 
$node_1$ is a new copy of $node$ that is returned by the search operation at line \ref{ins-call-search} or \ref{mov-call-search2}.
Since the lemma is true prior to $T$, $node_1$ is non-null.
Then, $node_2$ is a new leaf node that is created at line \ref{ins-call-create-node} or \ref{mov-call-create-node}.

If createNode($node_1$, $node_2$) is called at line \ref{mov-create-node1}, let $\langle$-, $p$, -, -, -, -$\rangle$ be the result returned by the call to search($val$) on line \ref{mov-call-search2}.
Then, $node_2$ is a new leaf node that is created at  line \ref{mov-create-node1} and $node_1$ is a child of $p$ that is read at line \ref{mov-read-sibling} before $T$.
Since the lemma is true before $T$, $node_1$ is a non-null node.

If createNode($node_1$, $node_2$) is called at line \ref{mov-create-node2}, let $\langle gp$, $p$, -, -, -, -$\rangle$ be the result returned by the call to search($val$) on line \ref{mov-call-search2}.
Then, $node_1$ is a child of $p$ that is read at line \ref{mov-read-sibling} and
$node_2$ is a child of $gp$ that is read at line \ref{mov-read-pSibling}.
Since the lemma is true before $T$, $node_1$ and $node_2$ are non-null nodes.

If createNode($node_1$, $node_2$) is called at line \ref{mov-create-node3}, $node_1$ is a new node that is created at line \ref{mov-create-node2} and
$node_2$ is a new leaf node that is created at line \ref{mov-create-node3}.

(3) Before entering the loop of the search routines, $node$ is set to $root$, which is not null.
We must show that if $node$ is set on line \ref{search-set-node} at $T$, it is non-null.
Since $node$ is set to a child of a node at $T$ and the lemma is true prior to $T$, $node$ is set to a non-null node at $T$.

(4) During the search operation, $node$ is initialized to $root$ at line \ref{search-set-initial-node}, so at least one iteration of the loop is performed 
(since $root$ is an internal node and $root.label = \varepsilon$ by Observation \ref{root-obs}).
By the test at line \ref{search-check-condition}, $p$ is an internal node after every execution of line \ref{search-set-p}.

In the first loop iteration, $gp$ is set to null.
In each subsequent iteration, $gp$ is set to an internal node at line \ref{search-set-gp} (since $p$ is always an internal node after each execution of line \ref{search-set-p}).
\end{proof}

By Lemma \ref{p-gp-internal-lem}, we have the following observation.

\begin{observation} \label{node-internal-obs}
Let $\langle gp_d$, -, -, -, -, -$\rangle$ be the result returned by a call to the search operation on line \ref{mov-call-search1} and
$\langle$ -, -, $node_i$, -, -, -$\rangle$ be the result returned by another call to the search operation on line \ref{mov-call-search2}.
If $gp_d = node_i$, then $node_i$ is an internal node.
\end{observation}

We use $a \cdot b$ to denote the concatenation of $a$ and $b$.
Now, we show that the search operation satisfies some of its post-conditions if it returns.
\begin{lemma} \label{search-lem}
Assume search($val$) returns $\langle gp$, $p$, $node$, $gpInfo$, $pInfo$, $rmvd \rangle$.
The following statements are true.
\begin{enumerate}
\item  If $gp$ is not null, then, at some time during the search operation, $gp.info$ was $gpInfo$, and 
at some later time during the search operation, $p$ was a child of $gp$.
\item Then, at some later time during the search operation, $p.info$ was $pInfo$, and at some later time during the search operation, $p.child[i] = node$ for some $i$.
\item $(p.label) \cdot i$ is a prefix of $val$.
\item If $node$ is an internal node, $node.label$ is not a prefix of $val$.
\end{enumerate}
\end{lemma}
\begin{proof}
During the search operation, $node$ is initialized to $root$ at line \ref{search-set-initial-node}, so at least one iteration of the loop is performed 
(since $root$ is an internal node and $root.label = \varepsilon$ by Observation \ref{root-obs}).
\begin{enumerate}
\item 
Assume $gp$ is not null. 
By Lemma \ref{p-gp-internal-lem}, $gp$ is an internal node.
Since $gp$ is set to null during the first loop iteration, the search operation does not exit the loop after the first loop iteration. 
During the second last loop iteration, $pInfo$ and $p$ are set at line \ref{search-set-p} and $node$ is set to a child of $p$ at line \ref{search-set-node}.
Then, during the last loop iteration, $gpInfo$ and $gp$ are set to $pInfo$ and $p$ at line \ref{search-set-gp} and $p$ is set to $node$ at line \ref{search-set-p}.
So, when line \ref{search-set-p} was executed for the second last time, $gp.info = gpInfo$ and then, when line \ref{search-set-node} was executed for the second last time, $p$ was a child of $gp$.
 
\item
By Lemma \ref{p-gp-internal-lem}, $p$ is an internal node.
During the last loop iteration, $pInfo$ and $p$ are set at line \ref{search-set-p} and $node$ is set to a child of $p$ at line \ref{search-set-node}.
So, when line \ref{search-set-p} was executed for the last time, $p.info = pInfo$ and then, when line \ref{search-set-node} was executed for the last time, $p.child[i] = node$ for some $i$.

\item 
During the last loop iteration, $p$ is set to $node$  at line \ref{search-set-p}.
Since the condition at line \ref{search-check-condition} is true at the beginning of that iteration, 
$p$ is set to an internal node whose $label$ is a prefix of $val$. 
At the last execution of line \ref{search-set-node}, $node$ is set to $p.child[i]$ such that $i = (|p.label|+1)$th bit of $val$.
So, $(p.label) \cdot i$ is a prefix of $val$.

\item 
Since the search operation exits the loop, the condition is false at the last execution of line \ref{search-check-condition}. 
If $node$ is an internal node, $node.label$ is not a prefix of~$val$.
\end{enumerate}
\end{proof}

Next, we show that, just after a $child$ field is set, the following Patricia trie property is preserved: 
if $x.child[i] = y$, then $(x.label) \cdot i$ is a prefix of $y.label$.
\begin{invariant} \label{prefix-inv}
Let $x$ be an internal node and $y = x.child[i]$.
Then, $(x.label) \cdot i$ is a prefix of $y.label$.
\end{invariant}  
\begin{proof}
By Observation \ref{basic-obs}, no $label$ field of a node is ever changed.
So, we need only show that the lines \ref{initial-root}, \ref{create-copy}, \ref{create-internal} and \ref{mov-create-copy}, 
which create new internal nodes, and line \ref{help-change-child}, which changes the $child$ field of an internal node, preserve the invariant.
We assume the invariant holds at all times prior to $T$ and we prove that the invariant is true just after $T$.
First, we show if a new internal node is created at $T$, the invariant is preserved.

Line \ref{initial-root} creates a new internal node $root$ that has two children. 
Initially, $root.label = \varepsilon$, $root.child[0].label = 000...0$ and $root.child[1].label = 111...1$.

If a new copy of an internal node is created on line \ref{create-copy} or \ref{mov-create-copy} at $T$, 
since the invariant is true at all times before $T$, the new copy of the node satisfies the invariant.

If a new node is created at line \ref{create-internal} inside createNode, 
its $label$ is the longest common prefix of the $label$ fields of its children. 
Since the two children's $label$s are not prefixes of each other, the invariant is true.

In the remainder of the proof, we show if a child CAS of $I$ succeeds at $T$, the invariant is preserved.
The child CAS changes $I.pNode[j].child[i]$ from $I.oldChild[j]$ to $I.newChild[j]$ (for some $j$) where $i$ is the $(|I.pNode[j].label+1|)$th bit of $I.newChild[j].label$.
Thus, it suffices to show that $I.pNode[j].label$ is a proper prefix of $I.newChild[j].label$. 
We consider each way that $I$ could have been created.
By Observation \ref{basic-obs}, after the initialization of $I$, no field of $I$ is changed except the $flagDone$ field. 

Case 1: $I$ is created at line \ref{ins-set-info1} or \ref{ins-set-info2}.
Let $\langle$ -, $p$, $node$, -, -, -$\rangle$ be the result returned by the call to search($val$) on line \ref{ins-call-search} that precedes the creation of $I$.
Then, $newNode$ is the new node that is created at line \ref{ins-call-create-node} and whose children are a new copy of $node$ and a new leaf node whose $label$ is $val$.
In this case, $I.pNode[0] = p$ and $I.newChild[0] = newNode$. 
By Lemma \ref{search-lem}, $p.child[k] = node$ for some $k$ at some time before $T$.
Since the invariant is true before $T$, $(p.label) \cdot k$ is a prefix of $node.label$.
By Lemma \ref{search-lem},  $(p.label) \cdot k$ is a prefix of $val$.
Since $newNode.label$ is the longest common prefix of $val$ and $node.label$, $(p.label) \cdot k$ is a prefix of $newNode.label$. 

Case 2: $I$ is created at line \ref{del-set-info}.
Let $\langle gp$, $p$, $node$, -, -, -$\rangle$ be the result returned by the call to search($val$) on line \ref{del-call-search} that precedes the creation of $I$.
Let $nodeSibling$ be the child of $p$ that is read at line \ref{del-read-sibling}.
Since the invariant is true before $T$, $p.label$ is a prefix of $nodeSibling.label$.
In this case, $I.pNode[0] = gp$, $I.oldChild[0] = p$, and $I.newChild[0] = nodeSibling$. 
Since the child CAS succeeds at $T$, $gp.child[i]$ was $p$ just before $T$.
Since the invariant is true before $T$, $(gp.label) \cdot i$ is a prefix of $p.label$.
Since $p.label$ is a prefix of $nodeSibling.label$, $(gp.label) \cdot i$ is a prefix of $nodeSibling.label$.

Case 3: $I$ is created at line \ref{mov-set-info1} or \ref{mov-set-info2}.
Let $\langle gp_d$, $p_d$, $node_d$, -, -, -$\rangle$ and $\langle$ -, $p_i$, $node_i$, -, -, -$\rangle$ be the results returned by 
the calls to the search routine on line \ref{mov-call-search1} and \ref{mov-call-search2} that precede the creation of $I$.
Then, $I.pNode[0] = p_i$ and $I.newChild[0]$ is the new node that is created at line \ref{mov-call-create-node}.
By the same argument as in Case 1, the invariant is true if a child CAS of $I$ on $I.pNode[0]$ succeeds at $T$.
Let $nodeSibling$ be the child of $p_d$ that is read on line \ref{mov-read-sibling}.
Then, $I.pNode[1]=gp_d$ and $I.newChild[1]=nodeSibling$.
By the same argument as in Case 2, the invariant is true if a child CAS of $I$ on $I.pNode[1]$ succeeds at $T$.

Case 4: $I$ is created at line \ref{mov-set-info-sc1}.
Let $\langle$ -, $p_i$, $node_i$, -, -, -$\rangle$ be the result returned by the call to search($val_i$) on line \ref{mov-call-search2} that precedes the creation of $I$.
In this case, $I.pNode[0] = p_i$ and $I.newChild[0]$ is a new leaf node whose $label$ is $val_i$. 
By Lemma \ref{search-lem}, $p_i.label$ is a proper prefix of $val_i$.

Case 5: $I$ is created at line \ref{mov-set-info-sc23}. 
Let $\langle gp_d$, $p_d$, $node_d$, -, -, -$\rangle$ and $\langle$ -, $p_i$, $node_i$, -, -, -$\rangle$ be the results returned by 
the calls to the search routine on line \ref{mov-call-search1} and \ref{mov-call-search2} that precede the creation of $I$.
Let $nodeSibling$ be the child of $p_d$ that is read at line \ref{mov-read-sibling} and 
$newNode$ be the new node that is created at line \ref{mov-create-node1} and whose children are $nodeSibling$ and a new leaf node whose $label$ is $val_i$.
In this case, $I.pNode[0] = p_i$, $I.oldChild[0] = node_i$ and $I.newChild[0] = newNode$.
Since the child CAS succeeds at $T$, $gp_d.child[i]$ was $p_d$ just before $T$.
Since the invariant is true before $T$, $(gp_d.label) \cdot i$ is a prefix of $p_d.label$.
Since the invariant is true before $T$ and $nodeSibling$ is a child of $p_d$ before $T$, $p_d.label$ is a prefix of $nodeSibling.label$.
So, $(gp_d.label) \cdot i$ is a prefix of $nodeSibling.label$.

If the condition at line \ref{mov-condition-sc2} is true, by Lemma \ref{search-lem}, $(gp_d.label) \cdot i$ is a prefix of $val_i$ (since $p_i = gp_d$ and $node_i = p_d$).
If the condition at line \ref{mov-condition-sc3} is true, by Lemma \ref{search-lem}, $p_i.label$ is a prefix of $val_i$, so  $(gp_d.label) \cdot i$ is a prefix of $val_i$ (since $p_i = p_d$).
Since the children of $newNode$ are $nodeSibling$ and a new leaf node whose $label$ is $val_i$, $(gp_d.label) \cdot i$ is a prefix of $newNode.label$.

Case 6: $I$ is created at line \ref{mov-set-info-sc4}.
Let $\langle gp_d$, $p_d$, $node_d$, -, -, -$\rangle$ and $\langle$ -, $p_i$, $node_i$, -, -, -$\rangle$ be the results returned by 
the calls to the search routine on line \ref{mov-call-search1} and \ref{mov-call-search2} that precede the creation of $I$.
Let $nodeSibling$ be the child of $p_d$ that is read at line \ref{mov-read-sibling} and $pSibling$ be the child of $gp_d$ that is read at line \ref{mov-read-pSibling}. 
Then, $newChild$ is the new node that is created at line \ref{mov-create-node2} and whose children are $nodeSibling$ and $pSibling$.
Then, $newNode$ is the new node that is created at line \ref{mov-create-node3} and whose children are $newChild$ and a new leaf node whose $label$ is $val_i$.
In this case, $node_i = gp_d$, $I.pNode[0] = p_i$, $I.oldChild[0] = node_i$  and $I.newChild[0] = newNode$. 

By Lemma \ref{search-lem}, $p_d$ was a child of $gp_d = node_i$ at some time before $T$.
Since the invariant is true before $T$, $node_i.label$ is a prefix of $p_d.label$.
Since the invariant is true before $T$ and $nodeSibling$ is a child of $p_d$ before $T$, $p_d.label$ is a prefix of $nodeSibling.label$.
So, $node_i.label$ is a prefix of $nodeSibling.label$.
Since the child CAS succeeds at $T$, $p_i.child[i] = node_i$ just before $T$.
Since the invariant is true before $T$, $(p_i.label) \cdot i$ is a prefix of $node_i.label$.
Since $node_i.label$ is a prefix of $nodeSibling.label$, $(p_i.label) \cdot i$ is a prefix of $nodeSibling.label$.

Since the invariant is true before $T$ and $pSibling$ is a child of $gp_d = node_i$ before $T$, $node_i.label$ is a prefix of $pSibling.label$.
Since $(p_i.label) \cdot i$ is a prefix of $node_i.label$, $(p_i.label) \cdot i$ is a prefix of $pSibling.label$.
Since the children of $newChild$ are $nodeSibling$ and $pSibling$, $(p_i.label) \cdot i$ is a prefix of $newChild.label$.

By Lemma \ref{search-lem}, $(p_i.label) \cdot i$ is a prefix of $val_i$.
Since the children of $newNode$ are $newChild$ and a new leaf node whose $label$ is $val_i$, $(p_i.label) \cdot i$ is a prefix of $newNode.label$.
\end{proof}

We say that an internal node $pNode$ is a {\it parent} of a node $cNode$ at time $T$, if $cNode$ is an element of $pNode.child$ at time $T$.
Note that a node $cNode$ may have more than one parent at time $T$.
(For example, this occurs if $cNode$'s parent has been removed from the tree but still has a $child$ pointer to $cNode$.)

We say that a node $anc$ is an {\it ancestor} of a node $x$ at time $T$ if there is a path consisting of child pointers from $anc$ to $x$ at time $T$. 
We say that node $x$ is a {\it descendant} of $anc$ at time $T$.
If $anc \ne x$, we say $anc$ is a {\it proper ancestor} of a node $x$ at time $T$ and $x$ is a {\it proper descendant} of a node $anc$ at time $T$.
We say a node $x$ is {\it reachable at time $T$} if $root$ is an ancestor of $x$ at time $T$. 

The $root$ cannot have a parent $x$ because $x.label$ would have to be a proper prefix of $root.label = \varepsilon$ by Invariant \ref{prefix-inv}.
\begin{corollary} \label{root-no-parent-col}
The $root$ does not have any parent at any time.
\end{corollary}

\begin{lemma} \label{no-same-node-lem} \label{reachable-parent-lem}
At all times, if $node_1$ and $node_2$ are reachable and have the same $label$, then $node_1 = node_2$ and there is exactly one reachable parent of $node_1$ (unless $node_1$ is the $root$).
\end{lemma}
\begin{proof}
Let $v_0, v_1, ..., v_{n-1}, v_n$ be a path consisting of child pointers at time $T$ such that $v_0 = root$ and $v_n = node_1$ and
$v'_0, v'_1, ..., v'_{m-1}, v'_m$ be a path consisting of child pointers at time $T$ such that $v'_0 = root$ and $v'_n = node_2$.
Without loss of generality, assume $n \le m$.

By induction, we show these two paths are the same.
Since $v_0 = root$ and $v'_0 = root$, $v_0 = v'_0$.
Assuming $v_{i-1} = v'_{i-1}$, we show that $v_i = v'_i$.
By Invariant \ref{prefix-inv}, $(v_{i-1}.label) \cdot j$ is a prefix of $node_1.label = node_2.label$ for some $j$.
Since, by Invariant \ref{prefix-inv}, $v_i.label$ and $v'_i.label$ are prefixes of $node_1.label = node_2.label$ and $v_{i-1}.label$ is a proper prefix of $v_i.label$ and $v'_i.label$,
$(v_{i-1}.label) \cdot j$ is a prefix of $v_i.label$ and $v'_i.label$.
By Invariant \ref{prefix-inv}, $v_{i-1}.child[j] = v_i$ at time $T$ and $v_{i-1}.child[j] = v'_i$ at time $T$, so $v_i = v'_i$.

So, $node_1 = v'_n$. 
Then, $n = m$ since $v'_n.label$ is not a proper prefix of $node_2.label$. 
Therefore, $v_n = v'_m$ and $node_1 = node_2$.

Since  $v_{n-1} = v'_{m-1}$, if $node_1 \ne root$, there is exactly one reachable parent of $node_1$ at time~$T$.
\end{proof}

\subsection{Behaviour of CAS Steps on $info$ Fields}
In this section, we show how CAS steps change the $info$ field of nodes.
First, we define flagging and unflagging of nodes formally.

\begin{definition}
Let $x$ be a node.
If $x.info$ is a Flag object at time $T$, we say $x$ is {\it flagged at time $T$}.
The CAS step at line \ref{help-flag-node} is called a {\it flag CAS}.
\end{definition}

\begin{definition}
Let $x$ be an internal node.
If $x.info$ is an Unflag object at time $T$, we say $x$ is {\it unflagged at time $T$}.
The CAS step at line \ref{help-unflag} is called an {\it unflag CAS}.
The CAS step at line \ref{help-unflag-bt} is called a {\it backtrack CAS}.
\end{definition}

All nodes are unflagged when they are created.
A leaf node can only become flagged at line \ref{help-flag-leaf}.
Once a leaf node is flagged, it can never become unflagged.
The following lemma describes how the $info$ field of an internal node is initialized and changed when the node becomes flagged or unflagged.

\begin{lemma} \label{info-lem}
Let $x$ be an internal node.
When node $x$ is created, $x.info$ is initially set to a new Unflag object.
The only changes to $x.info$ that can occur are 
(1) a flag CAS at line \ref{help-flag-node} that changes $x.info$ from an Unflag object to a Flag object, or
(2) an unflag CAS at line \ref{help-unflag} or a backtrack CAS at line \ref{help-unflag-bt} that changes $x.info$ from a Flag object to a newly created Unflag object.
\end{lemma}
\begin{proof} 
A new internal node is only created at line \ref{initial-root}, \ref{create-copy}, \ref{create-internal} or \ref{mov-create-copy} and 
the $info$ field of the new internal node is initially set to a new Unflag object. 
The only lines that can change $x.info$ are line \ref{help-flag-node}, \ref{help-unflag} and \ref{help-unflag-bt}.

Suppose a flag CAS at line \ref{help-flag-node} inside help($I$) succeeds to change $x.info$ from an element of $I.oldInfo$ to $I$.
By Lemma \ref{help-pre-con-lem}, $I$ is a Flag object.
Before creating $I$, the operation checks that each value stored in $I.oldInfo$ is an Unflag object at line \ref{newFlag-old-info}.

Suppose an unflag or backtrack CAS inside help($I$) succeeds to change $x.info$.
By Lemma \ref{help-pre-con-lem}, $I$ is a Flag object. 
So, $x.info$ was equal to $I$ just before the unflag or backtrack CAS and, just after the unflag or backtrack CAS, $x.info$ is changed to a newly created Unflag object. 
\end{proof}

Let $I$ be a Flag object.
Elements of $I.flag$ are ordered at line \ref{newFlag-sort-nodes} before $I$ is created.
Let $F_I$ be the set of all nodes in the $I.flag$ array and $U_I$ be the set of all nodes in the $I.unflag$ array.
By the pseudo-code, $U_I \subseteq F_I$.

A flag, unflag or backtrack CAS step executed inside help($I$) is called a CAS step of $I$.
A flag CAS of $I$ attempts to change the $info$ field of an internal node in $F_I$ from an Unflag object to $I$.
An unflag CAS of $I$ attempts to change the $info$ field of an internal node in $U_I$ from $I$ to a new Unflag object.
A backtrack CAS of $I$ attempts to change the $info$ field of an internal node in $F_I$ from $I$ to a new Unflag object.

By Lemma \ref{info-lem}, we have the following corollary.
\begin{corollary} \label{unflag-after-flag-col}
Let $I$ be a Flag object.
After an internal node $x$ is flagged by a flag CAS of $I$, the next change to $x.info$ can only be done by an unflag or backtrack CAS of~$I$.
\end{corollary}

Let $I$ be a Flag object.
A flag CAS of $I$ at line \ref{help-flag-node} inside help($I$) tries to change $I.flag[i].info$ from $I.oldInfo[i]$ to $I$ for some $i$.
The following lemma shows that $I.flag[i].info$ was equal to $I.oldInfo[i]$ at some time before $I$ was created.

\begin{lemma} \label{info-old-lem}
Let $I$ be a Flag object.
Then, each entry of $I.flag[i]$ is an internal node.
Also, if $x = I.flag[i]$ (for some $i$), then $x.info$ was $I.oldInfo[i]$ at some time before $I$ was created.
\end{lemma}
\begin{proof}
Just before creating $I$, an update operation calls search once (for insert and delete) or twice (for replace).
If one such search returns $\langle gp$, $p$, $node$, $gpInfo$, $pInfo$, -$\rangle$, 
then $I.flag[i]$ is set to $p$, $gp$ or $node$ at line \ref{ins-set-info1}, \ref{ins-set-info2}, \ref{del-set-info}, \ref{mov-set-info1}, \ref{mov-set-info2}, \ref{mov-set-info-sc1}, \ref{mov-set-info-sc23} or \ref{mov-set-info-sc4}.

By Lemma \ref{p-gp-internal-lem}, $p$ is an internal node and $gp$ is an internal node if $gp$ is not null.
If $I.flag[i]$ is set to $gp$ at line \ref{del-set-info},  \ref{mov-set-info1}, \ref{mov-set-info2}, \ref{mov-set-info-sc23}, $gp$ is not null by the pseudo-code.
If $I.flag[i]$ is set to $gp$ at line \ref{mov-set-info-sc4}, $gp$ is not null by Observation \ref{node-internal-obs}.  
When  $I.flag[i]$ is set to $p$ or $gp$ at line \ref{ins-set-info1}, \ref{ins-set-info2}, \ref{del-set-info}, \ref{mov-set-info1}, \ref{mov-set-info2}, \ref{mov-set-info-sc1}, \ref{mov-set-info-sc23} or \ref{mov-set-info-sc4}, $I.oldInfo[i]$ is set to the corresponding $pInfo$ or $gpInfo$.
By Lemma \ref{search-lem}, $I.flag[i].info = I.oldInfo[i]$ at some time during search($val$).

If $I.flag[i]$ is set to $node$ at line \ref{ins-set-info1} or \ref{mov-set-info1}, $node$ is an internal node.
Then, $I.oldInfo[i]$ is set to a value read from $node.info$ at line \ref{ins-read-node-info} or \ref{mov-read-node-info}.
\end{proof}

Suppose some process reads a value $old$ in some variable $x$ at time $T$ and then later performs a CAS that changes $x$ from $old$ to $new$ at time $T'$.
Some of the later parts of our proof of correctness will rely on the fact that $x$ has not been changed from $old$ between $T$ and $T'$.
If $x$ is changed between $T$ and $T'$ and $x$ is changed back to $old$ just before $T'$, the CAS step would incorrectly change $x$ to $new$. 
This situation is called the ABA problem.
Since CAS steps are used to change the $info$ and $child$ fields of internal nodes, we show that the ABA problem is avoided on the $info$ and $child$ fields of internal nodes.
First, we show that the $info$ field of an internal node is not set to a value that it has had previously.

\begin{lemma} \label{same-info-lem}
Let $x$ be an internal node.
Then, $x.info$ is never set to a value that it has had previously.
\end{lemma}
\begin{proof}
Assume $x.info$ is set to $new$ at time $T$.
We show that $x.info \ne new$ at all times before $T$.
First, consider the case where $new$ is an Unflag object.
By Lemma \ref{info-lem}, only the CAS steps at line \ref{help-unflag} and \ref{help-unflag-bt} change $x.info$ to an Unflag object.
Since, at those lines, the CAS step changes $x.info$ to a newly created Unflag object, $x.info$ is never set at line \ref{help-unflag} or \ref{help-unflag-bt} to a value that it has had previously.

Now, consider the case where $new$ is a Flag object.
By Lemma \ref{info-lem}, only the CAS step at line \ref{help-flag-node} inside help($new$) tries to change $x.info$ to $new$.
To derive a contradiction, assume the lemma is violated for the first time at time $T$.
Then, $x.info$ is set to $new$ at time $T$ and $x.info$ had value $new$ at some time before $T$.
Since $x.info$ is initially an Unflag object when $x$ is created, there must exist a CAS step that sets $x.info$ to $new$ at some time $T'$ before time $T$.
Let $i_1$ and $i_2$ be the values of $i$ when a process executes line \ref{help-flag-node} at $T'$ and $T$.
Then, $new.flag[i_1] = x$ and $new.flag[i_2] = x$.
Since only one copy of duplicate elements of $new.flag$ is kept at line \ref{newFlag-keep-one-copy}, $i_1 = i_2$.
Let $old$ be $new.oldInfo[i_1] = new.oldInfo[i_2]$. 
At time $T'$, $x.info$ is changed from $old$ to $new$.
Then, at time $T$, $x.info$ is changed from $old$ to $new$.
Thus, $x.info$ is set to $old$ again between $T'$ and $T$, contradicting the assumption that the lemma is violated for the first time at time $T$.
\end{proof}

So, by Lemma \ref{same-info-lem}, if a CAS step succeeds to change the $info$ field of an internal node from some old value to some new value, 
the $info$ field of the internal node has not been changed since the time that the value of the $info$ field of the internal node was read as the old value.
Thus, the ABA problem on the $info$ fields of internal nodes is avoided.

In the Patricia trie implementation, update operations might help one another to flag and unflag nodes by calling the help routine at line \ref{newFlag-call-help}.
So, there might be several CAS steps that try to change the $info$ field of some internal node from a value $old$ to some value $new$.
We show only the first CAS step among the CAS steps of this group can succeed to change the $info$ field of the internal node from $old$ to $new$.

\begin{lemma} \label{first-info-cas-lem}
Let $x$ be an internal node.
Assume a group of CAS steps tries to change $x.info$ from $old$ to $new$.
Only the first CAS step in this group might succeed.
\end{lemma}
\begin{proof}
Assume a group of CAS steps each tries to change $x.info$ from $old$ to $new$.
By Lemma \ref{info-lem}, no two CAS steps can try to change $x.info$ to the same Unflag object.
So, $new$ must be a Flag object.
Thus, the CAS steps that try to change $x.info$ are Flag CASs inside calls to help($new$).
So, there is some $i$ such that $new.flag[i] = x$ and $new.oldInfo[i] = old$.
By Lemma \ref{info-old-lem}, $x.info = old$ at some time before $new$ is created.
So, all CAS steps that try to change $x.info$ from $old$ to $new$ occur after a time when the value of $x.info$ is $old$.
After a CAS step changes $x.info$ from $old$ to some value, by Lemma \ref{same-info-lem}, $x.info$ is not set to $old$ again.
If the first CAS step among the CAS steps in the group changes $x.info$ from $old$ to $new$, no other CAS step in the group can change $x.info$ from $old$ to $new$.
If the first CAS step among the CAS steps in the group does not change $x.info$ from $old$ to $new$, $x.info$ has already been changed from $old$ to some other value and 
no other CAS step in the group can change $x.info$ from $old$ to $new$.
\end{proof}

Let $I$ be a Flag object.
Recall that $F_I$ is the set of all nodes in the $I.flag$ array.
To avoid live-lock in flagging internal nodes, all operations try to flag internal nodes in some defined order.
Before creating $I$, elements of $F_I$ are sorted at line \ref{newFlag-sort-nodes}.
We show that if elements of $F_I$ are successfully flagged by flag CAS steps of $I$, they are flagged in order.

\begin{lemma} \label{order-flag-lem}
Let $I$ be a Flag object and let $f$ be the number of of entries in $I.flag$.
For $0 \le i < f-1$, a flag CAS of $I$ does not successfully flag $I.flag[i+1]$ unless $I.flag[i]$ is flagged earlier by a flag CAS of $I$.
\end{lemma}
\begin{proof}
Assume a help routine tries to flag $I.flag[i+1]$ using a flag CAS of $I$.
Then, the $doChildCAS$ variable is true at the previous execution of line \ref{help-set-doChildCAS}.
After the help routine tries to flag a node, $doChildCAS$ is set to false at line \ref{help-set-doChildCAS} if the $info$ field of the node is not $I$.
So, the help routine tries to flag $I.flag[i+1]$ using a flag CAS of $I$ only if $I.flag[i].info$ was $I$.
Thus, by Lemma \ref{info-lem}, $I.flag[i]$ has already been flagged by a flag CAS of $I$.
\end{proof}

If $I.flagDone$ is true at line \ref{help-check-done}, help($I$) performs unflag CAS steps on the nodes in $U_I$ at line \ref{help-unflag}.
Otherwise, help($I$) performs backtrack CAS steps on the nodes in $F_I$ at line \ref{help-unflag-bt}.
The following lemma shows that $I.flagDone$ is set to true at line  \ref{help-set-done} only after all nodes in $F_I$ get flagged successfully by flag CAS steps of $I$.
\begin{lemma} \label{set-done-after-flag-lem}
Let $I$ be a Flag object.
Before $I.flagDone$ is set to true at line  \ref{help-set-done}, all nodes in $F_I$ are flagged by flag CAS steps of $I$.
\end{lemma}
\begin{proof} 
The help routine sets $I.flagDone$ to true at line \ref{help-set-done} only if the $doChildCAS$ variable is true at line \ref{help-check-doChildCAS}.
If the $info$ field of $I.flag[i] \ne I$ for any $i$, $doChildCAS$ gets set to false at line \ref{help-set-doChildCAS} and the loop terminates.
So, $doChildCAS$ is true at line \ref{help-check-doChildCAS} only if all nodes in $F_I$ were flagged successfully using flag CAS steps of $I$.
\end{proof}

From the pseudo-code, we have the following observations.
\begin{observation} \label{done-obs}
Let $I$ be a Flag object.
$I.flagDone$ is initially false and $I.flagDone$ is only set to true.
\end{observation}

\begin{observation} \label{change-child-after-set-done-obs}
Let $I$ be a Flag object.
Then, $I.flagDone$ is set to true before any child CAS step of $I$ occurs.
\end{observation}

We wish to show that the $child$ field of an internal node is changed by a child CAS of $I$ at line \ref{help-change-child} only while all internal nodes in $F_I$ are flagged by $I$. 
First, we show a child CAS of $I$ is executed only after all nodes in $F_I$ get flagged using flag CAS steps of $I$.
Then, we show that, for all $j$, there is no successful unflag or backtrack CAS of $I$ before the first child CAS of $I$ on $I.pNode[j]$.

\begin{lemma} \label{change-child-after-flag-lem}
Let $I$ be a Flag object.
A child CAS step of $I$ is preceded by flagging all nodes in $F_I$ using flag CAS steps of $I$.
\end{lemma}
\begin{proof}
By Observation \ref{change-child-after-set-done-obs}, a child CAS step of $I$ can be executed only after $I.flagDone$ is set to true.
By Lemma \ref{set-done-after-flag-lem}, setting $I.flagDone$ to true is preceded by flagging all nodes in $F_I$ using flag CAS steps of $I$.
\end{proof}

\begin{lemma} \label{no-unflag-before-first-change-child-lem}
Let $I$ be a Flag object.
If there is any child CAS of $I$, then for all $j$, there is no unflag or backtrack CAS of $I$ before the first child CAS step of $I$ on $I.pNode[j]$.
\end{lemma}
\begin{proof}
We prove the lemma by contradiction.
Assume that for some $j$, the first child CAS of $I$ on $I.pNode[j]$ is at time $T$ and an invocation $H$ of help($I$) executes the first unflag or backtrack CAS of $I$ at $T'$ before $T$.
Since the first child CAS of $I$ on $I.pNode[j]$ is at $T$, $H$ does not execute any child CAS of $I$ on $I.pNode[j]$ at line \ref{help-change-child} before $T'$.
So, the $doChildCAS$ variable is false when $H$ performs line \ref{help-check-doChildCAS} before $T'$.
Thus, $H$ sets the $doChildCAS$ variable to false at line \ref{help-set-doChildCAS} when $y.info \ne I$ just after $H$ tries to flag some node $y$ by a flag CAS of $I$.
By Lemma \ref{change-child-after-flag-lem}, since $y \in F_I$, $y.info$ is set to $I$ before $T$.
By Lemma \ref{first-info-cas-lem}, only the first flag CAS of $I$ on $y.info$ succeeds.
Since a flag CAS of $I$ on $y.info$ is performed just before the line \ref{help-set-doChildCAS}, $y.info$ is set to $I$ before $H$ executes line \ref{help-set-doChildCAS}.
Then, it is changed from $I$ to another value before $H$ reads $y.info$ at line \ref{help-set-doChildCAS},
contradicting the fact that the first unflag or backtrack CAS of $I$ is at time $T'$.
\end{proof}

Next, we show that there is no backtrack CAS of $I$ after the first child CAS of $I$. 
This will imply that nodes in $F_I - U_I$ remain flagged forever after a child CAS of $I$ occurs.

\begin{lemma} \label{no-backtrack-after-first-change-child-lem} 
Let $I$ be a Flag object.
If there is any child CAS of $I$, there is no backtrack CAS of $I$.
\end{lemma}
\begin{proof}
Assume the first child CAS of $I$ is at time $T$.
By Lemma \ref{no-unflag-before-first-change-child-lem}, there is no backtrack CAS of $I$ before $T$.
To derive a contradiction, assume the first backtrack CAS of $I$ after $T$ is performed by some invocation $H$ of help($I$).

Then, when $H$ checks $I.flagDone$ at line \ref{help-check-done}, $I.flagDone$ is false.
By Observation \ref{change-child-after-set-done-obs}, $I.flagDone$ is set to true before $T$, so $H$ performs line \ref{help-check-done} before $T$.
Since $H$ does not set $I.flagDone$ to true at line \ref{help-set-done}, $H$ must have set $doChildCAS$ to false at line \ref{help-set-doChildCAS} 
after seeing $y.info \ne I$ for some node $y \in F_I$ at line \ref{help-set-doChildCAS}.
By Lemma \ref{change-child-after-flag-lem}, $y.info$ is set to $I$ before $T$.
By Lemma \ref{first-info-cas-lem}, the first flag CAS of $I$ on $y.info$ succeeds.
So, $y.info$ is changed from $I$ to another value before $H$ reads $y.info$ at line \ref{help-set-doChildCAS}, which is prior to $T$ (since $H$ performs line \ref{help-check-done} before $T$).
This contradicts Lemma \ref{no-unflag-before-first-change-child-lem}.
\end{proof}

By the pseudo-code, we have the following lemma.
\begin{observation} \label{flag-parent-obs}
Let $I$ be a Flag object.
For each $i$, $I.pNode[i] \in F_I$ and $I.pNode[i] \in U_I$.
\end{observation}

Let $I$ be a Flag object and $x \in F_I$.
The following lemma shows what the value of $x.info$ is between the time when $x.info$ is read during the search preceding the creation of $I$ and the successful child CAS of $I$.
\begin{lemma}\label{update-x-info-lem}
Let $I$ be a Flag object and $x = I.flag[i]$ and $xInfo = I.oldInfo[i]$ for some $i$.
Suppose $x.info = xInfo$ at time $T_1$ and for some $j$ the first child CAS of $I$ on $I.pNode[j]$ occurs at time $T_2 > T_1$.
Then, at all times between $T_1$ and $T_2$, $x.info$ is either $xInfo$ or $I$.
\end{lemma}
\begin{proof}
By Lemma \ref{change-child-after-flag-lem}, $x.info$ is changed from $xInfo$ to $I$ between $T_1$ and $T_2$.
By Lemma \ref{same-info-lem}, $x.info = xInfo$ at all times between $T_1$ and the time when $x.info$ is set to $I$.
By Lemma \ref{same-info-lem} and \ref{no-unflag-before-first-change-child-lem}, $x.Info = I$ at all times between the time when $x.info$ is set to $I$ and $T_2$.
\end{proof}

Let $I$ be a Flag object that is created by a replace operation.
Then, the leaf node $I.rmvLeaf$ might be flagged at line \ref{help-flag-leaf} during help($I$). 
By Lemma \ref{set-done-after-flag-lem} and the pseudo-code, we have the following observation.

\begin{observation} \label{flag-leaf-move-obs}
Let $I$ be a Flag object.
Then, $I.rmvLeaf$ is a leaf node if and only if $I$ is created at line \ref{mov-set-info1} or \ref{mov-set-info2}.
For such Flag object $I$, setting $I.rmvLeaf.info$ to $I$ is preceded by flagging all nodes in $F_I$ using flag CASs of $I$ and setting $I.flagDone$ to true.
Furthermore, a child CAS of $I$ is preceded by setting $I.rmvLeaf.info$ to $I$.
After $I.rmvLeaf$ is flagged, it never becomes unflagged.
\end{observation}

\subsection{Behaviour of CAS Steps on $child$ Fields}
In this section, we show how CAS steps change the $child$ field of nodes.

Just before creating a Flag object $I$, an update operation calls search once (for insert and delete) or twice (for replace).
Then, for each $i$,  $I.pNode[i]$ and $I.oldChild[i]$ are set to either $gp$ and $p$ or $p$ and $node$ 
where $\langle gp$ , $p$, $node$, -, -, -$\rangle$ is the result of one of these searches.
So, by Lemma \ref{search-lem}, we have the following observation.

\begin{observation} \label{info-p-pos-old-val-obs}
Let $I$ be a Flag object.
The values of $I.pNode[i]$ and $I.oldChild[i]$ are returned by a  call to the search operation that precedes the creation of $I$.
At some point during that search operation, $I.oldChild[i]$ was a child of $I.pNode[i]$. 
\end{observation}

Next, we show that each successful child CAS changes the $child$ field of an internal node from some old value to some new value that is different from the old value.

\begin{lemma} \label{old-new-different-lem}
Let $I$ be a Flag object.
Then for all $i$, $I.oldChild[i] \ne I.newChild[i]$.
\end{lemma}
\begin{proof}
If $i = 0$ and $I$ is created on line \ref{ins-set-info1}, \ref{ins-set-info2}, \ref{mov-set-info1}, \ref{mov-set-info2}, \ref{mov-set-info-sc1}, \ref{mov-set-info-sc23} or \ref{mov-set-info-sc4}
then $I.newChild[i]$ is a newly created node, $I.oldChild[i]$ is set to $node$ or $p$
where $\langle -$ , $p$, $node$, -, -, - $\rangle$ was returned by a search operation that preceded the creation of $I$.
By Lemma \ref{search-lem}, $I.oldChild[i]$ was a child of some internal node during the search operation.
$I.newChild[i]$ is created at line \ref{create-internal} or \ref{mov-set-info-sc1} after the search operation returns.
So, $I.oldChild[i] \ne I.newChild[i]$.

If $i=0$ and $I$ is created on line \ref{del-set-info} or if $i=1$ and $I$ is created at line \ref{mov-set-info1} or \ref{mov-set-info2}, then $I.oldChild$ is set to $p$
and the value of $I.newChild$ is read from $p.child$ at line \ref{del-read-sibling} or \ref{mov-read-sibling}.
By Invariant \ref{prefix-inv}, $I.oldChild[i] \ne I.newChild[i]$.
\end{proof}

The decision of which child to update is made at line \ref{help-read-index} based on the $(|I.pNode[i].label|+1)$th bit of $I.newChild[i].label$.
Now, we show the child CAS changes the same element of the $child$ field that is read during the search operation.
\begin{lemma} \label{same-index-lem}
Let $I$ be a Flag object.
Then, for each $i$, $(|I.pNode[i].label|+1)$th bit of $I.oldChild[i].label$ and $I.newChild[i].label$ are the same.
\end{lemma}
\begin{proof}
Let $k$ be $|I.pNode[i].label|+1$.
To prove the lemma, we consider different cases according to what line created $I$.

Case 1: $I$ is created at line \ref{ins-set-info1} or \ref{ins-set-info2}.
Let $\langle$ -, $p$, $node$, -, -, -$\rangle$ be the result returned by the call to search($val$) on line \ref{ins-call-search} that precedes the creation of $I$.
Let $newNode$ be the new internal node that is created at line \ref{ins-call-create-node}.
Then, $I.pNode[0] = p$, $I.oldChild[0] = node$ and $I.newChild[0] = newNode$.
By Lemma \ref{search-lem}, $p.child[j] = node$ for some $j$ at some time during search($val$) and $(p.label) \cdot j$ is a prefix of $val$.
By Invariant \ref{prefix-inv}, $(p.label) \cdot j$ is also a prefix of $node.label$.
So, the $k$th bits of $node.label$ and $val$ are $j$.
Initially, the children of $newNode$ are a new leaf node whose $label$ is $val$ and a new copy of $node$.
Since $(p.label) \cdot j$ is a prefix of $node.label$ and $val$, the $k$th bit of $newNode.label$ is $j$.

Case 2: $I$ is created at line \ref{del-set-info}.
Let $\langle gp$, $p$, $node$, -, -, -$\rangle$ be the result returned by the call to search($val$) on line \ref{del-call-search} that precedes the creation of $I$.
Let $nodeSibling$ be the element of $p.child$ that is read at line \ref{del-read-sibling}.
Then, $I.pNode[0] = gp$, $I.oldChild[0] = p$ and $I.newChild[0] = nodeSibling$.
By Lemma \ref{search-lem}, $gp.child[j] = p$ for some $j$ at some time during search($val$). 
By Invariant \ref{prefix-inv}, $(gp.label) \cdot j$ is a prefix of $p.label$ and $p.label$ is a prefix of $nodeSibling.label$.
So, $k$th bits of $p.label$ and $nodeSibling.label$ are $j$.

Case 3: $I$ is created at line \ref{mov-set-info1} or \ref{mov-set-info2}.
Let $\langle gp_d$, $p_d$, $node_d$, -, -, -$\rangle$ be the result returned by the call to search($val_d$) on line \ref{mov-call-search1} that precedes the creation of $I$ and
$\langle$ -, $p_i$, $node_i$, -, -, -$\rangle$ be the result returned by the call to search($val_i$) on line \ref{mov-call-search2} that precedes the creation of $I$.
Then, $newNode_i$ is the new internal node that is created at line \ref{mov-call-create-node}.
Then, $I.pNode[0] = p_i$, $I.oldChild[0] = node_i$ and $I.newChild[0] = newNode_i$.
By the same argument as in Case 1, the lemma is true for $k$th bit of $node_i.label$ and $newNode_i.label$.
Let $nodeSibling_d$ be an element of $p_d.child$ that is read at line \ref{del-read-sibling}.
Then, $I.pNode[0] = gp_d$, $I.oldChild[0] = p_d$ and $I.newChild[0] = nodeSibling_d$.
By the same argument as in Case 2, the lemma is true for $k$th bit of $p_d.label$ and $nodeSibling_d.label$.

Case 4: $I$ is created at line \ref{mov-set-info-sc1}.
Let $\langle$ -, $p_i$, $node_i$, -, -, -$\rangle$ be the result returned by the call to search($val_i$) on line \ref{mov-call-search2} that precedes the creation of $I$.
Then, $I.pNode[0] = p_i$, $I.oldChild[0] = node_i$ and $I.newChild[0]$ is the new leaf node whose $label$ is $val_i$.
By Lemma \ref{search-lem}, $p_i.child[j] = node_i$ for some $j$ at some time during search($val_i$) and $(p_i.label) \cdot j$ is a prefix of $val_i$.
By Invariant \ref{prefix-inv}, $(p_i.label) \cdot j$ is a prefix of $node_i.label$.

Case 5: $I$ is created at line \ref{mov-set-info-sc23}. 
Let $\langle gp_d$, $p_d$, $node_d$, -, -, -$\rangle$ be the result returned by the call to search($val_d$) on line \ref{mov-call-search1} that precedes the creation of $I$ and
$\langle$ -, $p_i$, $node_i$, -, -, -$\rangle$ be the result returned by the call to search($val_i$) on line \ref{mov-call-search2} that precedes the creation of $I$.
Then, $newNode_i$ is the new internal node that is created at line \ref{mov-create-node1}.
Then, $I.pNode[0] = gp_d$, $I.oldChild[0] = p_d$ and $I.newChild[0] = newNode_i$.
Let $nodeSibling_d$ be the child of $p_d$ that is read at line \ref{mov-read-sibling}.
By Lemma \ref{search-lem}, $gp_d.child[j] = p_d$ for some $j$ at some time during search($val_d$).
By Invariant \ref{prefix-inv}, $(gp_d.label) \cdot j$ is a prefix of $p_d.label$ and $p_d.label$ is a prefix of $nodeSibling_d.label$.
So, $(gp_d.label) \cdot j$ is a prefix of $nodeSibling_d.label$ and the $k$th bit of $p_d.label$ is $j$.

If the condition at line \ref{mov-condition-sc2} is true, by Lemma \ref{search-lem}, $(gp_d.label) \cdot j$ is a prefix of $val_i$ (since $p_i = gp_d$ and $node_i = p_d$).
If the condition at line \ref{mov-condition-sc3} is true, by Lemma \ref{search-lem}, $p_i.label$ is a prefix of $val_i$, so  $(gp_d.label) \cdot j$ is a prefix of $val_i$ (since $p_i = p_d$).
Initially, the children of $newNode_i$ are the new leaf node whose $label$ is $val_i$ and $nodeSibling_d$.
Since $(gp_d.label) \cdot j$ is a prefix of $nodeSibling_d.label$ and $val_i$, the $k$th bit of $newNode_i.label$ is $j$.

Case 6: $I$ is created at line \ref{mov-set-info-sc4}.
Let $\langle gp_d$, $p_d$, $node_d$, -, -, -$\rangle$ be the result returned by the call to search($val_d$) on line \ref{mov-call-search1} that precedes the creation of $I$ and
$\langle$ -, $p_i$, $node_i$, -, -, -$\rangle$ be the result returned by the call to search($val_i$) on line \ref{mov-call-search2} that precedes the creation of $I$.
Then, $newChild_i$ is the new internal node that is created at line \ref{mov-create-node2} and $newNode_i$ is the new internal node that is created at line \ref{mov-create-node3}.
Then, $node_i = gp_d$, $I.pNode[0] = p_i$, $I.oldChild[0] = node_i$ and $I.newChild[0] = newNode_i$.
Let $nodeSibling_d$ be the child of $p_d$ that is read at line \ref{mov-read-sibling} and $pSibling_d$ be a child of $gp_d$ that is read at line \ref{mov-read-pSibling}.
By Lemma \ref{search-lem}, $p_i.child[j] = node_i = gp_d$ for some $j$ at some time during search($val_i$). 
By Invariant \ref{prefix-inv}, $(p_i.label) \cdot j$ is a prefix of $gp_d.label$, $gp_d.label$ is a prefix of $pSibling_d.label$ and of $p_d.label$ and $p_d.label$ is a prefix of $nodeSibling_d$.
So, $(p_i.label) \cdot j$ is a prefix of $pSibling_d.label$ and $nodeSibling_d.label$, and the $k$th bit of $node_i.label$ is $j$.

Initially, the children of $newChild_i$ are $pSibling_d$ and $nodeSibling_d$.
Since $(p_i.label) \cdot j$ is a prefix of $pSibling_d.label$ and $nodeSibling_d.label$, $(p_i.label) \cdot j$ is a prefix of $newChild_i.label$.
By Lemma \ref{search-lem}, $(p_i.label) \cdot j$ is a prefix of $val_i$.
Initially, the children of $newNode_i$ are the new leaf node whose $label$ is $val_i$ and $newChild_i$.
Since $(p_i.label) \cdot j$ is a prefix of $newChild_i.label$ and $val_i$, the $k$th bit of $newNode_i.label$ is $j$.
\end{proof}

In the implementation, update operations might help one another to change the $child$ fields of nodes by calling the help routine.
So, there might be several CAS steps that try to change an element of the $child$ field of some internal node from a value $old$ to some value $new$.
The following lemmas show that, as long as there is no ABA problem on the $child$ field of an internal node, 
the first child CAS of $I$ on $I.pNode[i]$ succeeds and no subsequent ones succeed.

Later we shall use these facts to prove inductively that there is no ABA problem on $child$ fields. 

\begin{lemma} \label{other-child-cas-if-no-aba-lem}
Assume that before some time $T$, for all internal nodes $y$, no child CAS sets $y.child[k]$ to a value $old$ after a child CAS of the form CAS($y.child[k]$, $old$, -) succeeds. 

Let $I$ be a Flag object.
Then, for each $i$, only the first child CAS of $I$ on $I.pNode[i]$ can succeed at~$T$ or before~$T$. 
\end{lemma}
\begin{proof}
Let $x = I.pNode[i]$ and $old = I.oldChild[i]$. 
Then, the child CAS of $I$ on $I.pNode[i]$ is of the form CAS($x.child[j]$, $old$, -) for some $j$.
Let $T_{1c}$ be the time when the first child CAS of $I$ on $I.pNode[i]$ occurs.
Suppose some other child CAS of $I$ on $I.pNode[i]$ occurs at $T_c$ where $T_{1c} < T_c \le T$.
We shall show the child CAS at $T_c$ fails.

By Observation \ref{info-p-pos-old-val-obs}, $x.child[j]$ was $old$ at some time before $I$ was created (before $T_{1c}$).
If $x.child[j] = old$ immediately before $T_{1c}$, then the child CAS at $T_{1c}$ changes $x.child[j]$ from $old$ to a different value (by Lemma \ref{old-new-different-lem}).
Thus, just after $T_{1c}$, $x.child[j] \ne old$.
So, at $T_{1c}$ or some time before $T_{1c}$, a child CAS of the form CAS($x.child[j]$, $old$, -) succeeds.
By the assumption, $x.child[j]$ is not changed back to $old$ between $T_{1c}$ and $T$, so the child CAS at time $T_c$ fails.
\end{proof}

Let $I$ be a Flag object. 
The following lemma shows that if the first child CAS of $I$ on $I.pNode[i]$ for some $i$ occurs, no other child CAS changes the $child$ fields of nodes in $F_I$ 
between the time when the $info$ fields of the nodes in $F_I$ are read for the last time before $I$ is created and the first child CAS of $I$ on $I.pNode[i]$.

\begin{lemma} \label{no-wrong-child-change-lem}
Assume that before some time $T$, for all internal nodes $y$, no child CAS sets $y.child[k]$ to a value $old$ after a child CAS of the form CAS($y.child[k]$, $old$, -) succeeds. 

Let $I$ be a Flag object, $x = I.flag[i]$ and $xInfo = I.oldInfo[i]$ for some $i$.
Suppose $x.info = xInfo$ at time $T_1$ and for some $j$, the first child CAS of $I$ on $I.pNode[j]$ occurs at time $T_2 > T_1$ and $T_2 < T$.
Then, no child CAS of any other Flag object $I' \ne I$ changes $x.child$ between $T_1$ and $T_2$.
\end{lemma}
\begin{proof}
By Lemma \ref{other-child-cas-if-no-aba-lem}, for each $k$, only the first child CAS of Flag object $I'$ on $I'.pNode[k]$ can succeed.
Thus, if a child CAS of $I'$ changes $x.child$, by Lemma \ref{change-child-after-flag-lem} and \ref{no-unflag-before-first-change-child-lem}, 
$x.info = I'$ just before the CAS (since $x \in F_{I'}$ by Observation \ref{flag-parent-obs}).
By Lemma \ref{update-x-info-lem}, $x.info$ is either $xInfo$ or $I$ at all times between $T_1$ and $T_2$.
By Lemma \ref{info-lem}, $xInfo$ is an Unflag object.
Since $I \ne I'$, $x.info$ is never equal to $I'$ between $T_1$ and $T_2$, so no child CAS of $I'$ changes $x.child$ between $T_1$ and $T_2$.
\end{proof}

\begin{lemma} \label{first-child-cas-if-no-aba-lem}
Assume that before some time $T$, for all internal nodes $y$, no child CAS sets $y.child[k]$ to a value $old$ after a child CAS of the form CAS($y.child[k]$, $old$, -) succeeds. 

Let $I$ be a Flag object.
Then, for each $i$, if the first child CAS of $I$ on $I.pNode[i]$ occurs before $T$, then it succeeds.
\end{lemma}
\begin{proof}
Let $x = I.pNode[i]$, $old = I.oldChild[i]$. 
Then, the child CAS of $I$ on $I.pNode[i]$ is of the form CAS($x.child[j]$, $old$, -).
Let $T_{1c}$ be the time when the first child CAS of $I$ on $I.pNode[i]$ occurs (before $T$).
We prove that it succeeds.

By Observation \ref{flag-parent-obs}, $x = I.pNode[i] \in F_{I}$.
By Lemma \ref{info-old-lem}, $x.info = I.oldInfo[i]$ at time $T_0$ before $I$ is created.
By Observation \ref{info-p-pos-old-val-obs} and Lemma \ref{search-lem}, $x.child[j] = old$ at some time between $T_0$ and $T_{1c}$.
So, to prove the lemma, it suffices to show that no child CAS changes $x.child[j]$ from $old$ to another value between $T_0$ and $T_{1c}$. 

First, we argue that such a change cannot be made by a child CAS of $I$ itself.
This can only happen if $I$ is created at line \ref{mov-set-info1} or \ref{mov-set-info2} and $i=1$ and the child CAS of $I$  at $T_{1c}$ is on $I.pNode[1]$. 
For this case, we must show that no child CAS of $I$ on $I.pNode[0]$ changes $x.child[j]$ from $old$ to another value: 
Since the condition at line \ref{mov-check-gc-condition} preceding the creation of $I$ is true, $I.oldNode[0] \ne I.oldNode[1]$.
If $I.pNode[0] = I.pNode[1]$, the child CAS of $I$ on $I.pNode[0]$ does not change $x.child[j]$ from $old$ to another value (since $I.oldNode[0] \ne old$). 

By Lemma \ref{no-wrong-child-change-lem}, no child CAS of $I' \ne I$ changes $x.child[j]$ between $T_0$ and $T_{1c}$.
\end{proof}

Let $I$ be a Flag object and $x$ be an internal node in $F_I - U_I$.
By Lemma \ref{change-child-after-flag-lem} and \ref{no-unflag-before-first-change-child-lem}, $x.info = I$ just before the first child CAS of $I$.
If there is any child CAS of $I$, by Lemma \ref{change-child-after-flag-lem}, \ref{no-unflag-before-first-change-child-lem} and \ref{no-backtrack-after-first-change-child-lem},
$x.info = I$ at all times after the first child CAS of $I$.
After the first child CAS of $I$, we say $x$ is {\it marked by $I$}.
(Thus, once a node is marked by $I$, it remains marked by $I$ forever.)

Now, we prove if an internal node becomes unreachable, it is marked after that.
\begin{lemma} \label{unreachable-mark-lem}
Assume that before some time $T$, for all internal nodes $y$, no child CAS sets $y.child[k]$ to a value $old$ after a child CAS of the form CAS($y.child[k]$, $old$, -) succeeds. 

Let $I$ be a Flag object and $x$ be an internal node.
Consider a child CAS of $I$ that succeeds at time $T'$ (before $T$).
If $x$ is reachable immediately before $T'$ and $x$ becomes unreachable at $T'$, $x$ is marked by $I$ at all times after $T'$.
\end{lemma}
\begin{proof}
Let $I$ be a Flag object. 
Let $\langle gp$, $p$, $node$, -, -, -$\rangle$ be the result returned by the call to search($val$) on line \ref{ins-call-search}, \ref{del-call-search} or \ref{mov-call-search1} that precedes the creation of $I$.
If $I$ is created inside replace($val$, $val'$), let $\langle$ -, $p'$, $node'$, -, -, -$\rangle$ be the result returned by the call to search($val'$) on line \ref{mov-call-search2} that precedes the creation of $I$.
By Lemma \ref{other-child-cas-if-no-aba-lem}, $T'$ is the first child CAS of $I$ on $I.pNode[i]$ for some $i$.
We consider different cases according to what line created $I$.
For each case, we show that only nodes in $F_I - U_I$ could become unreachable immediately after $T'$.
By Lemma \ref{change-child-after-flag-lem}, \ref{no-unflag-before-first-change-child-lem} and \ref{no-backtrack-after-first-change-child-lem}, any node in $F_I - U_I$ is flagged by $I$ at all times after $T'$.

Case 1: $I$ is created at line \ref{ins-set-info1}.
Then, $node$ is an internal node.
Then, $newNode$ is the new internal node created at line \ref{ins-call-create-node} whose non-empty children are a new leaf node and a new copy of $node$.
In this case, $I.oldChild[0] = node$ and $I.newChild[0] = newNode$.
Since $node \in F_I$, by Lemma \ref{no-wrong-child-change-lem}, 
$node.child$ is not changed between the time when $node.info$ is read at line \ref{ins-read-node-info} and $T'$.
Thus, just before $T'$, the $child$ field of the new copy of $node$ is the same as $node.child$.
The only internal node that could become unreachable at $T'$ is $node$.
In this case, $node \in F_I -U_I$.
By the definition, $node$ is marked at all times after $T'$.

Case 2: $I$ is created at line \ref{ins-set-info2}.
Then, $node$ is a leaf node.
In this case, $I.oldChild[0] = node$ is a leaf node, so no internal node becomes unreachable at $T'$.

Case 3: $I$ is created at line \ref{del-set-info}.
Let $nodeSibling$ be the element of $p.child$ that is read at line \ref{del-read-sibling}.
By Invariant \ref{prefix-inv} and Lemma \ref{child-num-lem}, $nodeSibling$ exists and is different from $node$.
In this case, $I.oldChild[0] = p$ and $I.newChild[0] = nodeSibling$.
Since $p \in F_I$, by Lemma \ref{no-wrong-child-change-lem}, $p.child$ is not changed between the time when $p.info$ is read for the last time during search($val$) and $T'$.
Since $p$ is a parent of $node$ during search($val$) (by Lemma \ref{search-lem}) and $p$ is a parent of $nodeSibling$ at line \ref{del-read-sibling}, 
the children of $p$ just before $T'$ are $node$ and $nodeSibling$.
Since the operation does not return false at line \ref{del-return-false}, $node$ is a leaf node.
So, the only internal node that could become unreachable at $T'$ is $p$.
In this case, $p \in F_I -U_I$. 
By the definition, $node$ is marked at all times after $T'$.

Case 4: $I$ is created at line \ref{mov-set-info1} or \ref{mov-set-info2}.
Then, $newNode_i$ is the new internal node that is created at line \ref{mov-call-create-node}.
If the first child CAS of $I$ on $I.pNode[0]$ occurs at time $T'$, an element of $p'.child$ is changed from $node'$ to $newNode_i$ at $T'$.
By the same argument as in Case 1 and 2, if $node'$ is an internal node, $node'$ is the only internal node that could become unreachable at $T'$ and $node'$ is marked at all times after $T'$.

Let $nodeSibling_d$ be an element of $p.child$ that is read at line \ref{mov-read-sibling}.
If the first child CAS of $I$ on $I.pNode[1]$ occurs at time $T'$, an element of $gp.child$ is changed from $p$ to $nodeSibling_d$ at $T'$.
By the same argument as in Case 3, $p$ is the only internal node that could become unreachable at $T'$ and $p$ is marked at all times after $T'$.

Case 5: $I$ is created at line \ref{mov-set-info-sc1}.
Then,  $I.oldChild[0] = node'$.
Since the operation does not return false at line \ref{mov-return-false1}, $node'$ is a leaf node.
Since $I.oldChild[0]$ is a leaf node, no internal node becomes unreachable at $T'$.

Case 6: $I$ is created at line \ref{mov-set-info-sc23}. 
Let $nodeSibling_d$ be the element of $p.child$ that is read at line \ref{mov-read-sibling}. 
Then, $newNode_i$ is the new internal node created at line \ref{mov-create-node1} whose children are a new leaf node and $nodeSibling_d$.
By Invariant \ref{prefix-inv} and Lemma \ref{child-num-lem}, $nodeSibling_d$ exists and is not equal to $node$.
In this case, $I.oldChild[0] = p$ and $I.newChild[0] = newNode_i$.
By Lemma \ref{no-wrong-child-change-lem}, since $p \in F_I$, $p.child$ is not changed between the time when $p.info$ is read for the last time during search($val$) and $T'$.
Since $node \ne nodeSibling$, by Lemma \ref{child-num-lem}, the children of $p$ are $node$ and $nodeSibling$ just before $T'$. 
Since the operation does not return false at line \ref{mov-return-false1}, $node$ is a leaf node.
So, the only internal node that could become unreachable at $T'$ is $p$.
In this case, $p \in F_I -U_I$.
By the definition, $p$ is marked at all times after $T'$.

Case 7: $I$ is created at line \ref{mov-set-info-sc4}.
Then, $gp = node'$.
Let $nodeSibling_d$ be the element of $p.child$ that is read at line \ref{mov-read-sibling} and 
$pSibling_d$ be the element of $gp.child$ that is read at line \ref{mov-read-pSibling}. 
Then, $newChild_i$ is the new internal node created at line \ref{mov-create-node2} whose children are $nodeSibling_d$ and $pSibling_d$ and
$newNode_i$ is the new internal node created at line \ref{mov-create-node3} whose children are a new leaf node and $newChild_i$. 
In this case, $I.oldChild[0] = gp$ and $I.newChild[0] = newNode_i$.
In this case, $gp \in F_I$ and $p \in F_I$.
By Lemma \ref{no-wrong-child-change-lem}, $gp.child$ is not changed between the time when $gp.info$ is read for the last time during search($val$) and $T'$.
By Lemma \ref{no-wrong-child-change-lem}, $p.child$ is not changed between the time when $p.info$ is read for the last time during search($val$) and $T'$.
So, just before $T'$, the children of $gp$ are $p$ and $pSibling_d$ and the children of $p$ are $node$ and $nodeSibling_d$.
Since the operation does not return false at line \ref{mov-return-false1}, $node$ is a leaf node.
So, the only internal nodes that become unreachable at $T'$ are $gp = node'$ and $p$.
In this case, $gp = node'$ and $p$ are in $F_I - U_I$.
By the definition, $p$ and $gp$ are marked at all times after $T'$.
\end{proof}

The following corollary shows if an internal node is not marked, it has not become unreachable.

\begin{corollary}\label{unmarked-reachable-col}
Assume that before some time $T$, for all internal nodes $y$, no child CAS sets $y.child[k]$ to a value $old$ after a child CAS of the form CAS($y.child[k]$, $old$, -) succeeds. 

If an internal node $x$ is not marked at time $T'$ (before $T$) and $x$ was reachable at some time before $T'$, then $x$ is reachable at time $T'$.
\end{corollary}

Now, we show that, for each $i$, $I.oldChild[i]$ becomes unreachable just after the successful child CAS of $I$ on $I.pNode[i]$.
\begin{lemma} \label{unreachable-oldChild-lem}
Assume that before some time $T$, for all internal nodes $y$, no child CAS sets $y.child[k]$ to a value $old$ after a child CAS of the form CAS($y.child[k]$, $old$, -) succeeds. 

Let $I$ be a Flag object. 
If the child CAS of $I$ on $I.pNode[i]$ succeeds at time $T'$ (before $T$) for some $i$, $I.oldChild[i]$ becomes unreachable just after $T'$.
\end{lemma}
\begin{proof}
Let $I$ be a Flag object. 
Let $\langle gp$, $p$, $node$, -, -, -$\rangle$ be the result returned by the call to search($val$) on line \ref{ins-call-search}, \ref{del-call-search} or \ref{mov-call-search1} that precedes the creation of $I$.
If $I$ is created inside replace($val$, $val'$), let $\langle$ -, $p'$, $node'$, -, -, -$\rangle$ be the result returned by the call to search($val'$) on line \ref{mov-call-search2} that precedes the creation of $I$.
By Lemma \ref{other-child-cas-if-no-aba-lem}, $T'$ is the first child CAS of $I$ on $I.pNode[i]$. 
We consider different cases according to what line created $I$.

Case 1: $I$ is created at line \ref{ins-set-info1} or \ref{ins-set-info2}.
Let $newNode$ be the new internal node created at line \ref{ins-call-create-node} whose non-empty children are a new leaf node and a new copy of $node$.
By Invariant \ref{prefix-inv}, $node$ is not descendant of the new copy of $node$ and hence is not a descendant of $newNode$.
In this case, $I.oldChild[0] = node$ and $I.newChild[0] = newNode$.
Since $p \in F_I$, $p$ is a parent of $node$ just before $T'$ (by Lemma \ref{no-wrong-child-change-lem}).
Since $p.info = I$ just before $T'$ (by Lemma \ref{change-child-after-flag-lem} and \ref{no-unflag-before-first-change-child-lem}), 
$p$ is not marked just before $T'$ and $p$ is reachable just before $T'$ (by Corollary \ref{unmarked-reachable-col}).
By Lemma \ref{reachable-parent-lem}, $p$ is the only reachable parent of $node$ just before $T'$.
Since $node$ is not a descendant of $newNode$, $node$ becomes unreachable just after $T'$. 

Case 2: $I$ is created at line \ref{del-set-info}.
Let $nodeSibling$ be the element of $p.child$ that is read at line \ref{del-read-sibling}.
By Invariant \ref{prefix-inv} and Lemma \ref{child-num-lem}, $nodeSibling$ exists and is different from $node$.
In this case, $I.oldChild[0] = p$ and $I.newChild[0] = nodeSibling$.
Since $p$ and $gp$ are in $F_I$, $gp$ is a parent of $p$ and $p$ is a parent of $node$ just before $T'$ (by Lemma \ref{no-wrong-child-change-lem}).
Since $gp.info = I$ just before $T'$ (by Lemma \ref{change-child-after-flag-lem} and \ref{no-unflag-before-first-change-child-lem}), 
$gp$ is not marked just before $T'$ and $gp$ is reachable just before $T'$ (by Corollary \ref{unmarked-reachable-col}).
Since $gp$ is the reachable parent of $p$ just before $T'$, $p$ is reachable just before $T'$.
By Lemma \ref{reachable-parent-lem}, $p$ is the only reachable parent of $node$ just before $T'$.
Since $node$ is not a descendant of $nodeSibling$ (by Invariant \ref{prefix-inv}), $node$ becomes unreachable just after $T'$. 

Case 3: $I$ is created at line \ref{mov-set-info1} or \ref{mov-set-info2}.
Let $newNode_i$ be the new internal node that is created at line \ref{mov-call-create-node}.
If the first child CAS of $I$ on $I.pNode[0]$ occurs at time $T'$, an element of $p'.child$ is changed from $node'$ to $newNode_i$ at $T'$.
By the same argument as in Case 1, $node'$ becomes unreachable just after $T'$.

Let $nodeSibling_d$ be an element of $p.child$ that is read at line \ref{mov-read-sibling}.
If the first child CAS of $I$ on $I.pNode[1]$ occurs at time $T'$, an element of $gp.child$ is changed from $p$ to $nodeSibling_d$ at $T'$.
By the same argument as in Case 2, $node$ becomes unreachable just after $T'$.

Case 4: $I$ is created at line \ref{mov-set-info-sc1}.
Then,  $node = node'$, $I.oldChild[0] = node'$ and $I.newChild[0]$ is a new leaf node.
Since $p \in F_I$, $p$ is a parent of $node = node'$ just before $T'$ (by Lemma \ref{search-lem} and \ref{no-wrong-child-change-lem}).
Since $p.info = I$ just before $T'$ (by Lemma \ref{change-child-after-flag-lem} and \ref{no-unflag-before-first-change-child-lem}), $p$ is reachable just before $T'$ (by Corollary \ref{unmarked-reachable-col}).
Since $I.newChild[0]$ is a leaf node and $p$ is the only reachable parent of $node = node'$ just before $T'$ (by Lemma \ref{reachable-parent-lem}), $node'$ becomes unreachable just after $T'$.

Case 5: $I$ is created at line \ref{mov-set-info-sc23}. 
Let $nodeSibling_d$ be the element of $p.child$ that is read at line \ref{mov-read-sibling}. 
Let $newNode_i$ be the new internal node created at line \ref{mov-create-node1} whose children are a new leaf node and $nodeSibling_d$.
By Invariant \ref{prefix-inv} and Lemma \ref{child-num-lem}, $nodeSibling_d$ exists and is not equal to $node$.
In this case, $I.oldChild[0] = p$ and $I.newChild[0] = newNode_i$.
By Invariant \ref{prefix-inv}, $p$ is not a descendant of $nodeSibling_d$ and hence is not a descendant of $newNode_i$.
Since $gp \in F_I$, $gp$ is a parent of $p$ just before $T'$ (by Lemma \ref{no-wrong-child-change-lem}).
Since $gp.info = I$ just before $T'$ (by Lemma \ref{change-child-after-flag-lem} and \ref{no-unflag-before-first-change-child-lem}), 
$gp$ is not marked just before $T'$ and $gp$ is reachable just before $T'$ (by Corollary \ref{unmarked-reachable-col}).
By Lemma \ref{reachable-parent-lem}, $gp$ is the only reachable parent of $p$ just before $T'$.
Since $p$ is not a descendant of $newNode_i$, $p$ becomes unreachable just after $T'$.

Case 6: $I$ is created at line \ref{mov-set-info-sc4}.
Then, $gp = node'$.
Let $nodeSibling_d$ be the element of $p.child$ that is read at line \ref{mov-read-sibling} and 
$pSibling_d$ be the element of $gp.child$ that is read at line \ref{mov-read-pSibling}. 
Let $newChild_i$ be the new internal node created at line \ref{mov-create-node2} whose children are $nodeSibling_d$ and $pSibling_d$.
Let $newNode_i$ be the new internal node created at line \ref{mov-create-node3} whose children are a new leaf node and $newChild_i$. 
In this case, $I.oldChild[0] = gp$ and $I.newChild[0] = newNode_i$.
By Invariant \ref{prefix-inv}, $gp$ is not a descendant of $nodeSibling_d$, $pSibling_d$, $newChild_i$ or $newNode_i$.
Since $p' \in F_I$, $p'$ is a parent of $node' = gp$ just before $T'$ (by Lemma \ref{no-wrong-child-change-lem}).
Since $p'.info = I$ just before $T'$ (by Lemma \ref{change-child-after-flag-lem} and \ref{no-unflag-before-first-change-child-lem}), 
$p'$ is not marked just before $T'$ and $p'$ is reachable just before $T'$ (by Corollary \ref{unmarked-reachable-col}).
By Lemma \ref{reachable-parent-lem}, $p'$ is the only reachable parent of $node' = gp$ just before $T'$.
Since $gp$ is not a descendant of $newNode_i$, $gp$ becomes unreachable just after $T'$. 
\end{proof}

We say that the search operation {\it visits} the node that gets stored in $node$ when the search performs line \ref{search-set-initial-node} or \ref{search-set-node}. 
The following lemma shows that a search operation does not visit a node that was unreachable at all times during the operation. 

\begin{lemma} \label{query-reachable-lem}
Assume that before some time $T$, for all internal nodes $y$, no child CAS sets $y.child[k]$ to a value $old$ after a child CAS of the form CAS($y.child[k]$, $old$, -) succeeds. 

Assume a search operation visits a node $xChild \ne root$ by reading $xChild$ in the $child$ field of some node $pNode$ before $T$.
Then, there is a time after the search operation begins and before the search operation visits $xChild$ that $pNode$ is reachable and a parent of $xChild$.
\end{lemma}
\begin{proof}
We prove the lemma by induction on the number of steps that the search operation has done.
Let $T'$ be a time that the search operation visits a node $xChild$ (before $T$).
We assume the lemma is true for all nodes visited before $T'$.
Now, we show that the lemma is true for $xChild$.
Then, $xChild$ is visited on line \ref{search-set-node} at time $T$. 
First, we show that $pNode$ is visited earlier during the operation.

If $xChild$ is visited at line \ref{search-set-node}, $pNode$ is visited earlier at line \ref{search-set-initial-node} 
or at line \ref{search-set-node} during the previous loop iteration.
Thus, $pNode$ is visited at some time during the search operation before time $T'$.
So, $pNode$ was reachable at some time $T''$ after the search operation begins and before the search operation visited $pNode$.
If $xChild$ was a child of $pNode$ at $T''$, the lemma is proved.

Otherwise, an element of $pNode.child$ is set to $xChild$ at some later time between $T''$ and $T'$.
Thus, a child CAS sets an element of $pNode.child$ to $xChild$ between $T''$ and $T'$.
Let $I$ be a Flag object such that a child CAS of $I$ sets an element of $pNode.child$ to $xChild$ between $T''$ and $T'$.
By Lemma \ref{first-child-cas-if-no-aba-lem}, the first child CAS of $I$ on $pNode$ succeeds between $T''$ and $T'$.
By Observation \ref{flag-parent-obs}, $pNode \in F_I$.
By Corollary \ref{change-child-after-flag-lem} and \ref{no-unflag-before-first-change-child-lem}, $pNode.info = I$ just before the child CAS of $I$ between $T''$ and $T'$. 
By Observation \ref{flag-parent-obs}, $pNode \in U_I$.
So, $pNode$ is not marked just before the child CAS of $I$.
Since $pNode$ was reachable at $T''$, by Corollary \ref{unmarked-reachable-col}, $pNode$ is still reachable just before the child CAS of $I$.
So, $pNode$ is reachable and a parent of $xChild$ just after the child CAS of $I$ between $T''$ and $T'$. 
\end{proof}

Now, we show, after a node becomes unreachable, the node does not become reachable again.
\begin{lemma} \label{no-reachable-after-unreachable-lem}
Assume that before some time $T$, for all internal nodes $y$, no child CAS sets $y.child[k]$ to a value $old$ after a child CAS of the form CAS($y.child[k]$, $old$, -) succeeds. 

Let $x$ be a Node object.
If, after $x$ is reachable, $x$ becomes unreachable at time $T'$ (before $T$), $x$ does not become reachable again between $T'$ and $T$.
\end{lemma}
\begin{proof}
To derive a contradiction, assume the lemma is false.
Let $T''$ (before $T$) be the first time it is violated and $x$ be the highest node in the tree at $T''$ that violates the lemma.
Let $T'$ be the time that $x$ becomes unreachable before $T''$.
So, $x$ becomes reachable when a child CAS is performed at $T''$.
We have two different cases.

Case 1: a child CAS of some Flag object $I$ sets an element of a $child$ field to $x$ at $T''$.
If $I$ and $x$ are created in the same loop iteration of an update operation, $x$ is not reachable at any time before $T''$.
So, $I$ is created at line \ref{del-set-info}, \ref{mov-set-info1} or \ref{mov-set-info2}.
Let $\langle gp_d$, $p_d$, -, -, -, -$\rangle$ be the result returned by the call to the search operation on line \ref{del-call-search} or \ref{mov-call-search1} that precedes the creation of $I$.
If $I$ is created at line \ref{del-set-info}, $I.newChild[0] = x$ and if $I$ is created at line \ref{mov-set-info1} or \ref{mov-set-info2}, $I.newChild[1] = x$.
Then, $x$ is read as a child of $p_d$ at line \ref{del-read-sibling} or \ref{mov-read-sibling} and an element of $gp_d.child$ is set to $x$ at $T''$.
By Observation \ref{flag-parent-obs}, $gp_d \in F_I$ and $gp_d \in U_I$.

By Lemma \ref{first-child-cas-if-no-aba-lem}, the first child CAS of $I$ on $gp_d$ succeeds at $T''$.
By Lemma \ref{change-child-after-flag-lem} and \ref{no-unflag-before-first-change-child-lem}, $gp_d.info = I$ just before $T''$.
Since $gp_d$ is not marked just before $T''$ and $gp_d$ was reachable during the search operation by Lemma \ref{query-reachable-lem}, $gp_d$ is reachable just before $T''$ by Corollary \ref{unmarked-reachable-col}.
Since $p_d$ was a child of $gp_d$ during the search operation (by Lemma \ref{search-lem}), $p_d$ is a child of $gp_d$ just before $T''$ (by Lemma \ref{no-wrong-child-change-lem}).
Since $gp_d$ is reachable just before $T''$, $p_d$ is reachable just before $T''$.
By Lemma \ref{no-wrong-child-change-lem}, since $p_d \in F_I$, $p.child$ is not changed between the time when $p.info$ was read for the last time during the search operation and $T''$.
Since $x$ was a child of $p_d$ at line \ref{del-read-sibling} or \ref{mov-read-sibling}, $x$ is still a child of $p_d$ just before $T''$.
Since $p_d$ is reachable just before $T''$, $x$ is reachable just before $T''$, contradicting the fact that $x$ becomes reachable at $T''$ for the first time after $T'$.

Case 2: the child CAS of some Flag object $I$ sets an element of a $child$ field to a proper ancestor $anc$ of $x$ at $T''$.
Since $x$ is not reachable just before $T''$, $anc$ is also not reachable just before $T''$.
Since $x$ is the highest node that violates the lemma at $T''$, $anc$ is not reachable at any time before $T''$.
So, $anc$ is created at line \ref{create-internal} inside createNode that is called at line 
\ref{ins-call-create-node}, \ref{mov-call-create-node}, \ref{mov-create-node1} or \ref{mov-create-node3} prior to the creation of $I$.
We consider different cases depending on what line calls createNode that created $anc$.

Case 2A: $anc$ is created at line \ref{create-internal}, which is called from line \ref{ins-call-create-node} or \ref{mov-call-create-node}.
Let $\langle-$, $p_i$, $node_i$, -, -, -$\rangle$ be the result returned by the call to the search operation on line \ref{ins-call-search} or \ref{mov-call-search2} that precedes the creation of $I$.
The children of $anc$ initially are a new copy of $node_i$ that is created at line \ref{create-copy} or \ref{mov-create-copy} and the leaf node that is created at line \ref{ins-call-create-node} or \ref{mov-call-create-node}.
Since $x$ was reachable before $T'$, $x$ is not created at line \ref{create-copy}, \ref{ins-call-create-node}, \ref{mov-create-copy} or \ref{mov-call-create-node} that precedes the creation of $I$.
So, $x$ is not a child of $anc$ just before $T''$.
Since $anc$ is an ancestor of $x$ just before $T''$, a child of $anc$ is an internal node just before $T''$.
So, $node_i$ is an internal node and $x$ is a descendant of a new copy of $node_i$ just before $T''$.
So, $I$ is created at line \ref{ins-set-info1} or \ref{mov-set-info1}.
Since $node_i \in F_I$, $node_i.child$ is equal to the $child$ field of the new copy of $node_i$ just before $T''$ (by Lemma \ref{no-wrong-child-change-lem}). 
So, $x$ is also a proper descendant of $node_i$ just before $T''$.
Let $nodeChild$ be a child of $node_i$ at line \ref{create-copy} or \ref{mov-create-copy} such that $x$ is a descendant of $nodeChild$ just before $T''$.

By Observation \ref{flag-parent-obs}, $p_i \in F_I$. 
Since $p_i.info = I$ just before $T''$ (by Lemma \ref{change-child-after-flag-lem} and  \ref{no-unflag-before-first-change-child-lem}), $p_i$ is not marked just before $T''$.
Since $p_i$ was reachable during the search operation (by Lemma \ref{query-reachable-lem}), $p_i$ is reachable just before $T''$ by Corollary  \ref{unmarked-reachable-col}.
By Lemma \ref{search-lem}, $p_i$ is a parent of $node_i$ at some time during the search operation before $T''$.
Since $p_i$ is a parent of $node_i$ just before $T''$ (by Lemma \ref{no-wrong-child-change-lem}) and $p_i$ is reachable just before $T''$, $node_i$ is reachable just before $T''$.
Since $node_i.child$ is not changed between line \ref{create-copy} or \ref{mov-create-copy} and $T''$ (by Lemma \ref{no-wrong-child-change-lem}), $nodeChild$ is a child of $node_i$ just before $T''$.
Since $node_i$ is reachable just before $T''$, $nodeChild$ is reachable just before $T''$.
Since $x$ is a descendant of $nodeChild$ just before $T''$, $x$ is reachable just before $T''$, contradicting that $x$ becomes reachable at $T''$ for the first time after $T'$.

Case 2B: $anc$ is created at line \ref{create-internal}, which is called from line \ref{mov-create-node1}.
Then, $I$ is created at line \ref{mov-set-info-sc23}.
Let $\langle gp_d$, $p_d$, $node_d$, -, -, -$\rangle$ be the result returned by the call to search($val_d$) on line \ref{mov-call-search1} that precedes the creation of $I$.
Let $nodeSibling_d$ be the node that is read from $p_d.child$ at line \ref{mov-read-sibling}.
The children of $anc$ initially are $nodeSibling_d$ and the leaf node that is created at line \ref{mov-create-node1}.
Since $x$ was reachable before $T'$, $x$ is not created at line \ref{mov-create-node1} that precedes the creation of $I$.
Since $anc$ is an ancestor of $x$ just before $T''$, $x$ is a descendant of $nodeSibling_d$ just before $T''$.

In this case, $gp_d \in F_I$ and $p_d \in F_I$.
By Observation \ref{flag-parent-obs}, $gp_d \in U_I$.
Since $gp_d.info = I$ just before $T''$ (by Lemma \ref{change-child-after-flag-lem} and  \ref{no-unflag-before-first-change-child-lem}), $gp_d$ is not marked just before $T''$.
Since $gp_d$ was reachable during the search operation (by Lemma \ref{query-reachable-lem}), $gp_d$ is reachable just before $T''$ (by Corollary  \ref{unmarked-reachable-col}).
By Lemma \ref{search-lem}, $gp_d$ is a parent of $p_d$ at some time during the search operation.
So, $gp_d$ is a parent of $p_d$ just before $T''$ (by Lemma \ref{no-wrong-child-change-lem}).
Since $gp_d$ is reachable just before $T''$, $p_d$ is reachable just before $T''$.
By Lemma \ref{no-wrong-child-change-lem}, $nodeSibling_d$ is a child of $p_d$ just before $T''$.
Since $p_d$ is reachable just before $T''$, $nodeSibling_d$ is reachable just before $T''$.
Since $x$ is a descendant of $nodeSibling_d$ just before $T''$, $x$ is reachable just before $T''$, contradicting that $x$ becomes reachable at $T''$ for the first time after $T'$.

Case 2C: $anc$ is created at line \ref{mov-create-node3}.
Then, $I$ is created at line \ref{mov-set-info-sc4}.
Let $\langle gp_d$, $p_d$, $node_d$, -, -, -$\rangle$ be the result returned by the call to search($val_d$) on line \ref{mov-call-search1} that precedes the creation of $I$ and
$\langle$ -, $p_i$, $node_i$, -, -, -$\rangle$ be the result returned by the call to search($val_i$) on line \ref{mov-call-search2} that precedes the creation of $I$.
Let $nodeSibling_d$ be the child of $p_d$ that is read at line \ref{mov-read-sibling} and $pSibling_d$ be the child of $gp_d$ that is read at line \ref{mov-read-pSibling}.
Let $newChild_i$ be the internal node that is created at line \ref{mov-create-node2}.
Initially, the children of $newChild_i$ are $nodeSibling_d$ and $pSibling_d$.
The children of $anc$ initially are $newChild_i$ and the leaf node that is created at line \ref{mov-create-node3}.
Since $x$ was reachable before $T'$, $x$ is not created at line \ref{mov-create-node2} or \ref{mov-create-node3} that precedes the creation of $I$.
Since $anc$ is an ancestor of $x$ just before $T''$ and $x$ is not created at line \ref{mov-create-node2} or \ref{mov-create-node3}, 
$x$ is a descendant of $nodeSibling_d$ or of $pSibling_d$ just before $T''$.

In this case, $node_i = gp_d$.
Since $I$ is created at line \ref{mov-set-info-sc4}, by the pseudo-code, $F_I = \{ p_i, gp_d, p_d\}$. 
Since $p_i.info = I$ just before $T''$ (by Lemma \ref{change-child-after-flag-lem} and  \ref{no-unflag-before-first-change-child-lem}), $p_i$ is not marked just before $T''$.
Since $p_i$ was reachable during the search operation (by Lemma \ref{query-reachable-lem}), $p_i$ is reachable just before $T''$ by Corollary  \ref{unmarked-reachable-col}.
By Lemma \ref{search-lem} and \ref{no-wrong-child-change-lem}, $p_i$ is a parent of $node_i = gp_d$ just before $T''$ and $gp_d$ is a parent of $p_d$ just before $T''$.
Since $p_i$ is reachable just before $T''$, $gp_d$ and $p_d$ are reachable just before $T''$.
By Lemma \ref{no-wrong-child-change-lem}, $nodeSibling_d$ is a child of $p_d$ just before $T''$ and $pSibling_d$ is a child of $gp_d$ just before $T''$.
Since $p_d$ and $gp_d$ are reachable just before $T''$, $nodeSibling_d$ and $pSibling_d$ are reachable just before $T''$.
Since $x$ is a descendant of $nodeSibling_d$ or of $pSibling_d$ just before $T''$, $x$ is reachable just before $T''$, contradicting that $x$ becomes reachable at $T''$ for the first time after $T'$.
\end{proof}

Now, we show that the ABA problem on the $child$ field of internal nodes is avoided.
\begin{lemma} \label{same-child-lem}
Let $x$ and $y$ be internal nodes.
After a child CAS of the form CAS($x.child[i]$, $old$, -) succeeds, no child CAS sets $y.child[j]$ to $old$.
\end{lemma}
\begin{proof}
To derive a contradiction, assume the lemma is violated for the first time at $T$.
Let $I$ be a Flag object such that the first child CAS of $I$ of the form CAS($x.child[i]$, -, $old$) succeeds at $T$ and 
let $I'$ be a Flag object such that the first child CAS of $I'$ of the form CAS($y.child[j]$, $old$, -) succeeds at some earlier time $T'$.

Since the child CAS of $I'$ changes $y.child[j]$ from $old$ to other value at $T'$, $y.child[j] = old$ just before $T'$.
By Observation \ref{flag-parent-obs}, $y \in F_I$.
Since $y.info = I'$ just before $T'$ (by Lemma \ref{change-child-after-flag-lem} and  \ref{no-unflag-before-first-change-child-lem}), $y$ is not marked just before $T'$.
By Lemma \ref{query-reachable-lem}, $y$ was reachable at some earlier time.
So, $y$ is reachable just before $T'$ (by Corollary \ref{unmarked-reachable-col}).
Since $old$ is a child $y$ just before $T'$, $old$ is reachable just before $T'$.

If $old$ is created at line \ref{ins-call-create-node}, \ref{mov-call-create-node}, \ref{mov-set-info-sc1}, \ref{mov-create-node1} or \ref{mov-create-node3} that precedes the creation of $I$, 
$old$ could not become reachable before the child CAS of $I$ at $T$.
So, $old$ and $I$ are not created in the same loop iteration of an update operation.
Thus, $I$ is created at line \ref{del-set-info}, \ref{mov-set-info1} or \ref{mov-set-info2}.
Let  $\langle gp_d$, $p_d$, $node_d$, -, -, -$\rangle$ be the result returned by the call to the search operation at line \ref{del-call-search} or \ref{mov-call-search1} that precedes the creation of $I$.
Then, $x = gp_d$ and $old$ is read as a child of $p_d$ at line \ref{del-read-sibling} or \ref{mov-read-sibling}.
Since $gp_d \in F_I$, $gp_d.info = I$ just before $T$ (by Lemma \ref{change-child-after-flag-lem} and \ref{no-unflag-before-first-change-child-lem}).
So, $gp_d$ is not marked just before $T$.
By Lemma \ref{search-lem}, $gp_d$ was a parent of $p_d$ at some time during the search operation.
By Lemma \ref{no-wrong-child-change-lem}, $gp_d$ was a parent of $p_d$ just before $T$.
Since $gp_d$ is reachable just before $T$, $p_d$ is reachable just before $T$.
By Lemma \ref{no-wrong-child-change-lem}, $old$ is a child of $p_d$ just before $T$.
Since $p_d$ is reachable just before $T$, $old$ is reachable just before $T$.
By Lemma \ref{unreachable-oldChild-lem}, $old$ becomes unreachable just after $T'$ (before $T$), contradicting Lemma \ref{no-reachable-after-unreachable-lem}.
\end{proof}

Since there is no ABA problem on the $child$ field of nodes, Corollary \ref{unmarked-reachable-col} and Lemma \ref{other-child-cas-if-no-aba-lem}, \ref{no-wrong-child-change-lem}, \ref{first-child-cas-if-no-aba-lem}, 
\ref{unreachable-mark-lem}, \ref{unreachable-oldChild-lem}, \ref{query-reachable-lem} and \ref{no-reachable-after-unreachable-lem}, 
which we proved earlier with the assumption that there is no ABA on the $child$ fields of nodes is true without the assumption.

By Lemma \ref{change-child-after-flag-lem}, \ref{no-unflag-before-first-change-child-lem}, \ref{first-child-cas-if-no-aba-lem} and \ref{same-child-lem}, we have the following corollary.
\begin{corollary}\label{flag-node-before-child-cas-col}
Let $I$ be a Flag object and $x \in F_I$.
Then, $x.info = I$ when a successful child CAS of $I$ occurs.
\end{corollary}

\subsection{Correctness of the Search Operation}
In this section, we first show the post-conditions of the search operation are satisfied.
Then, we show how to linearize search operations that terminate.

Let $I$ be a Flag object that is created at line \ref{mov-set-info1} or \ref{mov-set-info2}.
Then, the number of elements of $I.pNode$, $I.oldChild$ and $I.newChild$ is two and, by Observation \ref{flag-leaf-move-obs}, $I.rmvLeaf$ is set to a non-empty leaf node.
By Observation \ref{flag-leaf-move-obs}, $I.rmvLeaf$ is flagged at line \ref{help-flag-leaf} before the first child CAS of $I$ at line \ref{help-change-child}.
Any time after the first child CAS of $I$, we say $I.rmvLeaf$ is {\it logically removed}.     
The following lemma shows that after a leaf node is flagged by a Flag object, the leaf node is not flagged by another Flag object. 

\begin{lemma} \label{flag-leaf-lem}
Let $I$ be a Flag object that is created at line \ref{mov-set-info1} or \ref{mov-set-info2}.
If $I.rmvLeaf.info$ is set to $I$ by line \ref{help-flag-leaf} at time $T$, $I.rmvLeaf.info$ is not set to another Flag object $I' \ne I$ before $T$.
\end{lemma}
\begin{proof}
Let $\langle gp_d$, $p_d$, $node_d$, -, -, $rmvd_d\rangle$ be the result returned by the call to search($val_d$) on line \ref{mov-call-search1} that precedes the creation of $I$.
Then, $I.rmvLeaf = node_d$.
First, we show the following claim.

\noindent {\bf Claim:} 
Let $I''$ be a Flag object that is created at line \ref{mov-set-info1} or \ref{mov-set-info2} and
$\langle gp_d''$, $p_d''$, -, -, -, -$\rangle$ be the result returned by the call to search($val_d''$) on line \ref{mov-call-search1} that precedes the creation of $I''$.

If all nodes in $F_I''$ get flagged successfully, $p_d''$ is reachable at all times between the time when $p_d''$ was reachable during search($val_d''$) and the first child CAS of $I''$.

\noindent {\it Proof of Claim.} 
Since $I''$ is created at line \ref{mov-set-info1} or \ref{mov-set-info2}, $gp_d'' \in F_{I''}$.
By Lemma \ref{query-reachable-lem}, $gp_d''$ is reachable at some time during search($val_d''$).
Since $gp_d''.info = I''$ just before the first child CAS of $I''$ (by Corollary \ref{flag-node-before-child-cas-col}), $gp_d''$ is not marked just before the first child CAS of $I''$.
So, $gp_d''$ is reachable just before the first child CAS of $I''$ (by Corollary \ref{unmarked-reachable-col}).
By Lemma \ref{search-lem}, $gp_d''$ was a parent of $p_d''$ at some time during the search operation.
By Lemma \ref{no-wrong-child-change-lem}, $gp_d''$ is a parent of $p_d''$ just before the first child CAS of $I''$.
Since $gp_d''$ is reachable just before the fist child CAS of $I''$, $p_d''$ is reachable just before the first child CAS of $I''$. 
By Lemma \ref{no-reachable-after-unreachable-lem}, $p_d''$ is reachable at all times between the time when $p_d''$ was reachable during search($val_d''$) and the first child CAS of $I''$.
This completes the proof of the claim.

Let $T'$ be the time when search($val_d$) reads $node_d.info$ on line \ref{search-set-moved}.
First, we show $node_d.info$ is not set to a Flag object at any time before $T'$.
Then, we show $node_d.info$ is not set to a Flag object $I' \ne I$ at any time between $T'$ and $T$.

If $node_d.info$ is an Unflag object at $T'$, by Observation \ref{flag-leaf-move-obs}, $node_d.info$ is not set to a Flag object at any time before $T'$. 
We show that $node_d.info$ is an Unflag object at time $T'$.
We prove it by contradiction.
Assume $node_d.info$ is some Flag $I'$ at $T'$.
Let $\langle$ -, $p_d'$, $node_d$, -, -, -$\rangle$ be the result returned by the call to search($val_d'$) on line \ref{mov-call-search1} that precedes the creation of $I'$.
Since the replace operation that creates $I$ does not return false at the execution of line \ref{mov-return-false1} that precedes the creation of $I$, $rmvd_d$ is false.
Since $rmvd_d = false$ and $node_d.info = I'$ at line \ref{search-set-moved}, 
$I'.pNode[0]$ was a parent of $I'.oldChild[0]$ at line \ref{logRemoved-check-oldChild} during search($val_d$).
By Lemma \ref{first-child-cas-if-no-aba-lem} and \ref{same-child-lem}, no child CAS of $I'$ on $I'.pNode[0]$ occurs before $T'$. 
By Observation \ref{flag-leaf-move-obs}, all nodes in $F_{I'}$ are flagged by $I'$ before $T'$ by Lemma \ref{no-unflag-before-first-change-child-lem}.

By the claim, $p_d'$ is reachable just before $T'$.
Since $p_d.info$ is set to $I$ after search($val_d$) and before the first child CAS of $I$ by Corollary{flag-node-before-child-cas-col}.
So, $p_d.info$ is read for the last time during search($val_d$) before $T'$ and $p_d.info$ is changed from that value to $I$ after $T'$.
By Lemma \ref{same-info-lem} and \ref{info-lem}, $p_d.info$ is an Unflag object at $T'$.
Since $p_d'.info = I'$ at $T'$ and $p_d.info$ is an Unflag object at $T'$, $p_d \ne p_d'$.
Then, by Lemma \ref{search-lem} and \ref{no-wrong-child-change-lem}, $p_d'$ was a parent of $node_d$ and $p_d$ was a parent of $node_d$ just before $T'$ and, by the claim, 
$p_d'$ and $p_d$ are both reachable just before $T'$, contradicting Lemma \ref{reachable-parent-lem} (since $p_d \ne p_d'$).
Thus, $node_d.info$ is an Unflag object at $T'$ and, by Observation \ref{flag-leaf-move-obs}, $node_d.info$ is not set to a Flag object at any time before~$T'$. 

Now, we show that $node_d.info$ is not set to some Flag object $I' \ne I$ between $T'$ and $T$.
Let $\langle$ -, $p_d'$, $node_d$, -, -, -$\rangle$ be the result returned by the call to search($val_d'$) on line \ref{mov-call-search1} that precedes the creation of $I'$.
By Lemma \ref{search-lem} and \ref{no-wrong-child-change-lem}, $p_d$ was a parent of $node_d$ at all times between $T'$ and $T$.
By the claim, $p_d$ is reachable at all times between $T'$ and $T$.
If $node_d.info$ is set to $I'$ between $T'$ and $T$, by Lemma \ref{search-lem} and \ref{no-wrong-child-change-lem}, $p_d'$ was a parent of $node_d$ just before $node_d.info$ is set to $I'$ and, 
by the claim, $p_d'$ is reachable just before $node_d.info$ is set to $I'$.
So, by Lemma \ref{reachable-parent-lem}, $p_d = p_d'$.
In this case, $p_d = p_d' \in F_{I'}$.
Thus, by Observation \ref{flag-leaf-move-obs} and Lemma \ref{no-unflag-before-first-change-child-lem}, 
$p_d.info = I'$ just before $node_d.info$ is set to $I'$ between $T'$ and $T$.
By Lemma \ref{update-x-info-lem}, $p_d.info \ne I'$ at all times between $T'$ and $T$.
Thus, $node_d.info$ is not set to some Flag object $I'$ between $T'$ and $T$.
\end{proof}

Next, we have the following lemma that shows when a leaf node is logically removed, other operations see that it is removed using the $info$ field of the leaf node.

\begin{lemma} \label{logically-remove-lem}
Let $I$ be a Flag object that is created at line \ref{mov-set-info1} or \ref{mov-set-info2}.
$I.rmvLeaf$ is logically removed if and only if $I.rmvLeaf.info = I$ and $I.pNode[0]$ is not a parent of $I.oldChild[0]$. 
\end{lemma}
\begin{proof}
First, we show that $I.rmvLeaf.info = I$ and $I.pNode[0]$ is not a parent of $I.oldChild[0]$ if $I.rmvLeaf$ is logically removed.
Let $T$ be the time when the first child CAS of $I$ is executed.
By definition, $I.rmvLeaf$ is logically removed at all times after $T$.
By Observation \ref{flag-leaf-move-obs}, $I.rmvLeaf.info$ is set to $I$ before $T$.
By Lemma \ref{flag-leaf-lem}, $I.rmvLeaf.info = I$ at all times after that.
By Lemma \ref{first-child-cas-if-no-aba-lem} and \ref{same-child-lem}, the first child CAS of $I$ successfully changes $I.pNode[0].child[i]$ from $I.oldChild[0]$ to $I.newChild[0]$ at $T$.
By Lemma \ref{old-new-different-lem},  $I.oldChild[0] \ne I.newChild[0]$.
By Lemma \ref{same-child-lem}, $I.pNode[0]$ is not a parent of $I.oldChild[0]$ after $T$.
Thus, at all times after time $T$, $I.rmvLeaf.info = I$ and $I.pNode[0]$ is not a parent of $I.oldChild[0]$.

Now, we show that $I.rmvLeaf$ is logically removed if $I.rmvLeaf.info = I$ and $I.pNode[0]$ is not a parent of $I.oldChild[0]$. 
Suppose $I.rmvLeaf.info = I$ and $I.pNode[0]$ is not a parent of $I.oldChild[0]$ at some time $T$.
Then $I$ was created before $T$, and during the search that precedes the creation of $I$, $I.oldChild[0]$ was a child of $I.pNode[0]$
Thus, some child CAS changed $I.pNode[0].child[i]$ from $I.oldChild[0]$ to a different value before $T$.
We argue that this must have been a child CAS of $I$, which means that $I.rmvLeaf$ is logically removed at time $T$.
Let $\langle$ -, $p_i$, $node_i$, -, -, -$\rangle$ be the result returned by the call to the search operation on line \ref{mov-call-search2} that precedes the creation of $I$.
Then, $I.pNode[0] = p_i$, $I.oldChild[0] = node_i$ and $p_i \in F_I$ by Observation \ref{flag-parent-obs}.
By Lemma \ref{no-wrong-child-change-lem}, no child CAS of another Flag object $I' \ne I$ changes $p_i.child$ between 
the time when $p_i.info$ is read for the last time during the search operation that precedes the creation of $I$.
\end{proof}

\begin{definition}
A node is {\it logically in the trie} at time $T$ if the node is reachable at $T$ and the node is not logically removed at $T$. 
\end{definition}
Here, we show when the search operation terminates, all post-conditions of the search operation are satisfied.

\begin{lemma} \label{post-con-search-lem}
Each call to the search operation that terminates satisfies its post-conditions.
\end{lemma}
\begin{proof}
Assume search($val$) returns $\langle gp$, $p$, $node$, $gpInfo$, $pInfo$, $rmvd \rangle$.
Lemma \ref{search-lem} shows the first four post-conditions of the search operation are satisfied. 
We prove the last two here.

If $rmvd$ is true, we prove that $node$ is logically removed at line \ref{logRemoved-check-oldChild}.
Let $I$ be the value of $node.info$ at line \ref{search-set-moved}.
If $rmvd$ is set to true at line \ref{search-set-moved}, $node$ is a leaf node whose $info$ field is a Flag object at line \ref{search-set-moved} and 
$node.info.pNode[0]$ is not a parent of $node.info.oldChild[0]$ at line \ref{logRemoved-check-oldChild}.
By Lemma \ref{flag-leaf-lem}, $node.info = I$ at all times after line \ref{search-set-moved}.
Thus, $node$ is logically removed at line \ref{logRemoved-check-oldChild} by Lemma \ref{logically-remove-lem}.

If $rmvd$ is false, we prove that $node$ is logically in the trie at some time during the search operation.
By Lemma \ref{query-reachable-lem}, $node$ is reachable at some time $T$ between the beginning of the search and the time the search visits $node$.
If $node$ is an internal node, then it is logically in the trie at $T$.
Otherwise, $node$ is a leaf node.
Since $rmvd$ is false, line \ref{logRemoved-return-false} or \ref{logRemoved-check-oldChild} returns false.
If line \ref{logRemoved-return-false} returns false, then when $node.info$ is read on line \ref{search-set-moved}, it is not flagged.
This occurs after $node$ is visited at line \ref{search-set-node} and therefore after $T$.
By Observation \ref{flag-leaf-move-obs}, $node$ is not flagged at $T$, so $node$ is logically in the trie at $T$ by Lemma \ref{logically-remove-lem}.
If line \ref{logRemoved-check-oldChild} returns false, then $node = I.rmvLeaf$ (since the test at line \ref{logRemoved-return-false} failed)
and $I.pNode[0]$ is a parent of $I.oldChild[0]$ at line \ref{logRemoved-check-oldChild}.
By Lemma \ref{logically-remove-lem}, $node$ is not logically removed at line \ref{logRemoved-check-oldChild}.
Therefore, $node$ is not logically removed at $T$.
\end{proof}


\begin{lemma} \label{search-val-not-in-trie-lem}
Assume search($val$) returns $\langle$ -, -, $node$, -, -, $rmvd \rangle$.
If $node$ is not a leaf node whose $label$ is $val$, or $rmvd$ is true, there is a time during the search operation when no leaf node whose $label$ is $val$ is logically in the trie.
\end{lemma}
\begin{proof}
First, consider the case where $node$ is not a leaf node whose $label$ is $val$.
Let $p$ be the last internal node visited by search($val$) such that $p.label$ is a prefix of $val$.
Let $i$ be the index such that $(p.label).i$ is a prefix of $val$.
Let $x$ be the child of $p$ that the search operation visits.
By definition of $p$, $x$ is not a leaf node whose $label$ is $val$.
By Lemma \ref{query-reachable-lem}, there is a time during search($val$) that $p.child[i] = x$ and $p$ and $x$ are both reachable.
By Invariant \ref{prefix-inv}, at that time, no leaf node whose $label$ is $val$ is reachable.

Now, consider the case where $node$ is a leaf node whose $label$ is $val$ and $rmvd$ is true.
By Lemma \ref{post-con-search-lem}, $node$ is logically removed at some time during search($val$).
Let $T$ be the first time that $node$ is logically removed.
Then, the search operation executes line \ref{logRemoved-check-oldChild} after $T$.
By Lemma \ref{logically-remove-lem}, $node.info$ is a Flag object $I$ at $T$.
Then, $I$ is created at line \ref{mov-set-info1} or \ref{mov-set-info2}.
Let $p$ be the second last node visited by search($val$).
By Lemma \ref{query-reachable-lem}, there is a time $T'$ during search($val$) that $p.child[i] = node$ and $p$ and $node$ are both reachable.
At $T'$, no other leaf node whose $label$ is $val$ is reachable by Invariant \ref{prefix-inv}. 
We consider two cases.

Case 1: $T$ is before $T'$.
Since $node$ is logically removed at all times after $T$, $node$ is logically removed at $T'$.
So, no leaf node whose $label$ is $val$ is logically in the trie at~$T'$.

Case 2: $T$ is between $T'$ and the time when the search operation executes line \ref{logRemoved-check-oldChild}.
By the definition, the first child CAS of $I$ is performed at $T$.
By Lemma \ref{flag-leaf-lem}, $node.info = I$ at all times after $T$. 
Let $\langle gp_d$, $p_d$, $node_d$, -, -, -$\rangle$ be the result returned by the call to the search operation on line \ref{mov-call-search1} that precedes the creation of $I$ and
$\langle$-, $p_i$, $node_i$, -, -, -$\rangle$ be the result returned by the call to the search operation on line \ref{mov-call-search2} that precedes the creation of $I$.
By the pseudo-code and Observation \ref{flag-leaf-move-obs}, $I.rmvLeaf = node_d = node$.

Since $I$ is created at line \ref{mov-set-info1} or \ref{mov-set-info2}, $p_d \in F_I - U_I$ and $gp_d \in F_I$. 
By Corollary \ref{flag-node-before-child-cas-col}, $p_d.info = I$ and $gp_d.info = I$ just before $T$.
By Lemma \ref{search-lem} and \ref{no-wrong-child-change-lem}, $gp_d.child[i_{gp}] = p_d$ and $p_d.child[i_p] = node$ just before $T$ for some $i_{gp}$ and $i_p$.
Since $gp_d \in F_I$ and $gp_d$ was reachable at some time during the search operation by Lemma \ref{query-reachable-lem}, 
$gp_d$, $p_d$ and $node$ are reachable just before $T$ by Corollary \ref{unmarked-reachable-col}.
By the pseudo-code, $p_i.child$ is changed at $T$.
Now, we show that $gp_d.child[i_{gp}]$ and $p_d.child[i_p]$ are not changed at $T$.
In this case, since $p_d \ne node_i$, $gp_d.child[i_{gp}]$ is not changed from $p_d$ to other value at $T$.
Also, since $p_d \ne p_i$, $p_d.child[i_p]$ is not changed at $T$.
So, just after $T$, $gp_d.child[i_{gp}] = p_d$ and $p_d.child[i_p] = node$ and $gp_d$ is reachable since $gp_d \notin F_I - U_I$.
By Invariant \ref{prefix-inv}, no other leaf node whose $label$ is $val$ is reachable just after $T$.
Since $node$ is logically removed at $T$, no leaf node whose $label$ is $val$ is logically in the trie just after $T$.
\end{proof}

Now, we define a set that represents all non-empty leaf nodes that are logically in the trie at the same time.
\begin{definition}
We define the set $activeValues$ at time $T$ to be the set of all values contained in leaf nodes that are logically in the trie at time $T$.
\end{definition}

Since the children of $root$ are initially two leaf nodes whose $label$s are $00...0$ and $11...1$, $activeValues$ is initially $\{00...0$, $11...1\}$.
By the definition of logically in the trie, $activeValues$ can be changed only by successful child CAS steps.

A find($val$) operation that returns true executes a search($val$) that returns $\langle$ -, -, $node$, -, -, $rmvd \rangle$ where $node$ is a leaf node whose $label$ is $val$ and $rmvd$ is false.
By the last post-condition of search, there is a time during the search when $node$ is logically in the trie, so $val \in activeValues$ at that time.
This is the linearization point of search($val$) and find($val$) that returns true.

A find($val$) operation that returns false executes a search($val$) that returns $\langle$ -, -, $node$, -, -, -$\rangle$ where $node$ is not a leaf node whose $label$ is $val$, or $rmvd$ is true.
By Lemma \ref{search-val-not-in-trie-lem}, there is a time during the search when no leaf node whose $label$ is $val$ is logically in the trie, so $val \notin activeValues$ at that time.
This is the linearization point of search($val$) and find($val$) that returns false.

\subsection{Correctness of Update Operations} \label{correctness-update-sec}
In this section, we show that the update operations behave correctly.
First, we show that if a child CAS of an update is performed, the operation returns true.

\begin{lemma} \label{update-return-true-lem}
If an update operation creates a Flag object $I$ and a child CAS of $I$ is performed, then the update operation returns true (unless it crashes).
\end{lemma}
\begin{proof}
Let $op$ be the update operation that created $I$.
Assume that, a child CAS of $I$ on $I.pNode[i]$ is performed and $op$ does not crash.
We show that $op$ returns true. 
After creating $I$, $op$ calls help($I$) at line \ref{ins-call-help}, \ref{del-call-help} or \ref{mov-call-help}. 
There are two different cases inside help($I$) depending on the value of the $doChildCAS$ variable at line \ref{help-check-doChildCAS}.

Case 1: $doChildCAS$ is true at line \ref{help-check-doChildCAS}.
So, $op$ sets $I.flagDone$ to true at line \ref{help-set-done}.
By Observation \ref{done-obs}, $I.flagDone$ is not set to false after that.
Since $I.flagDone$ is true at line \ref{help-check-done}, the call to help($I$) by $op$ returns true at line \ref{help-return-true}. 
So, $op$ returns true just after help($I$) returns.

Case 2: $op$ sets $doChildCAS$ to false on line \ref{help-set-doChildCAS} at time $T$.
So, $node.info \ne I$ at $T$ for some $node \in F_I$.
Before $T$, $op$ tries to set $node.info$ to $I$ using a flag CAS of $I$.
So, the first flag CAS of $I$ on $node$ is executed before $T$.
By Corollary \ref{flag-node-before-child-cas-col} and Lemma \ref{first-info-cas-lem}, this first flag CAS of $I$ on $node$ succeeds before $T$.
Since $node.info \ne I$ at $T$, $node.info$ is set to $I$ and, then changed from $I$ to some other value before $T$.
By Lemma \ref{no-unflag-before-first-change-child-lem}, a child CAS of $I$ is performed before $T$. 
By Observation \ref{change-child-after-set-done-obs}, $I.flagDone$ is set to true before $T$.
By Observation \ref{done-obs}, $I.flagDone$ is not set to false after that.
Since $I.flagDone$ is true at line \ref{help-check-done}, the call to help($I$) by $op$ returns true at line \ref{help-return-true}. 
So, $op$ returns true just after help($I$) returns.
\end{proof}

By Lemma \ref{update-return-true-lem}, we have the following corollary.

\begin{corollary} \label{update-no-child-cas-return-false-col}
Let $I$ be a Flag object that is created by an update operation.
If the update operation returns false, there is no child CAS of $I$.
\end{corollary}

During each loop iteration of an update operation, the update operation might create a new Flag object $I$.
The following lemma shows if the update operation begins the next iteration of the loop, no child CAS of $I$ succeeds.  

\begin{lemma} \label{update-no-child-cas-lem}
Let $I$ be a Flag object that is created during a loop iteration of an update operation.
If the update operation begins the next iteration of the loop, no child CAS of $I$ is ever performed.
\end{lemma}
\begin{proof}
Assume that an update operation $op$ creates a Flag object $I$ during a loop iteration and begins the next iteration of the loop.
After creating $I$, $op$ calls help($I$) at line \ref{ins-call-help}, \ref{del-call-help} or \ref{mov-call-help}.
Since $op$ does not return true at the end of the loop iteration after help($I$) returns, help($I$) returns false.
By Lemma \ref{update-return-true-lem}, there is no child CAS of $I$.
\end{proof}

Let $I$ be a Flag object that is created at line \ref{mov-set-info1} or \ref{mov-set-info2}.
Since the condition at line \ref{mov-check-gc-condition} that precedes the creation of $I$ is satisfied, $I.oldChild[0] \ne I.oldChild[1]$.
So, by Lemma \ref{same-child-lem}, we have the following Corollary.
\begin{corollary} \label{different-child-cas-col}
Let $I$ be a Flag object in which $I.pNode$ has two elements.
Then, the child CAS of $I$ on $I.pNode[0]$ and $I.pNode[1]$ cannot succeed on the same location.
\end{corollary}

\begin{lemma} \label{mov-sc1-lem}
Let $I$ be a Flag object that is created at line \ref{mov-set-info-sc1}.
Let $\langle$-, $p_d$, $node_d$, -, -, -$\rangle$ be the result returned by the call to search($val_d$) on line \ref{mov-call-search1} and 
$\langle$-, $p_i$, $node_i$, -, -, -$\rangle$ be the result returned by the call to search($val_i$) on line \ref{mov-call-search2} that precedes the creation of $I$. 
Then, if there is a successful child CAS of $I$ at time $T$, $p_d = p_i$.
\end{lemma}
\begin{proof}
In this case, $node_i = node_d$.
By Lemma \ref{other-child-cas-if-no-aba-lem}, $T$ is the first child CAS of $I$.
By Lemma \ref{search-lem} and \ref{no-wrong-child-change-lem}, $p_d$ was a parent of $node_d$ at all times between the time when search($val_d$) reads $p_d.info$ and $T$.
Since $p_d \in F_I$, by Corollary \ref{flag-node-before-child-cas-col}, $p_d.info = I$ just before $T$.
By Lemma \ref{query-reachable-lem}, $p_d$ was reachable at some time during search($val_d$).
By Lemma \ref{update-x-info-lem}, $p_d$ is not marked between the time when search($val_d$) reads $p_d.info$ and $T$.
By Corollary \ref{unmarked-reachable-col}, $p_d$ is reachable at all times between the time when search($val_d$) returns and $T$.
By Lemma \ref{query-reachable-lem}, there is a time during search($val_i$) that $p_i$ is a parent of $node_i = node_d$ and $p_i$ is reachable.
So, at that time, $p_d$ and $p_i$ are both parent of $node$ and both reachable.
By Lemma \ref{reachable-parent-lem}, $p_i = p_d$.
\end{proof}

\begin{lemma}\label{update-p-reachable-lem} \label{update-node-reachable-lem}
Consider a Flag object $I$.
Let $\langle gp$, $p$, $node$, $gpInfo$, $pInfo$, -$\rangle$ be the result returned by the call to search($val$) on line \ref{ins-call-search}, \ref{del-call-search}, \ref{mov-call-search1}  
that precedes the creation of $I$. 
If $I$ is created by a replace operation, let $\langle$-, $p_i$, $node_i$, -, $pInfo_i$, -$\rangle$ be the result returned by the call to search($val_i$) on line \ref{mov-call-search2} that precedes the creation of $I$.
Then, if there is a successful child CAS of $I$ on $I.pNode[0]$ at time $T$, 
\begin{itemize}
\item $p$ and $node$ are reachable just before $T$, and
\item $p_i$ and $node_i$ are reachable just before $T$ (if $I$ is created by a replace operation).
\end{itemize}
\end{lemma}
\begin{proof}
To prove the lemma, first we show $p$ and $node$ are reachable just before $T$.
Since $I$ is created by an update operation, $p \in F_I$ and $T$ is the first child CAS of $I$ by Lemma \ref{first-child-cas-if-no-aba-lem}. 
By Corollary \ref{flag-node-before-child-cas-col}, $p.info = I$ just before $T$.
By Lemma \ref{query-reachable-lem}, $p$ was reachable at some time during search($val$).
Since $p$ is not marked before $T$, by Corollary \ref{unmarked-reachable-col}, $p$ is reachable just before $T$.
By Lemma \ref{search-lem}, $p$ is a parent of $node$ at some time during search($val$).
By the pseudo-code, $p \in F_I$.
By Lemma \ref{no-wrong-child-change-lem}, $p$ is a parent of $node$ just before $T$.
Since $p$ is reachable just before $T$, $node$ is also reachable just before $T$.

For the rest of the proof, assume $I$ is created by a replace operation.
Now, we show $p_i$ and $node_i$ are reachable just before $T$.
If $I$ is created at line \ref{mov-set-info-sc1}, by Lemma \ref{mov-sc1-lem}, $p = p_i$ and $p_i \in F_I$.
If $I$ is created at any other line inside a replace operation, by the pseudo-code, $p_i \in F_I$.
By Corollary \ref{flag-node-before-child-cas-col}, $p_i.info = I$ just before $T$.
By Lemma \ref{query-reachable-lem}, $p_i$ was reachable at some time during search($val_i$).
Since $p_i$ is not marked before $T$, by Corollary \ref{unmarked-reachable-col}, $p_i$ is reachable just before $T$.
By Lemma \ref{search-lem}, $p_i$ is a parent of $node_i$ at some time during search($val_i$).
By Lemma \ref{no-wrong-child-change-lem}, $p_i$ is a parent of $node_i$ just before $T$ since $p_i \in F_I$.
Since $p_i$ is reachable just before $T$, $node_i$ is also reachable just before $T$.
\end{proof}

\begin{lemma} \label{update-node-unreachable-lem}
Let $I$ be a Flag object that is created at line \ref{mov-set-info1} or \ref{mov-set-info2} and 
$\langle gp_d$, $p_d$, $node_d$, -, -,  -$\rangle$ be the result returned by the call to search($val_d$) on line \ref{mov-call-search1} that precedes the creation of $I$.
Then, if there are a successful child CAS of $I$ on $I.pNode[0]$ at $T$ and a successful child CAS of $I$ on $I.pNode[1]$ at $T'$, 
$p_d$ and $node_d$ are reachable at all times between $T$ and $T'$ and $p_d$ and $node_d$ become unreachable immediately after $T'$.
\end{lemma}
\begin{proof}
First, we show that $p_d$ and $node_d$ are reachable at all times between $T$ and $T'$.
By Lemma \ref{update-p-reachable-lem} and \ref{update-node-reachable-lem}, $node_d$ and $p_d$ are reachable just before $T$.
In this case, $gp_d$ and $p_d$ are in $F_I$.
By Lemma \ref{search-lem}, $gp_d$ is a parent of $p_d$ at some time during search($val_d$) and $p_d$ is a parent of $node_d$ at some time during search($val_d$).
By Lemma \ref{no-wrong-child-change-lem} and Corollary \ref{different-child-cas-col}, $gp_d$ is a parent of $p_d$ and $p_d$ is a parent of $node_d$ just before $T'$ (since $p_i \ne p_d$).
By Corollary \ref{flag-node-before-child-cas-col} and Lemma \ref{no-unflag-before-first-change-child-lem}, $gp_d.info = I$ just before $T'$ and $gp_d$ is not marked just before $T'$.
By Corollary \ref{unmarked-reachable-col}, $gp_d$ is reachable just before $T'$.
So, $p_d$ and $node_d$ are reachable just before $T'$.
Thus, by Lemma \ref{no-reachable-after-unreachable-lem}, $p_d$ and $node_d$ are reachable at all times between $T$ and $T'$.

Now, we show that $p_d$ and $node_d$ become unreachable immediately after $T'$.
By Lemma \ref{unreachable-oldChild-lem}, $p_d$ becomes unreachable at $T'$.
By Lemma \ref{reachable-parent-lem}, $p_d$ is the only reachable parent of $node_d$ just before $T'$.
Thus, $node_d$ also becomes unreachable at $T'$.
\end{proof}

\begin{lemma} \label{update-node-in-trie-lem}
Consider a Flag object $I$. 
Let $\langle$ -, $p$, $node$, -, $pInfo$,  $rmvd \rangle$ be the result returned by the call to the search operation on line 
\ref{ins-call-search}, \ref{del-call-search} or \ref{mov-call-search1} that precedes the creation of $I$.
If the first successful child CAS of $I$ is at time $T$, then $node$ is logically in the trie just before $T$.
\end{lemma} 
\begin{proof}
By Lemma \ref{update-node-reachable-lem}, $node$ is reachable just before $T$.
We show that $node$ is not logically removed at any time before $T$. 

To derive a contradiction, assume $node$ becomes logically removed at $T'$ before $T$. 
Let $I' \ne I$ be the Flag object such that $node$ is logically removed by the first child CAS of $I'$ at $T'$.
Then, $I'$ is created at line \ref{mov-set-info1} or \ref{mov-set-info2}.
Let $\langle$ -, $p_i$, $node_i$, -, -, -, -, - $\rangle$ be the result returned by the call to the search operation on line \ref{mov-call-search1} that precedes the creation of $I'$.
Since $I'.rmvLeaf$ is set to $node_i$, $node = node_i$.
Since $node$ and $p_i$ become unreachable immediately after the successful child CAS of $I'$ on $I'.pNode[1]$ (by Lemma \ref{update-node-unreachable-lem}), 
$node$ and $p_i$ do not become reachable after that (by Lemma \ref{no-reachable-after-unreachable-lem}).
But, $node$ and $p_i$ are reachable just before $T$ (by Lemma \ref{update-p-reachable-lem} and \ref{update-node-reachable-lem}),
so the successful child CAS of $I'$ on $I'.pNode[1]$ does not occur before $T$.
By Lemma \ref{no-wrong-child-change-lem} and Corollary \ref{different-child-cas-col}, $p_i$ is a parent of $node$ just before $T$ 
(since $T$ is between the first child CAS of $I'$ on $I'.pNode[0]$ and the first child CAS of $I'$ on $I'.pNode[1]$).
By Lemma \ref{search-lem} and \ref{no-wrong-child-change-lem}, $p$ is a parent of $node$ just before $T$.
By Lemma \ref{update-p-reachable-lem}, $p$ is reachable just before $T$.
By Lemma \ref{reachable-parent-lem}, $p_i = p$ (since both $p$ and $p_i$ are reachable parents of $node$ just before $T$).
Since $I'$ is created at line \ref{mov-set-info1} or \ref{mov-set-info2}, by the pseudo-code, $p = p_i \in F_{I'}$.
By Corollary \ref{flag-node-before-child-cas-col}, $p.info =I'$ just before $T$ (that is between the first child CAS of $I'$ on $I'.pNode[0]$ and the first child CAS of $I'$ on $I'.pNode[1]$), 
contradicting the fact that $p.info = I \ne I'$ just before $T$ by Lemma \ref{update-x-info-lem}.
\end{proof}

\subsubsection{Correctness of the Insert Operation}

In this section, we show how the set $activeValues$ is changed by insert operations and, then we show how to linearize insert operations.

\begin{lemma} \label{unsuccessful-ins-lem}
If insert($val$) returns false, there is a time during the insert($val$) when $val \in activeValues$.
\end{lemma}
\begin{proof}
Assume insert($val$) returns false.
Let $\langle$ -, -, $node$, -, -, $rmvd \rangle$ be the result returned by the last call to search($val$) just before insert($val$) returns.
Since insert($val$) returns false at line \ref{ins-return-false}, keyInTrie($node, val, rmvd$) returns true at line \ref{ins-return-false}.
Then, $node$ is a leaf node containing $val$ and $rmvd$ is false.
By Lemma \ref{post-con-search-lem}, $node$ is logically in the trie at some time during search($val$).
So, $val \in activeValues$ at that time.
\end{proof}

Let $I$ be a Flag object created by insert($val$).
By Corollary \ref{update-no-child-cas-return-false-col}, there is no successful child CAS of $I$ if insert($val$) returns false.
So, $activeValues$ is not changed by a child CAS of $I$.

Let $I$ be a Flag object that is created during a loop iteration of insert($val$).
If a child CAS of $I$ succeeds, by Lemma \ref{update-return-true-lem}, insert($val$) returns true at the end of the loop iteration unless it crashes.
If a child CAS of $I$ succeeds, we say insert($val$) is {\it successful}.
By Lemma \ref{other-child-cas-if-no-aba-lem} and \ref{first-child-cas-if-no-aba-lem}, only the first child CAS of $I$ succeeds.
In the following lemma, we show how the successful insert operations change $activeValues$. 

\begin{lemma} \label{successful-ins-lem}
Let $I$ be a Flag object that is created by insert($val$).
Assume a child CAS of $I$ succeeds at time $T$.
Let $s$ be the set $activeValues$ just before $T$ and let $s'$ be the set $activeValues$ just after $T$.
Then, $val \notin s$ and $s' = s \cup \{val\}$
\end{lemma}
\begin{proof}

Let $\langle$ -, $p$, $node$, -, -, $rmvd \rangle$ be the result returned by the call to search($val$) that precedes the creation of $I$.
Then, $p \in F_I$.
By Lemma \ref{search-lem}, $p.label$ is a prefix of $val$ and if $node$ is an internal node, $node.label$ is not a prefix of $val$.
Since insert($val$) does not return false at line \ref{ins-return-false} before the creation of $I$ by Corollary \ref{update-no-child-cas-return-false-col}, 
either $node$ is not a leaf node whose $label$ is $val$ or $rmvd$ is true. 
If $node$ is not a leaf node whose $label$ is $val$, since $p$ is parent of $node$ just before $T$ (by Lemma \ref{search-lem} and \ref{no-wrong-child-change-lem}) and
$p$ is reachable just before $T$ (by Lemma \ref{update-p-reachable-lem}), $p$ is not an ancestor of a leaf node whose $label$ is $val$ just before $T$.
If $node$ is a leaf node whose $label$ is $val$ and $rmvd$ is true, by Lemma \ref{post-con-search-lem}, $node$ is logically removed at some time during search($val$), so $node$ is logically removed just before $T$.
Since $p$ is a parent of $node$ just before $T$ and $p$ is reachable just before $T$, no leaf node whose $label$ is $val$ is logically in the trie just before $T$ by Invariant \ref{prefix-inv}.
Thus, $val \notin s$.

The successful child CAS of $I$ changes an element of $I.pNode[0].child$ from $I.oldChild[0]$ to $I.newChild[0]$.
When $I$ is initialized, $I.pNode[0]$ and $I.oldChild[0]$ are set to $p$ and $node$ respectively. 
The successful child CAS of $I$ changes an element of $p.child$ from $node$ to $I.newChild[0]$ at time $T$.
When $I$ is initialized, $I.newChild[0]$ is set to the new internal node created at line \ref{create-internal}.
The non-empty elements of the $child$ field of the new internal node are a new copy of $node$ and a new leaf node whose $label$ is $val$.

If $node$ is a leaf node, since $p$ is reachable just before $T$ (by Lemma \ref{update-p-reachable-lem}), just after $T$, $node$ is unreachable and the new leaf node and the new copy of $node$ are reachable.
So, $s' = s \cup \{val\}$.

If $node$ is an internal node, $node \in F_I$ by the pseudo-code and, by Lemma \ref{no-wrong-child-change-lem}, 
no child CAS changes $node.child$ between the time when $node.info$ is read for the last time on line \ref{ins-read-node-info} and $T$.
Since the new copy of $node$ is made at line \ref{create-copy} between the time when $node.info$ is read for the last time on line \ref{ins-read-node-info} and $T$, 
$node.child$ just before $T$ is the same as the children of the new copy of $node$ just after $T$.
By Lemma \ref{update-p-reachable-lem}, $p$ is reachable just before $T$.
Just after $T$, $node$ is unreachable and the new leaf node and the new copy of $node$ are reachable.
Also, all children of $node$ just before $T$ are reachable just after $T$.
So, $activeValues = s \cup \{val\}$, just after $T$.
\end{proof}

An insert($val$) returns false executes a search($val$) that returns $\langle$ -, -, $node$, -, -, $rmvd \rangle$ where $node$ is a leaf node whose $label$ is $val$ and $rmvd$ is false.
By the last post-condition of search, there is a time during the search when $node$ is logically in the trie, so $val \in activeValues$ at that time.
This is the linearization point of search($val$) and insert($val$) that returns false.
By Corollary \ref{update-no-child-cas-return-false-col}, if insert($val$) returns false, there is no child CAS of $I$ where $I$ is created by insert($val$).

If the first child CAS of $I$ that is created by insert($val$) occurs, insert($val$) is linearized at that child CAS of $I$.
By Lemma \ref{first-child-cas-if-no-aba-lem}, the first child CAS of $I$ succeeds and by Lemma \ref{other-child-cas-if-no-aba-lem}, no other child CAS of $I$ succeeds.
By Lemma \ref{successful-ins-lem}, $val \notin activeValues$ just before the first child CAS of $I$ and $val$ is added to $activeValues$ just after the first child CAS of $I$.
By Lemma \ref{update-return-true-lem}, if the first child CAS of $I$ occurs, insert($val$) returns true (unless it crashes). 

If insert($val$) does not return false and the first child CAS of $I$ that is created by insert($val$) does not occur, no linearization point is assigned to insert($val$).

\subsubsection{Correctness of the Delete Operation}
In this section, we show how the set $activeValues$ is changed by delete operations and, then we show how to linearize delete operations.

\begin{lemma} \label{unsuccessful-del-lem}
If delete($val$) returns false, there is a time during the delete($val$) when $val \notin activeValues$.
\end{lemma}
\begin{proof}
Assume delete($val$) returns false at line \ref{del-return-false}.
Let $\langle$ -, -, $node$, -, - $rmvd \rangle$ be the result returned by the last call to search($val$) before delete($val$) returns false.
Since delete($val$) returns false at line \ref{del-return-false}, keyInTrie($nod, val, rmvd$) returns false at line \ref{del-return-false}.
Then, $node$ is not a leaf node whose $label$ is $val$ or $rmvd$ is true.
By Lemma \ref{search-val-not-in-trie-lem}, there is a time during search($val$) when there is no leaf node that has $label$ $val$ is logically in the trie.
So, $val \notin activeValues$ at that time.
\end{proof}

Let $I$ be a Flag object created by delete($val$).
By Corollary \ref{update-no-child-cas-return-false-col}, there is no successful child CAS of $I$.
So, $activeValues$ is not changed by a child CAS of $I$.

Let $I$ be a Flag object that is created by delete($val$) during a loop iteration.
If a child CAS of $I$ on $I.pNode[0]$ succeeds, by Lemma \ref{update-return-true-lem}, delete($val$) returns true at the end of the loop iteration unless it crashes.
If a child CAS of $I$ on $I.pNode[0]$ succeeds, we say delete($val$) is {\it successful}.
By Lemma \ref{other-child-cas-if-no-aba-lem} and \ref{first-child-cas-if-no-aba-lem}, only the first child CAS of $I$ succeeds. 
In the following lemma, we show how the successful delete operations change $activeValues$. 

\begin{lemma} \label{successful-del-lem}
Let $I$ be a Flag object that is created by delete($val$).
Assume a child CAS of $I$ succeeds at time $T$.
Let $s$ be the set $activeValues$ just before $T$ and let $s'$ be the set $activeValues$ just after $T$.
Then, $val \in s$ and $s' = s - \{val\}$
\end{lemma}
\begin{proof}
Let $\langle gp$, $p$, $node$, -, -, $rmvd \rangle$ be the result returned by the call to search($val$) that precedes the creation of $I$.
Since delete($val$) does not return false at line \ref{del-return-false} by Corollary \ref{update-no-child-cas-return-false-col}, $node$ is a leaf node whose $label$ is $val$ and $rmvd$ is false.
By Lemma \ref{update-node-in-trie-lem}, $node$ is logically in the trie just before $T$, so $val \in s$.

The successful child CAS of $I$ changes an element of $I.pNode[0].child$ from $I.oldChild[0]$ to $I.newChild[0]$.
Now, we show that $s' = s - \{val\}$.
The $nodeSibling$ variable is set to an element of $p.child$ at line \ref{del-read-sibling}.
When $I$ is initialized, $I.pNode[0]$, $I.oldChild[0]$ and $I.newChild[0]$ are set to $gp$, $p$ and $nodeSibling$ respectively at line \ref{del-set-info}.
Let $T'$ be the time when $p.info$ is read for the last time during the search($val$) that precedes the creation of $I$.
By Lemma \ref{no-wrong-child-change-lem}, since $p \in F_I$, $p.child$ is not changed between $T'$ and $T$.
Since $p$ is a parent of $node$ after $T'$ during search($val$) (by Lemma \ref{search-lem}), $node$ and $nodeSibling$ are children of $p$ at all times between $T'$ and $T$.
The child CAS of $I$ changes an element of $gp.child$ from $p$ to $nodeSibling$ at $T$.
By Lemma \ref{update-p-reachable-lem}, $p$ is reachable just before $T$.
So, just before $T$, $node$ and $nodeSibling$ are also reachable.
By Lemma \ref{no-same-node-lem}, no other leaf node whose $label$ is $val$ is reachable just before $T$.

In this case, $gp \notin F_I - U_I$.
By Lemma \ref{update-x-info-lem} and  Corollary \ref{unmarked-reachable-col}, $gp$ is reachable just before $T$.
So, just after $T$, $p$ and $node$ are not reachable, but $gp$ and $nodeSibling$ are reachable.
If $nodeSibling$ is an internal node, all elements of its $child$ field that are reachable just before $T$ are reachable just after $T$.
Thus, $s' = s - \{val\}$.
\end{proof}

A delete($val$) returns false executes a search($val$) that returns $\langle$ -, -, $node$, -, -, $rmvd \rangle$ where $node$ is not a leaf node whose $label$ is $val$, or $rmvd$ is true.
By Lemma \ref{search-val-not-in-trie-lem}, there is a time during the search when no leaf node whose $label$ is $val$ is logically in the trie, so $val \notin activeValues$ at that time.
This is the linearization point of search($val$) and delete($val$) that returns false.
By Corollary \ref{update-no-child-cas-return-false-col}, if delete($val$) returns false, there is no child CAS of $I$ where $I$ is created by delete($val$).

If the first child CAS of $I$ that is created by delete($val$) occurs, delete($val$) is linearized at that child CAS of $I$.
By Lemma \ref{first-child-cas-if-no-aba-lem}, the first child CAS of $I$ succeeds and by Lemma \ref{other-child-cas-if-no-aba-lem}, no other child CAS of $I$ succeeds.
By Lemma \ref{successful-del-lem}, $val \in activeValues$ just before the first child CAS of $I$ and $val$ is removed from $activeValues$ just after the first child CAS of $I$.
By Lemma \ref{update-return-true-lem}, if the first child CAS of $I$ occurs, delete($val$) returns true (unless it crashes). 

If delete($val$) does not return false and the first child CAS of $I$ that is created by delete($val$) does not occur, no linearization point is assigned to delete($val$).

\subsubsection{Correctness of the replace operation}
In this section, we show how the set $activeValues$ is changed by replace operations and, then we show how to linearize replace operations.

\begin{lemma} \label{unsuccessful-mov-lem}
If replace($val_d , val_i$) returns false, there is a time during replace($val_d, val_i$) when $val_d \notin activeValues$ or $val_i \in activeValues$.
\end{lemma}
\begin{proof}
Assume replace($val_d , val_i$) returns false.
Then, we have two cases according to what line returns false.

Case 1: replace($val_d , val_i$) returns false at line \ref{mov-return-false1}.
Let $\langle$ -, -, $node_d$, -, -, $rmvd_d \rangle$ be the result returned by the last call to search($val_d$) on line \ref{mov-call-search1}.
So, keyInTrie($node_d, val_d, moved_d$) returns false at line \ref{mov-return-false1}.
Then, $node_d$ is not a leaf node containing $val_d$ or $rmvd_d$ is true.
By Lemma \ref{search-val-not-in-trie-lem}, there is a time during search($val_d$) when there is no leaf node that has $label$ $val_d$ and is logically in the trie.
So,  $val_d \notin activeValues$ at that time.

Case 2: replace($val_d , val_i$) returns false at line \ref{mov-return-false2}.
Let $\langle$ -, -, $node_i$, -, - $rmvd_i \rangle$ be the result returned by the last call to search($val_i$) on line \ref{mov-call-search2} before replace($val_d , val_i$) returns false.
So, keyInTrie($node_i, val_i, rmvd_i$) returns true.
Then, $node_i$ is a leaf node containing $val_i$ and $rmvd_i$ is false.
By Lemma \ref{post-con-search-lem}, $node_i$ is logically in the trie at some time during search($val_i$).
So, $val_i \in activeValues$ at that time.
\end{proof}

Let $I$ be a Flag object created by replace($val_d, val_i$).
By Corollary \ref{update-no-child-cas-return-false-col}, there is no successful child CAS of $I$.
So, $activeValues$ is not changed by a child CAS of~$I$.

Let $I$ be a Flag object that is created by replace($val_d, val_i$) during a loop iteration.
If, for all $i$, a child CAS of $I$ on $I.pNode[i]$ succeeds, by Lemma \ref{update-return-true-lem}, replace($val_d, val_i$) returns true at the end of the loop iteration unless it crashes.
If a child CAS of $I$ on $I.pNode[0]$ succeeds, we say replace($val_d, val_i$) is {\it successful}.
By Lemma \ref{other-child-cas-if-no-aba-lem} and \ref{first-child-cas-if-no-aba-lem}, only the first child CAS of $I$ on $I.pNode[0]$ succeeds. 

There are different cases of replace($val_d, val_i$) depending on what line creates $I$.
The general case of replace($val_d, val_i$) is when $I$ is created at line \ref{mov-set-info1} or \ref{mov-set-info2}.
The special cases of replace($val_d, val_i$) are when $I$ is created at line \ref{mov-set-info-sc1}, \ref{mov-set-info-sc23} or \ref{mov-set-info-sc4}.
First, we show how successful replace operations that perform the general case change $activeValues$.

\begin{lemma} \label{successful-general-mov-lem}
Let $I$ be a Flag object that is created by replace($val_d, val_i$) at line \ref{mov-set-info1} or \ref{mov-set-info2}.
Assume a child CAS of $I$ on $I.pNode[0]$ succeeds at time $T_0$ and a child CAS of $I$ on $I.pNode[1]$ succeeds at time $T_1$.
Let $s_0$ be the set $activeValues$ just before $T_0$, let $s_0'$ be the set $activeValues$ just after $T_0$, 
let $s_1$ be the set $activeValues$ just before $T_1$ and let $s_1'$ be the set $activeValues$ just after $T_1$.
Then, $val_d \in s_0$, $val_i \notin s_0$, $s_0' = s_0 \cup \{val_i\} - \{val_d\}$ and $s_1' = s_1$.
\end{lemma}
\begin{proof}
Let $\langle$ -, $p_i$, $node_i$, -, -, $rmvd_i \rangle$ be the result returned by the call to search($val_i$) on line \ref{mov-call-search2} that precedes the creation of $I$.
When $I$ is initialized, $I.pNode[0]$, $I.oldChild[0]$ and $I.newChild[0]$ are set to $p_i$, $node_i$ and the new node that is created at line \ref{mov-call-create-node}.
The child CAS of $I$ on $I.pNode[0]$ at $T_0$ behaves in the same way as a successful child CAS of an insert($val_i$) operation.
By a similar argument to Lemma \ref{successful-ins-lem}, only $node_i$ becomes unreachable just after $T_0$, $val_i \in s_0'$ and $val_i \notin s_0$.

Now, we show that $val_d \in s_0$. 
Let $\langle gp_d$, $p_d$, $node_d$, -, -, $rmvd_d \rangle$ be the result returned by the call to search($val_d$) on line \ref{mov-call-search1} that precedes the creation of $I$.
Since replace($val_d, val_i$) does not return false at line \ref{mov-return-false1}, $node_d$ is a leaf node whose $label$ is $val_d$ and $rmvd_d$ is false.
By Lemma \ref{update-node-in-trie-lem}, $node_d$ is logically in the trie just before $T_0$, so $val_d \in s_0$.
Next, we show that, $val_d \notin s_0'$.

By the pseudo-code, $I.rmvLeaf$ is $node_d$.
By the definition of logically removed, $node_d$ is logically removed just after $T_0$ and $node_d$ is not logically in the trie at any time after $T_0$.
Since $node_d$ is reachable just before $T_0$ (by Lemma \ref{update-node-reachable-lem}), no other leaf node whose $label$ is $val_d$ is reachable just before $T_0$ by Lemma \ref{no-same-node-lem}.
Since the condition at line \ref{mov-check-gc-condition} is true, $node_d \ne node_i$.
So, $node_d$ is reachable just after $T_0$ (since $I.oldChild[0] \ne node_d$).
So, no other leaf node whose $label$ is $val_d$ is reachable just after $T_0$. 
So, $s_0' = s_0 - \{val_d\} \cup \{val_i\}$ (since $val_i \notin s_0$ and $val_i \in s_0'$). 

Now, we show $s_1 = s_1'$.
By the definition of logically removed, $node$ is logically removed at all times after $T_0$. 
By Lemma \ref{update-node-reachable-lem}, $node_d$ is reachable just before $T_1$.
By Lemma \ref{no-same-node-lem}, no other leaf node whose $label$ is $val_d$ is reachable just before $T_1$.
So, $val_d \notin s_1$. 

Let $nodeSibling$ be the node that is read as an element of $p_d.child$ at line \ref{mov-read-sibling}.
Since $I$ is created at line \ref{mov-set-info1} or \ref{mov-set-info2}, $I.pNode[1]$, $I.oldChild[1]$ and $I.newChild[1]$ are set to $gp_d$, $p_d$ and $nodeSibling$, respectively. 
The child CAS of $I$ on $I.pNode[1]$ is the same as a successful child CAS of $I'$ where $I'$ is created by delete($val_d$) at line \ref{del-set-info}.
However, by the definition of logically removed, $node_d$ is logically removed just after $T_0$.
Since the condition at line \ref{mov-check-gc-condition} is true, $node_d \ne node_i$.
So, an element of $p_d.child$ is not changed from $node_d$ to another value at $T_0$.  
By a similar argument to Lemma \ref{successful-del-lem}, only $node_d$ and $p_d$ become unreachable just after $T_1$ and no other leaf node whose $label$ is $val_d$ becomes reachable just after $T_1$.
Since $node_d$ is a leaf node whose $label$ is $val_d$, $node_d$ is the only leaf node that becomes unreachable just after $T_1$ and $val_d \notin s_1$, $s_1 = s_1'$.
\end{proof}

Next, we show how successful replace operations that perform special cases of the replace operation change $activeValues$.

\begin{lemma} \label{successful-special-mov-lem}
Let $I$ be a Flag object that is created by replace($val_d, val_i$) at line \ref{mov-set-info-sc1}, \ref{mov-set-info-sc23} or \ref{mov-set-info-sc4}.
Assume a child CAS of $I$ succeeds at time $T$.
Let $s$ be the set $activeValues$ just before $T$ and let $s'$ be the set $activeValues$ just after $T$.
Then, $val_d \in s$, $val_i \notin s$ and $s' = s \cup \{val_i\} - \{val_d\}$.
\end{lemma}
\begin{proof}
Let $\langle gp_d$, $p_d$, $node_d$, -, -, $rmvd_d \rangle$ be the result returned by the call to search($val_d$) on line \ref{mov-call-search1} that precedes the creation of $I$. 
Since replace($val_d, val_i$) did not return false at line \ref{mov-return-false1}, $node_d$ is a leaf node whose $label$ is $val_d$ and $rmvd_d$ is false.
By Lemma \ref{update-node-in-trie-lem}, $node_d$ is logically in the trie just before $T$.
So, $val_d \in s$.

Now, we show no leaf node containing $val_i$ is reachable just before $T$.
Let $\langle$ -, $p_i$, $node_i$, -, -, $rmvd_i \rangle$ be the result returned by the call to search($val_i$) on line \ref{mov-call-search2} that precedes the creation of $I$.
By Lemma \ref{search-lem}, $p_i.label$ is a prefix of $val_i$ and if $node_i$ is an internal node, $node_i.label$ is not a prefix of $val_i$.
So, by Invariant \ref{prefix-inv}, $p_i$ is the lowest internal node whose $label$ is a prefix of $val_i$ that search($val_i$) visits.
By Lemma \ref{search-lem}, $p_i$ is a parent of $node_i$ at some time during search($val_i$).
By Lemma \ref{no-wrong-child-change-lem}, $p_i$ is a parent of $node_i$ just before $T$.
By Lemma \ref{update-p-reachable-lem}, $p_i$ is reachable just before $T$. 

Since replace($val_d, val_i$) did not return false at line \ref{mov-return-false2}, $node_i$ is not a leaf node containing $val_i$ or $rmvd_i$ is true.
If $rmvd_i$ is true, by Lemma \ref{post-con-search-lem}, $node_i$ is logically removed at some time during search($val_i$).
If $node_i$ is a leaf node whose $label$ is $val_i$, since $rmvd_i$ is true, $node_i$ is logically removed just before $T$.
Since $p_i$ is a reachable parent of $node_i$ just before $T$, 
$p_i$ is not an ancestor of any other leaf node containing $val_i$ just before $T$.
If $node_i$ is a leaf node and $node_i.label \ne val_i$, since $p_i$ is a reachable parent of $node_i$ just before $T$ and $p_i.label$ is a prefix of $val_i$ (by Invariant \ref{prefix-inv}), 
$p_i$ is not an ancestor of any leaf node containing $val_i$ just before $T$.
If $node_i$ is an internal node, since $p_i$ is a reachable parent of $node_i$ just before $T$ and $node_i.label$ is not a prefix of $val_i$ (by Lemma \ref{search-lem}), 
$p_i$ is not an ancestor of any leaf node containing $val_i$ just before $T$.
Thus, no leaf node containing $val_i$ is logically in the trie just before $T$ and $val_i \notin s$.

The successful child CAS of $I$ changes an element of $I.pNode[0].child$ from $I.oldChild[0]$ to $I.newChild[0]$.
There are different special cases.
We show that $s' = s - \{val_d\} \cup \{val_i\}$ for each case.

Case 1: $I$ is created at line \ref{mov-set-info-sc1}.
Then, $node_d = node_i$.
When $I$ is initialized, $I.pNode[0]$ $I.oldChild[0]$ and $I.newChild[0]$ are set to $p_d$, $node_i$ and a new leaf node whose $label$ is $val_i$.
Since the operation does not return at line \ref{mov-return-false1}, $node_d = node_i$ is a leaf node and $node_d.label = val_d$.
By Lemma \ref{update-node-reachable-lem}, $node_i$ is reachable just before $T$.
By Lemma \ref{no-same-node-lem}, no other leaf node containing $val_d$ is reachable just before $T$ (since $node_i = node_d$).
So, $p_d$ and the new leaf node whose $label$ is $val_i$ are reachable just after $T$ and $node_i = node_d$ becomes unreachable at $T$.
Since $node_i = node_d$ is a leaf node whose $label$ is $val_d$, $s' = s - \{val_d\} \cup \{val_i\}$.

Case 2: $I$ is created at line \ref{mov-set-info-sc23}. 
Let  $nodeSibling$ be an element of $p_d.child$ that is read at line \ref{mov-read-sibling}.
Then, $newNode$ be the node that is created at line \ref{mov-create-node1}.
Then, $newNode.child$ are initially a new leaf node whose $label$ is $val_i$ and $nodeSibling$.
When $I$ is created, $I.pNode[0]$, $I.oldChild[0]$ and $I.newChild[0]$ are set to $gp_d$, $p_d$ and $newNode$. 
Since the operation does not return at line \ref{mov-return-false1}, $node_d$ is a leaf node whose $label$ is $val_d$.
Then, $nodeSibling$ is set to an element of $p_d.child$ that is not $node_d$ at line \ref{mov-read-sibling} (by Invariant \ref{prefix-inv}).
So, $node_d$ and $nodeSibling$ are children of $p_d = node_i$ just before $T$ (by Lemma \ref{no-wrong-child-change-lem}).
By Lemma \ref{update-p-reachable-lem}, $node_d$ is reachable just before $T$.
Since $node_d$ is a leaf node containing $val_d$, by Lemma \ref{no-same-node-lem}, no other leaf node containing $val_d$ is reachable just before $T$.
By Lemma \ref{reachable-parent-lem}, $gp_d$ is the only reachable parent of $p_d$ just before $T$ and $p_d$ is the only reachable parent of $node_d$ just before $T$.
At $T$, an element of $gp_d.child$ is changed from $p_d$ to $newNode$.
So, just after $T$, $p_d$ and $node_d$ become unreachable and $newNode$ and a new Leaf node whose $label$ is $val_i$ become reachable.
Since $nodeSibling$ is a child of $newNode$ just after $T$, $nodeSibling$ is reachable just after $T$.
Thus, $s' = s - \{val_d\} \cup \{val_i\}$. 

Case 3: $I$ is created at line \ref{mov-set-info-sc4}.
Then, $gp_d = node_i$.
Let  $nodeSibling$ be an element of $p_d.child$ that is read at line \ref{mov-read-sibling}, $pSibling$ be an element of $gp_d$ that is read at line \ref{mov-read-pSibling}.
Then, $newChild$ is the new node that is created at line \ref{mov-create-node2} and $newNode$ is the Internal node that is created at line \ref{mov-create-node3}.
Initially, the children of $newChild$ are $nodeSibling$ and $pSibling$.
Initially, the children of $newNode$ are a new Leaf node whose $label$ is $val_i$ and $newChild$.
When $I$ is created, $I.pNode[0]$, $I.oldChild[0]$ and $I.newChild[0]$ are set to $p_i$, $node_i$ and $newNode$.
Since $p_i \in F_I$, by Lemma \ref{update-x-info-lem}, $p_i$ is not marked at all times between the last time $p_i.info$ is read during search($val_i$) and $T$.
By Corollary \ref{unmarked-reachable-col}, $p_i$ is reachable at all times  between the last time $p_i.info$ is read during search($val_i$) and $T$. 
By Lemma \ref{search-lem} and \ref{no-wrong-child-change-lem}, $p_i$ is a parent of $node_i$ just before $T$.
So, $node_i$ is reachable just before $T$.
Since $gp_d \in F_I$, the children of $gp_d = node_i$ are $p_d$ and $pSibling$ just before $T$ (by Lemma \ref{no-wrong-child-change-lem}).
Since $p_d \in F_I$, the children of $p_d$ are $node_d$ and $nodeSibling$ just before $T$ (by Lemma \ref{no-wrong-child-change-lem}).
By Lemma \ref{update-p-reachable-lem}, $p_i$, $node_i = gp_d$, $p_d$, $pSibling$, $node_d$ and $nodeSibling$ are reachable just before $T$.
Since $node_d$ is reachable just before $T$, no other leaf node containing $val_d$ is reachable just before $T$ (by Lemma \ref{no-same-node-lem}).
By Lemma  \ref{reachable-parent-lem}, $p_i$ is the only reachable parent of $gp_d = node_i$ just before $T$, $gp_d$ is the only reachable parent of $p_d$ and $pSibling$ just before $T$ and 
$p_d$ is the only reachable parent of $node_d$ and $nodeSibling$ just before $T$.
At $T$, an element of $p_i.child$ is changed from $gp_d = node_i$ to $newNode$.
Just after $T$, $node_i = gp_d$, $p_d$ and $node_d$ become unreachable and $newNode$, $newChild$, a new Leaf whose $label$ is $val_i$ are reachable.
Since $nodeSibling$ and $pSibling$ are children of $newNode$ just after $T$, $nodeSibling$ and $pSibling$ are reachable just after $T$.
Since $node_d$ is the only reachable leaf node containing $val_d$ just before $T$, there is no reachable leaf node containing $val_d$ just after $T$.
Thus, $s' = s - \{val_d\} \cup \{val_i\}$. 
\end{proof}

A replace($val_d$, $val_i$) returns false at line \ref{mov-return-false1} executes 
a search($val_d$) that returns $\langle$ -, -, $node_d$, -, -, $rmvd_d \rangle$ where $node_d$ is not a leaf node whose $label$ is $val_d$, or $rmvd_d$ is true.
By Lemma \ref{search-val-not-in-trie-lem}, there is a time during the search($val_d$) when no leaf node whose $label$ is $val_d$ is logically in the trie, so $val_d \notin activeValues$ at that time.
This is the linearization point of search($val_d$) and delete($val_d$) that returns false at line \ref{mov-return-false1}.
A replace($val_d$, $val_i$) returns false at line \ref{mov-return-false2} executes 
a search($val_i$) that returns $\langle$ -, -, $node_i$, -, -, $rmvd_i \rangle$ where $node_i$ is a leaf node whose $label$ is $val_i$ and $rmvd_i$ is false.
By the last post-condition of search, there is a time during the search($val_i$) when $node_i$ is logically in the trie, so $val_i \in activeValues$ at that time.
This is the linearization point of search($val_d$) and delete($val_d$) that returns false at line \ref{mov-return-false2}.
By Corollary \ref{update-no-child-cas-return-false-col}, if replace($val_d$, $val_i$) returns false, there is no child CAS of $I$ where $I$ is created by replace($val_d$, $val_i$).

Let $I$ be a Flag object that is created by replace($val_d$, $val_i$).
If the first child CAS of $I$ on $I.pNode[0]$ occurs, replace($val_d$, $val_i$) is linearized at that child CAS of $I$.
By Lemma \ref{first-child-cas-if-no-aba-lem}, the first child CAS of $I$ on $I.pNode[0]$ succeeds and by Lemma \ref{other-child-cas-if-no-aba-lem}, no other child CAS of $I$ on $I.pNode[0]$ succeeds.
By Lemma \ref{successful-general-mov-lem} and \ref{successful-special-mov-lem}, $val_d \in activeValues$ and $val_i \notin activeValues$ just before the first child CAS of $I$  on $I.pNode[0]$ and 
$val_d$ is removed from $activeValues$ and $val_i$ is added to $activeValues$ just after the first child CAS of $I$ on $I.pNode[0]$.
By Lemma \ref{update-return-true-lem}, if the first child CAS of $I$ on $I.pNode[0]$ occurs, replace($val_d$, $val_i$) returns true (unless it crashes). 
If $I$ is created at line \ref{mov-set-info1} or \ref{mov-set-info2} and the first child CAS of $I$ on $I.pNode[1]$ occurs, 
$activeValues$ is not changed at the first child CAS of $I$ on $I.pNode[1]$ (by Lemma \ref{successful-general-mov-lem}).

In Section \ref{correctness-update-sec}, we have shown that $activeValues$ is changed correctly just after the linearization points of update operations.
By the definition of logically in the trie, $activeValues$ can be changed only by successful child CAS steps.
By Lemma \ref{other-child-cas-if-no-aba-lem} and \ref{first-child-cas-if-no-aba-lem}, only the first child CAS of $I$ on each element of $I.pNode$ succeeds.
We show that $activeValues$ is changed at the first child CAS of $I$ on $I.pNode[0]$ by Lemma \ref{successful-ins-lem}, \ref{successful-del-lem}, \ref{successful-general-mov-lem} and \ref{successful-special-mov-lem}.
If $I$ is created at line \ref{mov-set-info1} or \ref{mov-set-info2}, we show that the first child CAS of $I$ on $I.pNode[1]$ does not change $activeValues$.
So, $activeValues$ is only changed by the child CAS of $I$ on $I.pNode[0]$.
Since $activeValues$ is changed correctly, the find operation returns correct results at its linearization point, according to Lemma \ref{post-con-search-lem} and \ref{search-val-not-in-trie-lem}.

\subsection{Progress} \label{correctness-progress-lem}
By the lemmas in previous sections, each operation returns the same result as it would if the operations were done in the order of their linearization points.
So, operations are linearized correctly.
Now, we show the implementation is non-blocking.

First, we show that the search operation is wait-free.
 
\begin{lemma} \label{search-wf-lem}
The search operation is wait-free.    
\end{lemma}
\begin{proof}
Let $\ell$ be the length of the keys in $U$.
By Invariant \ref{prefix-inv}, length of $node.label$ increases by at least one in each iteration of loop.
Since $label$s of nodes have length at most $\ell$, there are at most $\ell$ iterations.
\end{proof}

Initially, the children of $root$ are two leaf nodes whose $label$s are $00...0$ and $11...1$.
The following lemma shows that two leaf nodes whose $label$s are $00...0$ and $11...1$ are logically in the trie at any time.

\begin{lemma} \label{dummy-lem}
Two leaf nodes whose $label$s are $00...0$ and $11...1$ are logically in the trie at all times.
\end{lemma}
\begin{proof}
Initially, $root.child$ is set to two leaf nodes whose $label$s are $00...0$ and $11...1$.
So, initially, $00...0$ and $11...1$ are in the $activeValues$ set.
By the definition, only the successful child CAS steps change the $activeValues$ set.
By the precondition of update operations, if insert($val$), delete($val$) is called, $val$ is not equal to $00...0$ or $11...1$.
By Lemma \ref{update-no-child-cas-lem}, \ref{successful-ins-lem} and \ref{successful-del-lem}, 
no child CAS of $I$ that $I$ is created by an insert operation or a delete operation removes $00...0$ or $11...1$ from $activeValues$.
By the precondition of replace operation, if replace($val$, $val'$) is called, $val$ and $val'$ are not equal to $00...0$ or $11...1$.
By Lemma \ref{update-no-child-cas-lem}, \ref{successful-general-mov-lem} and \ref{successful-special-mov-lem}, no child CAS of $I$ that $I$ is created by a replace operation removes $00...0$ or $11...1$ from $activeValues$.
So, $00...0$ and $11...1$ are in the $activeValues$ set at any time.
Thus, two leaf nodes whose $label$s are $00...0$ and $11...1$ are logically in the trie at all times.
\end{proof}

The next lemma shows that all nodes in $F_I - U_I$ are unreachable after a child CAS step of $I$ on each element of $I.pNode$ succeeds.
\begin{lemma} \label{unreachable-after-child-cas-lem}
Let $I$ be a Flag object. 
If the child CAS of $I$ on each element of $I.pNode$ is performed by time $T$, all nodes in $F_I - U_I$ are unreachable at all times after $T$.
\end{lemma}
\begin{proof}
Since the first child CAS of $I$ on each element of $I.pNode$ is performed by time $T$, a child CAS of $I$ on each element of $I.pNode$ succeeds by time $T$ (by Lemma \ref{first-child-cas-if-no-aba-lem}).
If $I$ is created at line \ref{ins-set-info1}, \ref{ins-set-info2}, \ref{del-set-info}, \ref{mov-set-info1}, \ref{mov-set-info2}, \ref{mov-set-info-sc1} or \ref{mov-set-info-sc23}, 
every node in $F_I - U_I$ is also in $I.oldChild$.
So, if $I$ is not created at line \ref{mov-set-info-sc4}, by Lemma \ref{unreachable-oldChild-lem} and \ref{no-reachable-after-unreachable-lem}, all nodes in $F_I - U_I$ are unreachable at all times after $T$.

Now, we show if $I$ is created at line \ref{mov-set-info-sc4}, all nodes in $U_I - F_I$ are unreachable at all times after $T$.
Let $\langle gp_d$, $p_d$, $node_d$, -, -, -$\rangle$ be the result returned by the call to search($val_d$) on line \ref{mov-call-search1} that precedes the creation of $I$ and 
$\langle$ -, $p_i$, $node_i$, -, -, -$\rangle$ be the result returned by the call to search($val_i$) on line \ref{mov-call-search2} that precedes the creation of $I$.
Since $I$ is created at line \ref{mov-set-info-sc4}, $F_I - U_I = \{ gp_d, p_d \}$.
Since $I.oldChild[0] = gp_d$, $gp_d$ is unreachable at all times after $T$ (by Lemma \ref{unreachable-oldChild-lem} and \ref{no-reachable-after-unreachable-lem}).
We show that $p_d$ becomes unreachable just after the first child CAS of $I$ (which is before $T$).

Since $p_i \in F_I $, $p_i.info = I$ just before the first child CAS of $I$ (by Corollary \ref{flag-node-before-child-cas-col}) and $p_i$ is not marked just before the first child CAS of $I$.
So, $p_i$ is reachable just before the first child CAS of $I$ (by Corollary \ref{unmarked-reachable-col}).
By Lemma \ref{query-reachable-lem}, $p_i$ is a parent of $node_i = gp_d$ at some time during search($val_i$) and $gp_d$ is a parent of $p_d$ at some time during search($val_d$).
So, $p_i$ is a parent of $node_i = gp_d$ and $gp_d$ is a parent of $p_d$ just before the first child CAS of $I$  (by Lemma \ref{no-wrong-child-change-lem}). 
Thus, $gp_d$ is the only reachable parent of $p_d$ just before the first child CAS of $I$ (by Lemma \ref{reachable-parent-lem}).
Since $gp_d$ becomes unreachable just after the first child CAS of $I$, $p_d$ also becomes unreachable just after the first child CAS of $I$.
By Lemma \ref{no-reachable-after-unreachable-lem}, $gp_d$ is unreachable at all times after $T$.
\end{proof}

\begin{lemma} \label{help-terminate-lem}
Let $I$ be a Flag object. 
After some call to help($I$) terminates, the $info$ field of no reachable internal node is $I$.
\end{lemma}
\begin{proof}
Let $h$ be a call to help($I$) that terminates.  
Now, we show that no flag CAS of $I$ succeeds after $h$ executes line \ref{help-check-doChildCAS}.
First, $h$ tries to flag at least one node in $F_I$ on line \ref{help-flag-node} (since $doChildCAS$ is initially true).
Let $j \ge 1$ be the number of flag CAS steps that $h$ performs.
Since $h$ attempts to flag $I.flag[1]$ to $I.flag[j]$ before executing line \ref{help-check-doChildCAS},
and only the first flag CAS of $I$ on each node can succeed (by Lemma \ref{first-info-cas-lem}), none of these nodes become flagged after $h$ executes line \ref{help-check-doChildCAS}.


If $x$ is not the last node in $I.flag$, $x.info \ne I$ on line \ref{help-set-doChildCAS} (since $x$ is the last node in $I.flag$ that $h$ tries to flag).
Since the first flag CAS of $I$ on $x$ is performed before $x.info \ne I$ on line \ref{help-set-doChildCAS}, $x.info \ne I$ at all times after that (by Lemma \ref{first-info-cas-lem}).
So, after $h$ executes line \ref{help-check-doChildCAS}, no other call to help($I$) attempt to flag any node that is in $I.flag$ after $x$.
Thus, after $h$ executes line \ref{help-check-doChildCAS}, no flag CAS of $I$ succeeds.
Now, we consider two cases according to the value of $I.flagDone$ at line \ref{help-check-done}.

Case 1: $I.flagDone$ is false when $h$ executes line \ref{help-check-done}.
Then, $h$ executes a backtrack CAS of $I$ on each element of $F_I$ at line \ref{help-unflag-bt}.
Let $x$ be a node in $F_I$.
If $x.info = I$ when $h$ executes line \ref{help-unflag-bt}, $x.info$ is changed to an Unflag object when $h$ executes line \ref{help-unflag-bt} and 
$x.info$ is never set back to $I$ after that.
If $x.info \ne I$ when $h$ executes line \ref{help-unflag-bt}, $x.info$ is not set to $I$ at any time after that (since no flag CAS of $I$ succeeds after that). 
So, the $info$ field of no node in $F_I$ is $I$ after $h$ terminates.

Case 2: $I.flagDone$ is true when $h$ executes line \ref{help-check-done}.
First, we show that a child CAS step of $I$ on each element of $I.pNode$ is performed before $h$ executes line \ref{help-check-done}. 
If $doChildCAS$ is true, then $h$ performs these child CAS steps itself.
Otherwise , $h$ set $doChildCAS$ to false at line \ref{help-set-doChildCAS} when the $info$ field of node $x$ is not $I$ on line \ref{help-set-doChildCAS}.
Since $I.flagDone$ is true when $h$ executes line \ref{help-check-done}, 
the $info$ field of $x$ was set to $I$ before $h$ executes line \ref{help-check-done} (by Lemma \ref{set-done-after-flag-lem}).
Since $x.info \ne I$ when $h$ executes line \ref{help-set-doChildCAS} after $h$ tries to set $x.info$ to $I$, 
$x.info$ was set to $I$ and changed from $I$ to some other value before $h$ executes line \ref{help-set-doChildCAS} (by Lemma \ref{first-info-cas-lem}).
So, an unflag or backtrack CAS step of $I$ on $x$ succeeded before $h$ reads $x.info$ at line \ref{help-set-doChildCAS}.

Let $h'$ be the call to help($I$) that executes the first unflag or backtrack CAS of $I$.
Since no unflag or backtrack CAS of $I$ succeeds before $h'$ executes line \ref{help-check-done} and the $info$ fields of all nodes in $F_I$ are set to $I$ before $h$ executes line \ref{help-check-done},
after $h'$ tries to flag nodes in $F_I$ at line \ref{help-flag-node}, $h'$ sets $doChildCAS$ to true at each execution of line \ref{help-set-doChildCAS}.
Since $doChildCAS$ is true when $h'$ executes line \ref{help-check-doChildCAS}, $h'$ executes a child CAS of $I$ on each element of $I.pNode$ at line \ref{help-change-child} 
before it performs any unflag or backtrack CAS step of $I$. 
So, a child CAS of $I$ on each element of $I.pNode$ at line \ref{help-change-child} is executed before $h$ reads $x.info$ at line \ref{help-set-doChildCAS}.
By Lemma \ref{unreachable-after-child-cas-lem}, all nodes in $F_I - U_I$ are unreachable after $h$ terminates. 

Since $I.flagDone$ is true when $h$ executes line \ref{help-check-done}, $h$ performs unflag CAS of $I$ on all nodes in $U_I$ at line \ref{help-unflag}.
Let $x$ be a node in $U_I$.
If $x.info = I$ when $h$ executes line \ref{help-unflag}, $x.info$ is changed from $I$ to an Unflag object at line \ref{help-unflag} and 
$x.info$ is never set back to $I$ after that (by Lemma \ref{same-info-lem}).
If $x.info \ne I$ when $h$ executes line \ref{help-unflag}, $x.info$ is not set to $I$ at any time after that (since no flag CAS of $I$ succeeds after that). 
So, the $info$ field of no node in $U_I$ is $I$ after $h$ terminates.
Thus, the $info$ field of no reachable node is $I$ after $h$ terminates.
\end{proof}

The following lemma show that when the update operation reads the node $nodeSibling$ at line \ref{del-read-sibling} or \ref{mov-read-sibling} or the node $pSibling$ at line \ref{mov-read-pSibling}, 
the node was reachable during the update operation at some earlier time.
The proof of the lemma is similar to the proof of Lemma \ref{query-reachable-lem} and we do not present the argument here again.

\begin{lemma} \label{reachable-sibling-lem}
If the update operation reads $nodeSibling$ at line \ref{del-read-sibling} or \ref{mov-read-sibling} during a loop iteration, 
there is a time after the update operation begins the loop iteration and before the update operation reads $nodeSibling$ that $nodeSibling$ is reachable.
\end{lemma}

Finally, we show that the implementation is non-blocking.
\begin{lemma}
The implementation is non-blocking.    
\end{lemma}
\begin{proof}
We prove the lemma by contradiction.
Assume the implementation is not non-blocking.
So, there is an execution of the implementation such that, after time $T$, each pending operation $op$ take infinitely many steps and no operation terminates.

Assume update operations $op_1, op_2, ..., op_n$ each take infinitely many steps after time $T$.
By Lemma \ref{search-wf-lem}, these operations are update operations.

Each call to help terminates, by Lemma \ref{help-terminate-lem}, the $info$ field of no reachable node is $I$ after a call to help($I$) terminates.
So, if $I$ is created before $T$ and help($I$) is called after that, after the call to help($I$) terminates, no other operation sees $I$ in the $info$ field of any reachable node.
Now, we show that if the operation $op$ calls help($I$) during a loop iteration and $I$ is created by another operation, 
then $I$ was read from the $info$ field of a node that was reachable during the loop iteration. 

If $op$ calls help($I$) and $I$ is created by another operation, then help($I$) is called on line \ref{newFlag-call-help} or \ref{createNode-call-help1}.
If help($I$) is called at line \ref{newFlag-call-help}, $I$ is read from the $info$ field some node that the search operation returns earlier.
If help($I$) is called at line \ref{createNode-call-help1} inside the createNode routine, which is called at line \ref{ins-call-create-node} or \ref{mov-call-create-node}, 
$I$ is read from the $info$ field of some node that the search operation returns earlier.
 In either case, the node that the search operation returns was reachable at some time during the search operation by Lemma \ref{query-reachable-lem}.
If help($I$) is called at line \ref{createNode-call-help1} inside the createNode routine, which is called at line \ref{mov-create-node1}, 
$I$ is read from the $info$ field some node that is read at line \ref{mov-read-sibling}.
By Lemma \ref{reachable-sibling-lem}, the node was reachable at some time during the the loop iteration that help($I$) is called.

So, if some $op_i$ starts its loop iteration after the first call to help($I$) returns, $op_i$ cannot call help($I$) during that loop iteration. 
Since the number of Flag objects that are created before $T$ is finite, 
there is a time $T_x$ after $T$ such that no operation calls help($I$) if $I$ is not created by $op_1, op_2, ..., op_n$ after $T''$. 
Let $T''$ ($> T_x$) be a time by which each running operation has started a new loop iteration after $T_x$.
Thus, after $T''$, help($I$) is only called for Info object $I$ that were created by $op_1, op_2, ..., op_n$.

\noindent {\bf Claim:}
Some operation calls the help routine after $T''$.

\noindent {\it Proof of Claim.}
To derive a contradiction, assume no help routine is called after $T''$.
So, there is a time $T'$ ($> T''$), no operation executes any line of the help routine.
Since the $info$ and $child$ fields of nodes are only changed at line \ref{help-flag-node}, \ref{help-flag-leaf}, \ref{help-change-child}, \ref{help-unflag} or \ref{help-unflag-bt} inside the help routine, 
the $info$ and $child$ field of no node is changed after $T'$.
Let $op$ be an update operation that is running after $T'$.

First, we show if $op$ calls newFlag after $T'$, newFlag does not return null.
Suppose $op$ calls newFlag at line \ref{ins-set-info1}, \ref{ins-set-info2}, \ref{del-set-info}, \ref{mov-set-info1}, \ref{mov-set-info2}, \ref{mov-set-info-sc1}, \ref{mov-set-info-sc23} or \ref{mov-set-info-sc4}.
Since no $info$ field is changed after $T'$, no newFlag returns null at line \ref{newFlag-return-null2}.
If newFlag returns null at line \ref{newFlag-return-null1}, the help routine is called at line \ref{newFlag-call-help} just before that.
So, no newFlag returns null after $T'$. 

We consider different cases according to the type of the update operation that $op$ is.
For each case, we show that $op$ calls the help routine.

Case 1: $op$ is insert($val$).
We show if $op$ does not call the help routine at line \ref{ins-call-help} during a loop iteration after $T'$, $op$ calls the help routine at line \ref{createNode-call-help1} during that loop iteration.
Let $\langle$ -, $p$, $node$, -, -, $rmvd \rangle$ be the result returned by the call to search($val$) at line \ref{ins-call-search} during the iteration.
Since $op$ does not call the help routine at line \ref{ins-call-help} and newFlag does not return false at line \ref{ins-set-info1} or \ref{ins-set-info2}, 
createNode returns null when it is called at line \ref{ins-call-create-node}.
Suppose createNode($node_1$, $node_2$, -) is called at line \ref{ins-call-create-node}.
Then, $node_1$ is a new copy of $node$ and $node_2$ is a leaf node whose $label$ is $val$.
So, $node.label$ is a prefix of $val$ or $val$ is a prefix of $node.label$.
If $node$ is an internal node, $node.label$ is not a prefix of $val$ by Lemma \ref{search-lem}.
So, $node$ is a leaf node and $node.label = val$.
Since $op$ does not return false at line \ref{ins-return-false}, $rmvd$ is true.
Since $rmvd$ is set to true at line \ref{search-set-moved}, $node.info$ that is read at line \ref{search-set-moved} is not an Unflag object. 
Let $I'$ be a Flag object that $node.info = I'$ at line \ref{search-set-moved}.
By Lemma \ref{flag-leaf-lem}, $node.info = I'$ at line \ref{ins-read-node-info}.
Since $I'$ is a Flag object, $op$ calls the help routine at line \ref{createNode-call-help1}.

Case 2: $op$ is delete($val$).
We show that $op$ calls the help routine at line \ref{del-call-help} during each loop iteration after $T'$.
Let $\langle gp$, $p$, $node$, -, -, -$\rangle$ be the result returned by the call to search($val$) at line \ref{del-call-search} during the iteration.
First, we show $gp$ is not null.
If $gp$ is null, the search operation exits the loop after the first loop iteration 
(since $gp$ is set to $p$ at line \ref{search-set-gp} during each loop iteration and after the first execution of line \ref{search-set-p}, $p$ is an internal node by Lemma \ref{p-gp-internal-lem}).
Then, $p = root$.
Since $op$ does not return false at line \ref{del-return-false}, $node$ is a leaf node whose $label$ is $val$.
If $node$ is set to $root.child[0]$ at line \ref{search-set-node}, by Lemma \ref{dummy-lem} and Invariant \ref{prefix-inv}, $node.label = 00...0$.
If $node$ is set to $root.child[1]$ at line \ref{search-set-node}, by Lemma \ref{dummy-lem} and Invariant \ref{prefix-inv}, $node.label = 11...1$.
Since $val \ne 00...0$ and $val \ne 11...1$ by the precondition of delete($val$), $gp$ is not null.
Since newFlag does not return null at line \ref{del-set-info}, $op$ calls the help routine at line \ref{del-call-help}.

Case 3: $op$ is replace($val_d$, $val_i$).
Let $\langle gp_d$, $p_d$, $node_d$, -, -, -$\rangle$ be the result returned by the call to search($val_d$) at line \ref{mov-call-search1} during a loop iteration of $op$ that started after $T'$ and 
$\langle$ -, $p_i$, $node_i$, -, -, $rmvd_i \rangle$ be the result returned by the call to search($val_i$) at line \ref{mov-call-search2} during the same loop iteration.
The same argument as in Case 2 establishes that $gp_d$ is not null.
Now, we consider different cases of replace($val_d$, $val_i$).
For each case, we show that $op$ calls the help routine.

Case 3A: $node_i \notin \{ node_d, p_d, gp_d\}$ and $p_i \ne p_d$.
We show if $op$ does not call the help routine at line \ref{mov-call-help}, $op$ calls the help routine at line \ref{createNode-call-help1}.
Since $op$ does not call the help routine at line \ref{mov-call-help} and newFlag does not return null at line \ref{mov-set-info1} or \ref{mov-set-info2}, 
neither line \ref{mov-set-info1} nor \ref{mov-set-info2} is executed.
So, createNode returns null when it is called at line \ref{mov-call-create-node}.
Suppose createNode($node_1$, $node_2$, -) is called at line \ref{mov-call-create-node}.
Then, $node_1$ is a new copy of $node_i$ and $node_2$ is a leaf node whose $label$ is $val_i$.
So, $node_i.label$ is a prefix of $val_i$ or $val_i$ is a prefix of $node_i.label$.
If $node_i$ is an internal node, $node_i.label$ is not a prefix of $val_i$ by Lemma \ref{search-lem}.
So, $node_i$ is a leaf node and $node_i.label = val_i$.
Since $op$ does not return false at line \ref{mov-return-false2}, $rmvd_i$ is true.
Since $rmvd_i$ is set to true at line \ref{search-set-moved}, $node_i.info$ that is read at line \ref{search-set-moved} is not an Unflag object. 
Let $I'$ be a Flag object that $node_i.info = I'$ at line \ref{search-set-moved}.
By Lemma \ref{flag-leaf-lem}, $node_i.info = I'$ at line \ref{ins-read-node-info}.
Since $I'$ is a Flag object, $op$ calls the help routine at line \ref{createNode-call-help1}.

Case 3B: $node_i = node_d$.
Since  newFlag does not return null after $T'$, $op$ sets $movInfo$ to a Flag object at line \ref{mov-set-info-sc1} and calls the help routine at line \ref{mov-call-help}.

Case 3C: $node_i \ne node_d$ and either $node_i = p_d$ or $p_i = p_d$.
We show if $node_i = p_d$, $p_i = gp_d$.
By Lemma \ref{search-lem}, $gp_d$ was a parent of $p_d = node_i$ at some time during search($val_d$) and $p_i$ was a parent of $node_i = gp_d$ at some time during search($val_i$).
Since no $child$ field is changed after $T'$, $gp_d$ and $p_i$ are parents of $p_d = node_i$ at all times after $T'$.
By Lemma \ref{query-reachable-lem}, $gp_d$ was reachable at some time during search($val_d$) and $p_i$ was reachable at some time during search($val_i$).
Since no $child$ field is changed after $T'$, $gp_d$ and $p_i$ are reachable at all times after $T'$.
So, by Lemma \ref{reachable-parent-lem}, $p_i = gp_d$.

Since either $node_i = p_d$ and $p_i = gp_d$ or $p_i = p_d$ (and $gp_d$ is not null), createNode is called at line \ref{mov-create-node1}.
If createNode does not return null, $op$ calls the help routine at line \ref{mov-call-help} since newFlag does not return null at line \ref{mov-set-info-sc23}.
Now, we show if createNode returns null, the help routine is called at line \ref{createNode-call-help1}.
Let $nodeSibling_d$ be the child of $p_d$ that is read at line \ref{mov-read-sibling}.
Suppose createNode($node_1$, $node_2$, -) is called at line \ref{mov-create-node1}.
Then, $node_1$ is $nodeSibling_d$ and $node_2$ is a leaf node whose $label$ is $val_i$.
So, $nodeSibling_d.label$ is a prefix of $val_i$ or $val_i$ is a prefix of $nodeSibling_d.label$.
By Lemma \ref{search-lem}, $p_d.child[i] = node_d$ for some $i$ at some time during search($val_d$).

If $node_i = p_d$,  since $node_i$ is an internal node (by Lemma \ref{p-gp-internal-lem}), $node_i.label$ is not a prefix of $val_i$ (by Lemma \ref{search-lem}).
Since $p_d.label$ is a prefix of $nodeSibling_d$ (by Invariant \ref{prefix-inv}), $nodeSibling_d.label$ is not a prefix of $val_i$.
Thus, createNode does not return null.

If $p_i = p_d$, since $node_i \ne node_d$ and no $child$ field is changed after $T'$, $node_i$ and $node_d$ are children of $p_i = p_d$ at all times after $T'$.
Since $nodeSibling_d$ is also a child of $p_d$ at all times after $T'$ and $nodeSibling_d.label \ne node_d.label$ (by Invariant \ref{prefix-inv}), $nodeSibling_d = node_i$. 
If $nodeSibling_d$ is an internal node, since $nodeSibling_d = node_i$, $nodeSibling_d.label$ is not a prefix of $val_i$ (by Lemma \ref{search-lem}).
Then, createNode does not return null.
If $nodeSibling_d$ is a leaf node and $node_i.label \ne val_i$, createNode does not return null.
If $nodeSibling_d$ is a leaf node and $node_i.label = val_i$, since $op$ does not return false at line \ref{mov-return-false2}, $rmvd_i$ is true.
Since $rmvd_i$ is set to true at line \ref{search-set-moved}, $node_i.info$ is not an Unflag object when it is read at line \ref{search-set-moved} . 
Let $I'$ be a Flag object that $node_i.info = I'$ at line \ref{search-set-moved}.
By Lemma \ref{flag-leaf-lem}, since $node_i=nodeSibling_d$, $nodeSibling_d.info = I'$ at line \ref{mov-create-node1}.
Since $I'$ is a Flag object, $op$ calls the help routine at line \ref{createNode-call-help1}.

Case 3D: Otherwise.
First, we show that the condition at line \ref{mov-sc5} is true and $node_i = gp_d$.
Since the condition at line \label{mov-check-gc-condition} is not true and , $node_i$ is $node_d$, $p_d$ or $gp_d$.
Since the conditions at line \ref{mov-sc3} and \ref{mov-sc4} are not true, then $node_i \ne node_d$ and $node_i \ne p_d$.
So, $node_i = gp_d$ and the condition at line \ref{mov-sc5} is true.
Let $nodeSibling_d$ be the child of $p_d$ that is read at line \ref{mov-read-sibling} and $pSibling_d$ be the child of $gp_d$ that is read at line \ref{mov-read-pSibling}.
First, we show that createNode($nodeSibling_d$, $pSibling_d$, -) that is called at line \ref{mov-create-node2} does not return null.
By Lemma \ref{search-lem}, $gp_d.child[i] = p_d$ at some time during search($val_d$) for some $i$.
By Invariant \ref{prefix-inv}, $(gp_d.label) \cdot i$ is a prefix of $p_d.label$.
Since $nodeSibling_d$ was a child of $p_d$ at line \ref{mov-read-sibling}, $(gp_d.label) \cdot i$ is a prefix of $nodeSibling_d.label$ (by Invariant \ref{prefix-inv}).
By Lemma \ref{search-lem}, $p_d.label$ is a prefix of $val_d$.
So,$(gp_d.label) \cdot i$ is a prefix of $val_d$.
Since $pSibling_d.label$ is read as a child of $gp_d$ at line \ref{mov-read-pSibling}, $(gp_d.label) \cdot i$ is not a prefix of $pSibling_d.label$ (by Invariant \ref{prefix-inv}).
So, createNode($nodeSibling_d$, $pSibling_d$, -) that is called at line \ref{mov-create-node2} does not return null.

Since createNode that is called at line \ref{mov-create-node2} does not return null, createNode is also called at line \ref{mov-create-node3}.
Since $gp_d = node_i$, $node_i$ is an internal node by Lemma \ref{p-gp-internal-lem} and $gp_d.label$ is not a prefix of $val_i$ (by Lemma \ref{search-lem}).
Since $gp_d.child[i] = p_d$ at some time during search($val_d$) and $p_d$ is a parent of $nodeSibling_d$ at line \ref{mov-read-sibling}, $gp_d.label$ is a prefix of $nodeSibling_d$ (by Invariant \ref{prefix-inv}).
Since $gp_d$ is a parent of $pSibling_d$ at line \ref{mov-read-pSibling}, $gp_d.label$ is a prefix of $pSibling_d$ (by Invariant \ref{prefix-inv}).
Let $newChild_i$ be the internal node that is created at line \ref{mov-create-node2}.
Since the children of $newChild_i$ are initially $nodeSibling_d$ and $pSibling_d$, $gp_d.label$ is a prefix of $newChild_i.label$ (by Invariant \ref{prefix-inv}).
Since $gp_d.label$ is not a prefix of $val_i$, $newChild_i.label$ is not a prefix of $val_i$.
So, createNode does not return null when it is called at line \ref{mov-create-node3}. 
Since the createNode routines that are called at line \ref{mov-create-node2} and \ref{mov-create-node3} do not return null, $op$ calls the newFlag routine at line \ref{mov-set-info-sc4}.
Since newFlag does not return null after $T'$, $op$ calls the help routine at line \ref{mov-call-help}.
This complete the proof of the claim.

Thus, after $T'$, some running operation calls the help routine.
Let $I$ be a Flag object that is created by some $op_i$ where $1 \le i \le n$.
After $op_i$ creates $I$, it calls help($I$) at line \ref{ins-call-help}, \ref{del-call-help} or \ref{mov-call-help}. 
Any other operation that visits any node whose $info$ field is $I$ might also call help($I$) at line \ref{newFlag-call-help} or \ref{createNode-call-help1}.
So, there is a group of calls to help($I$) after $T'$.

If a child CAS step of $I$ is executed, the operation that created $I$ returns true (by Lemma \ref{update-return-true-lem}).
Since $I$ is created by one of $op_1, op_2, ..., op_n$ and $op_1, op_2, ..., op_n$ do not return after $T'$, no call to help($I$) executes a child CAS of $I$ after~$T'$.

So, no call to help($I$) sets $I.flagDone$ to true at line \ref{help-set-done}.
Thus, all calls to help($I$) set $doChildCAS$ to false at line \ref{help-set-doChildCAS}, so the $info$ field of some node in $F_I$ was not $I$ at line \ref{help-set-doChildCAS}.
Let $node$ be the first such node that causes any call to help($I$) to set $doChildCAS$ to false. 
We say that a group of calls to help($I$) blames $node$.

Since no child CAS is performed after $T'$, the set of nodes logically in the trie does not change.
Let $node_m$ be the node in the trie whose $label$ is the greatest among the $label$s of nodes that are blamed in this way, according to the total ordering that is used to sort on line \ref{newFlag-sort-nodes}.
Let $I_m$ be a Flag object such that the group of calls to help($I_m$) blames $node_m$.

First, we show that $node_m.info$ is not set to $I_m$ by any flag CAS of $I_m$.
Consider the first time that a call to help($I_m$) sets $doChildCAS$ to false because $node_m.info \ne I_m$ on line \ref{help-set-doChildCAS} at time $T_m$.
Since $node_m$ is the first node that a call to help($I_m$) did not set $node_m.info$ to $I_m$, $node_m.info$ is not set to $I_m$ before $T_m$.
Since the first Flag CAS of $I_m$ on $node_m$ occurs before $T_m$, no flag CAS of $I_m$ changes $node_m.info$ to $I_m$ after $T_m$ (by Lemma \ref{first-info-cas-lem}).

If the old value of $node_m.info$ that is passed to the newFlag routine is not an Unflag object, the newFlag returns null at line \ref{newFlag-return-null1}.
Since $I_m$ is created inside the newFlag routine, the old value of $node_m.info$ that is passed to the newFlag routine that created $I_m$ was an Unflag object. 
Since no call to help($I_m$) flags $node_m$ successfully, $node_m.info$ is changed from the Unflag object before any flag CAS of $I_m$ on $node_m$.
By Lemma \ref{info-lem}, $node_m.info$ is set to a Flag object $I_j \ne I_m$ before any flag CAS of $I_m$ on $node_m$.
Let $g_j$ be the group of calls to help($I_j$).
One such call changes $node_m.info$ to $I_j$.
Let $node_j$ be the node that $g_j$ blames.
Then, $node_j.label < node_m.label$ by definition of $node_m$.
(Since the newFlag routine keeps only one copy of duplicates in $I_j.flagNode$ at line \ref{newFlag-keep-one-copy}, $node_j.label \ne node_m.label$.)
However, $node_m$ appears earlier in $I_j.flag$ array than $node_j$, 
contradicting the fact that nodes in $I_j.flag$ are flagged in order of their $label$s  and $node_j.info$ must set to $I_j$ before $node_m.info$ is set to $I_j$ (by Lemma~\ref{order-flag-lem}).

Therefore, no operation takes infinitely many steps after $T$ and the implementation is non-blocking.
\end{proof}

\end{document}